\documentclass{article}
\usepackage[utf8]{inputenc}

\usepackage{pdflscape}
\usepackage{array}
\usepackage[onehalfspacing]{setspace}
\usepackage{titlesec}
\usepackage{cmbright}
\usepackage{amsthm,amsmath,bm,bbm}
\usepackage{amssymb,mathtools}
\usepackage{pifont}
\newcommand{\cmark}{\ding{51}}%
\newcommand{\xmark}{\ding{55}}
\usepackage{fullpage}
\usepackage[dvipsnames]{xcolor}
\usepackage[colorlinks=true, allcolors=violet]{hyperref}
\usepackage[margin=1in]{geometry}

\usepackage{multirow}
\usepackage{caption}
\usepackage{booktabs}
\usepackage{enumitem}

\usepackage{natbib}
\setcitestyle{authoryear, open={(},close={)}}

\usepackage{minitoc}

\newtheorem{theorem}{Theorem}
\newtheorem{definition}{Definition}[section]

\newtheorem{assumption}{Assumption}
\newtheorem{lemma}{Lemma}[section]

\makeatletter
\renewcommand{\paragraph}{%
  \@startsection{paragraph}{4}%
  {\z@}{1ex \@plus 1ex \@minus .2ex}{-1em}%
  {\normalfont\normalsize\bfseries}%
}
\makeatother

\usepackage{tikz}
\usetikzlibrary{shapes,shadows,arrows,positioning,arrows.meta,decorations.pathreplacing,decorations.pathmorphing,decorations.shapes}

\usepackage{cancel}
\usepackage[colorinlistoftodos,textsize=scriptsize,textwidth=.8in,color=blue!20]{todonotes}

\def\independent{\perp\!\!\!\perp}

\colorlet{no}{cyan!200}
\colorlet{ok}{black!70}
\colorlet{rmod}{gray}
\colorlet{butx}{blue}
\colorlet{uhoh}{red}
\colorlet{half}{blue!50!red}

\allowdisplaybreaks

\newlength\tindent
\def\odds{\mathrm{odds}}
\def\E{\mathrm{E}}
\def\P{\mathrm{P}}
\def\I{\mathrm{I}}

\def\II{\text{II}}
\def\III{\text{III}}
\def\IV{\text{IV}}
\def\std{\mathrm{std}}
\def\hatomega{\hat{\omega}}

\def\hateta{\hat{\eta}}

\DeclareFontFamily{U}{mathb}{}
\DeclareFontShape{U}{mathb}{m}{n}{<-5.5> mathb5 <5.5-6.5> mathb6 
<6.5-7.5> mathb7 <7.5-8.5> mathb8 <8.5-9.5> mathb9 <9.5-11> mathb10 
<11-> mathb12}{}
\DeclareSymbolFont{mathb}{U}{mathb}{m}{n}
\DeclareFontSubstitution{U}{mathb}{m}{n}

\DeclareMathSymbol{\blackdiamond}{\mathbin}{mathb}{"0C} 
\usepackage{amssymb}

\title{Estimation Strategies for Causal Decomposition Analysis with Allowability Specifications}

\author{John W. Jackson, Ting-Hsuan Chang, Aster Meche, Trang Quynh Nguyen}
\titlespacing*{\section} {0pt}{0ex}{0ex}
\titlespacing*{\subsection} {0pt}{0ex}{0ex}
\expandafter\def\expandafter\normalsize\expandafter{%
    \normalsize%
    \setlength\abovedisplayskip{2pt}%
    \setlength\belowdisplayskip{2pt}%
    \setlength\abovedisplayshortskip{2pt}%
    \setlength\belowdisplayshortskip{2pt}%
}
\begin{document}

\maketitle

\begin{abstract}
Causal decomposition analysis (CDA) is an approach for modeling the impact of hypothetical interventions to reduce disparities. It is useful for identifying foci that future interventions, including multilevel and multimodal interventions, could focus on to reduce disparities. Based within the potential outcomes framework, CDA has a causal interpretation when the identifying assumptions are met. CDA also allows an analyst to consider which covariates are allowable (i.e., fair) for defining the disparity in the outcome and in the point of intervention, so that its interpretation is also meaningful. While the incorporation of causal inference and allowability promotes robustness, transparency, and dialogue in disparities research, it can lead to challenges in estimation such as the need to correctly model densities. Also, how CDA differs from commonly used statistical decomposition estimators from the econometrics literature may not be clear, which may limit its uptake. To address these challenges, we provide a tour of estimation strategies for CDA, reviewing existing proposals and introducing novel estimators that overcome key estimation challenges. Among them we introduce what we call ``bridging" estimators that avoid modeling any density, and sequential weighted regression estimators that are multiply robust. Additionally, we provide diagnostics to assess the quality of the nuisance density models and weighting functions they rely on. We formally establish the estimators' robustness to model mis-specification, demonstrate their performance through a simulation study based on real data, and apply them to study disparities in uncontrolled hypertension using electronic health records in a large healthcare system.
\end{abstract}

\section{Introduction}

Causal decomposition analysis (CDA) is an approach to understand how an observed outcome disparity would be reduced, and how much would remain, upon hypothetically eliminating the disparity in a point of intervention, a factor hypothesized to causally affect the outcome. CDA can identify factors that interventions aiming to reduce outcome disparities could intervene upon. A key feature of CDA, as developed in \citet{jackson2021meaningful}, is to explicitly specify which covariate differences are allowable (i.e., fair) for defining disparity in the outcome and in the point of intervention so that the estimands are meaningful and transparent to stakeholders \citep{chang2024importance}. For example, measures of disparity in healthcare typically account for differences in clinical need, which are often seen as allowable (i.e., not leading to disparity) \citep{duan2008disparities,cook2012measuring}.

Building from \citet{vanderweele2014causal} and \citet{jackson2018decomposition}, CDA was introduced in epidemiology \citep{jackson2021meaningful} and has been extended across the social sciences (see \citet{qin2025review} for a summary). However, recent innovations sometimes overlook the issue of allowability, and many estimators, including multiply robust estimators, require correctly modeling the distribution of the variable designated as the point of intervention, limiting their robustness properties. The new estimators require many nuisance components, but few practical diagnostics are available to evaluate their quality. Further, their relevance and connection to traditional statistical decomposition techniques used in the applied research community is unclear, potentially limiting their uptake. Our review and novel contributions aim to overcome these limitations.

\subsection{Contributions \& Related Literature} \label{sec:Contributions}

We tour estimation strategies for causal decomposition analysis (CDA) that can explicitly incorporate assumptions about which covariates reflect fair sources of outcome differences (outcome-allowability) and which covariates are appropriate for guiding interventions (intervention-allowability). Classifying estimators by whether they do or do not model densities for the point-of-intervention variable or covariates, we further organize estimators by whether they focus on re-weighting, outcome modeling, or both. As part of our tour, we generalize the proposals of \citet{jackson2018decomposition}, \citet{sudharsanan2021educational}, and \citet{park2024estimation} to incorporate allowability. Certain estimators generalize widely used statistical decomposition estimators \citep{oaxaca1973male,blinder1973wage,fairlie2005extension} to incorporate assumptions about causality and allowability. Our tour updates the review of \citet{park2024choosing} that focused on continuous outcomes and ignored allowability, and of \citet{qin2025review} that covered allowability but did not examine estimation details.

Our work also develops new estimators that overcome key challenges in CDA, including the need to model densities and the need for robust estimation. We introduce novel ``bridging" estimators that avoid posing models for densities and instead rely on constructed samples to facilitate the use of empirical densities or the re-expression of density ratios. For both modeling and bridging strategies, we introduce novel sequential weighted regression estimators \citep{isenberg2024marshall,gabriel2024inverse} that are multiply robust, remaining consistent when certain nuisance components are incorrect. We prove their robustness properties, demonstrate their consistency and robustness through simulation, and show their relation to the influence function, which we derive in the Supplement. They complement recent estimators \citep {lundberg2024gap,yu2025nonparametric,park2025causal} based on Augmented Inverse Probability Weighting \citep{robins1994estimation,glynn2010introduction} that accommodate machine learning tools but, relying on models for densities, are limited to discrete points of intervention. We limit our review of multiply robust estimators to sequential weighted regression estimators because they are intuitive and readily accommodate discrete, continuous, and multiple points of intervention.

Finally, we provide tools to improve the implementation of CDA estimators. Many estimators, even some that are multiply robust, require a correctly specified model for the density of the point of intervention or covariates, limiting their robustness. We therefore propose diagnostics that evaluate the quality of these density models. Furthermore, many estimators make use of weighting functions that, unlike other estimands in causal inference \citep{cole2008constructing,austin2015moving,jackson2016diagnostics,jackson2019diagnosing}, have few relevant metrics for covariate balance. \citet{ben2024estimating} have considered diagnostics for how well weights balance covariates across social groups. We propose two sets of diagnostics to assess the weights. The first examines the weights' mean within an appropriate subsample. The other, based on the idea of target balance \citep{chattopadhyay2020balancing}, examines how well the weighted covariate distributions, including the point of intervention, mimic their target distributions. Our diagnostics for densities and weighting functions support a design-based estimation process where analysts can evaluate and address the potential quality of their estimator before obtaining the results.

\subsection{Motivating Example \& Notation}\label{sec:Setup}

Suppose we observe that, in a population of patients with hypertension, the level of uncontrolled hypertension at follow-up $Y$ (systolic blood pressure at or above 140 mm Hg or diastolic blood pressure at or above 90 mm Hg) is higher among a historically disadvantaged social group $G=1$ as compared to a historically advantaged social group $G=0$. We want to know how the disparity in $Y$ would change if we eliminate the disparity in a determinant of $Y$, such as antihypertensive treatment intensification, which is defined as starting, increasing the dose, or adding a new class of antihypertensive medication. We consider this treatment intensification variable our point of intervention $Z$. Let $A_y$ denote the variables we designate as allowable for measuring disparity in our outcome $Y$ (i.e., outcome-allowable covariates) such as age and sex. Let $A_z$ denote the covariates we designate as allowable for defining our intervention on $Z$ (i.e., intervention-allowable covariates) such as baseline blood pressure. Let $N$ denote non-allowable covariates that may be unfair causes of the outcome or their correlates (e.g., neighborhood disadvantage) that we use to identify causal effects but not to measure the disparity in $Y$ or to influence how we intervene on $Z$. Let $Y(z)$ represent the potential outcome had $Z$ been set to $z$. 

To simplify notation, we use subscripts to denote the population over which a probability statement or expectation is obtained, e.g., $\P_g(\cdot)\coloneq\P(\cdot|G=g)$ and $\E_g[\cdot]\coloneq\E[\cdot|G=g]$. When we denote weighting functions $\omega$ used in a weighted average, the superscript and subscript denote how the weight shifts a sample's covariate distribution. For example, the weighting function $\omega_{g^{\prime}g^{\prime\prime}}^{V}(\cdot)$ indicates that it shifts the observed conditional density of the random variable $V$ among the group $G=g^{\prime}$ to a conditional density of $V$ defined by group $G=g^{\prime\prime}$. We use $\P_n[X]$ to refer to the sample average of $X$, i.e., $\frac{1}{n}\textstyle \sum_i^n X_i$ where in this usage $n$ represents the sample size (elsewhere, it represents the instantiation of the random variable $N$) and $i$ indexes the individual. We use the term `group' to refer to a subgroup of the population. Finally, for readability, we sometimes abuse notation by abbreviating a function $f(V)$ as $f(\cdot)$ or simply as $f$.

\subsection{Outline of the Paper}

Following an overview of statistical and causal decomposition analysis in section \ref{sec:SDA_CDA}, we describe existing estimators and introduce novel estimators that overcome key estimation challenges in sections \ref{sec:Estimators_Observation_Arm} and \ref{sec:Estimators_Intervention_Arm}. We provide diagnostics to evaluate the density models and weighting functions used to implement these estimators in section \ref{sec:Diagnostics}. We assess the consistency and robustness of the reviewed estimators in a simulation study based on real data and the clinical literature in section \ref{sec:Simulation_Study}, and apply them to study how eliminating the racial disparity in treatment intensification would change the subsequent disparity in uncontrolled hypertension within a large healthcare system in section \ref{sec:DataApplication}. We conclude by identifying favorable estimators for application and future development based on our results. Extended proofs and derivations appear in the Supplement.

\section{Statistical versus Causal Decomposition Analysis}\label{sec:SDA_CDA}

\subsection{Statistical Decomposition Analysis}\label{sec:SDA}
Statistical decomposition methods (see \citet{fortin2011decomposition} for a review) consider how the joint distribution of covariates differs across groups, and quantify how much group-differences in the joint distribution of covariates associates with group-differences in outcomes. A ``detailed" decomposition breaks this association down into pieces that are statistically attributable to each covariate's differential distribution. In our example, the detailed decomposition would isolate the association of $Z$, the single variable of interest. This framework does not consider interventions, so $Z$ here is not a point of intervention but just a covariate of interest, and our use of $Z$ here (an abuse of notation) is simply to facilitate comparison with what is to come in the presentation of causal decomposition.

The detailed decomposition of \citet{oaxaca1973male} and \citet{blinder1973wage} assumes a linear outcome model:
\begin{align*}
	\E_1[Y|Z,X]=\beta_0+\beta_z Z+\beta_x X 
\end{align*}
Based on this model, the outcome mean under a shift of $Z$ in the $G=1$ group to resemble the marginal distribution of $Z$ in the $G=0$ group is	$\E_1[Y]+\beta_z \bigl(\E_0[Z]-\E_1[Z]\bigr)$. The second term involving $\beta_z$ quantifies how much group differences in $Z$ associate with group differences in $Y$.

Generalizing to non-linear outcomes and interactions, the detailed decomposition of \citet{fairlie2005extension} assumes an outcome regression function with a possibly non-linear link $m(\cdot)$, for example:
\begin{align*}
	m(\E_1[Y|Z,X])=\beta_0+\beta_z Z+\beta_x X+\beta_{z,x} ZX 
\end{align*}
The outcome mean under a shift of $Z$ in the $G=1$ group to resemble the marginal distribution of $Z$ in the $G=0$ group is $\E_\bullet[\hat{\E}_1[Y|Z,X]]$. The expectation is over the model's predicted values in the \emph{artificial} group $\bullet$ created by replacing the $Z$ values of the $G=1$ group with those of persons matched randomly (i.e., ignoring $X$) from the $G=0$ group. Contrasting the mean of the predicted values' in the artificial group with the mean of $Y$ in the $G=1$ group quantifies how much group differences in $Z$ associate with group differences in $Y$.

As discussed in \citet{jackson2018decomposition}, these decompositions can have challenging causal interpretations. Typically, $Z$ is given equal status with $X$, so that the associations of $X$ are also assessed. Investigators do not choose $X$ to contain confounders of $Z$ and $Y$'s relationship, or to exclude potential effects of $Z$, which degrade the causal interpretation for $Z$. As discussed in \citet{jackson2021meaningful}, detailed decompositions can have challenging substantive interpretations. These estimators analyze how marginal differences in $Z$ associate with marginal differences in $Y$. Depending on the setting, marginal differences may not map to what stakeholders consider to be unfair or unjust patterns. Causal Decomposition Analysis, introduced in \citet{jackson2021meaningful} and reviewed in \citet{qin2025review}, addresses these limitations by incorporating causal assumptions, which influence variable selection, and by carefully defining the disparity in $Y$ and the removal of disparity in $Z$. 

\subsection{Causal Decomposition Analysis} \label{sec:CDA}
\subsubsection{Hypothetical Study Design}\label{subsec:HypotheticalDesign} 
To provide intuition for CDA, we cast it within a design that nests the Target Study framework \citep{jackson2025target} (to measure disparity across social groups) within the Target Trial framework \citep{hernan2016using} (to measure causal effects of interventions). While both frameworks envision hypothetical studies with features of eligibility and follow-up, the hypothetical target study balances allowable covariates across social groups by sampling persons (without assigning who belongs to which social group), whereas the hypothetical target trial balances potential confounders across intervention arms by randomly assigning persons to an intervention plan. \citet{qin2025review} provide an overview of CDA and \citet{sun2025integrated} introduce the integration of these frameworks with broader interventions. 

We want to understand how an intervention that removes the disparity in $Z$ would change the disparity in $Y$. To proceed, we acknowledge that groups may differ on many factors, but the distribution of some covariates (i.e., outcome-allowable covariates, $A_y$), on normative or ethical grounds, may not contribute to disparity in the outcome $Y$. We therefore define disparity as the difference in $Y$ across social groups $G$ that would be observed in a hypothetical study after enrolling them via a stratified sampling plan that balances $A_y$ across these social groups according to a standard population denoted as $T=1$ (i.e., the Target Study component). 
\begin{definition}[Standard Population $T=1$]
The standard population, which may be defined among the study population, is the population whose distribution of $A_y$ is used to determine the stratified sampling plan of the target study component. Formally, we denote membership in the standard population by the deterministic binary function $T\coloneq \I(g\in\mathcal{G},a_y\in\mathcal{A}_y)$, where $\I(\cdot)$ is the indicator function, and $\mathcal{G}$ and $\mathcal{A}_y$ are the set or range of values $g$ and $a_y$ that comprise the standard population. For convenience, we denote the standard population's $A_y$ distribution as $\P_\std(A_y)\equiv\P(A_y|T=1)$, and the expectation over its $A_y$ distribution as $\E_\std[\cdot]\equiv\E[\cdot|T=1]$.\footnote{We will assume a within-study standard population. Under an an external study population, where data from the study population are supplemented with that of the external standard population, we would adapt our abbreviated notation as follows: $\E_g[\cdot]:=\E[\cdot|G=g,T=0]$ and $\P_g[\cdot]:=\P(\cdot|G=g,T=0)$ to emphasize that the study and standard populations are distinct.}
\label{def:Standard_Population}
\end{definition}
Among each group, enrolled individuals are randomized to one of two study arms and then followed for outcomes (the Target Trial component). In the observation arm, no action is taken. In the intervention arm, $Z$ is assigned based on covariates $A_z$ that, on normative grounds, are appropriate for determining how $Z$ is allocated.

\subsubsection{The Intervention}\label{subsec:Intervention}
The intervention can be defined in many ways (see, e.g., \citet{qin2025review}) but often it may suffice to define the intervention as in \citet{jackson2021meaningful}. Among the historically advantaged group  we do not intervene on $Z$ so it takes its natural value. But among the historically disadvantaged group $G=1$, we set $Z$ as a random draw from the distribution of $Z$ among the historically advantaged group $G=0$ conditional on $A_z$, for each level of $A_y$. This intervention to set $Z\sim\P_0(Z|A_z,A_y)$, our focus in this paper, is defined by covariates $A_z$ relevant for defining disparity in $Z$ (if any) and by covariates $A_y$ relevant for defining disparity in $Y$ (if any). \citet{jackson2021meaningful} discuss principles for choosing $A_z$ and $A_y$. This intervention  removes the disparity in $Z$ and renders $Z$ and $N$ independent given $(A_z,A_y)$ among the group $G=1$.

\subsubsection{Estimands and Identification}
Comparing the disparity in the observation and intervention arms will inform how much the hypothetical intervention on $Z$ would remove the disparity in $Y$. Because our study design involving sampling-based enrollment and stratified randomization of an intervention to equalize $Z$ is hypothetical, we rely on assumptions to use non-experimental data to estimate the arm-specific group means and disparities. Although these assumptions rely on non-allowables $N$,\footnote{One only needs to measure as many non-allowables as necessary (if any) to satisfy the identification assumptions. The non-allowable category could also include variables that are not deemed as allowable but are included to increase efficiency.} the disparity in each arm is defined by $A_y$ and the intervention on $Z$ is defined by $(A_z,A_y)$. This allows analysts to respect their assumptions about fairness in the distribution of $Z$ and $Y$.

For the observation arm, we assume sufficient overlap in the $A_y$ distribution for stratified sampling ($A_y$-overlap), and that enrollment of each person $i$ into the trial sample $\Omega$ would not change the conditional outcome distribution (innocuous sampling). Formally,
\begin{assumption}[$A_y$-overlap]
$\P_g(A_y=a_y)>0$ for $g\in\{0,1\}$ and all $a_y$ with $P_\std(A_y=a_y)>0$\label{def:Assumption_Ay_Overlap}
\end{assumption}
\begin{assumption}[Innocuous Sampling]
$\P_{i\in\Omega}(Y|Z,N,A_z,A_y,G)=\P_{i\not\in\Omega}(Y|Z,N,A_z,A_y,G)$\label{def:Assumption_Innocuous_Sampling}
\end{assumption}
Under these assumptions, $\theta_g$ the group-specific expected value of $Y$ in the observation arm is:
\begin{align}
\theta_g=\E_\std\bigl[\E_g[Y|A_y]\bigr]\label{eq:ID_Result_Observation_Arm}
\end{align}
For the $G=1$ group among the intervention arm, where $Z\sim\P_0(Z|A_z,A_y)$, the distribution of the data is:
\begin{align}
    \P_1^*(Y(Z),Z,N,A_z,A_y) = \P_1^*(Y(Z) | Z,N,A_z,A_y)\P_1^*(Z,N|A_z,A_y)\P_1^*(A_z|A_y) \P^*_\std(A_y)\label{eq:Distribution_Counterfactual}
\end{align}
Here and throughout, the superscript $*$ notation indicates the intervention arm of the hypothetical target trial. To identify this distribution, we assume overlap in the conditional distribution of $A_z$ ($A_z$-overlap) to permit the assignment of $Z$ given ($A_z,A_y$) per the intervention; unconfoundedness of $Z$'s effect on $Y$ given $(N,A_z,A_y)$ ($Z$-exchangeability); relevant variation is observed in $Z$ given $(N,A_z,A_y)$ ($Z$-positivity); the same outcomes are obtained regardless of whether $Z$ is observed or assigned by the intervention to have the value $z$ ($Y$-consistency). 
\begin{assumption}[$A_z$-overlap]
$\P_0(A_z=a_z|A_y=a_y)>0$ for all values $a_y$ where $\P_\std(A_y=a_y)>0$ and all values $a_z$ where $\P_1(A_z=a_z|A_y=a_y)>0$.\label{def:Assumption_Az_Overlap}
\end{assumption}
\begin{assumption}[$Z$-exchangeability]
$Y(z)\independent Z|N=n,A_z=a_z,A_y=a_y,G=1$ for all values $a_y$ where $\P_\std(A_y=a_y)>0$, all values $(n,a_z)$ where $\P_1(N=n,A_z=a_z|A_y=a_y)>0$, and all values $z$ where $\P_0(Z=z|A_z=a_z,A_y=a_y)>0$\label{def:Assumption_Z_Exchangeability}
\end{assumption}
\begin{assumption}[$Z$-positivity]
$P_1(Z=z|N,A_z,A_y)>0$ for all values $a_y$ where $\P_\std(A_y=a_y)>0$, all values $(n,a_z)$ where $P_1(N=n,A_z=a_z|A_y=a_y)>0$, and all values $z$ where $\P_0(Z=z|A_z=a_z,A_y=a_y)>0$.\label{def:Assumption_Z_Positivity} 
\end{assumption}
\begin{assumption}[$Y$-consistency]
$Y_i\equiv Y_i(z)$ if $Z_i=z$ for each individual $i$\label{def:Assumption_Z_Consistency}
\end{assumption}
These assumptions, which only need to hold at certain values of the observed data, are somewhat weaker than the standard exchangeability, positivity, and consistency assumptions described in \citet{hernan2006estimating}. 

Under Assumptions \ref{def:Assumption_Ay_Overlap}-\ref{def:Assumption_Z_Consistency}, we can identify the distribution of the data in the intervention arm as:
\begin{align}
    \P_1^*(Y(Z),Z,N,A_z,A_y) = \P_1(Y | Z,N,A_z,A_y)\P_0(Z|A_z,A_y)\P_1(N|A_z,A_y)\P_1(A_z|A_y) \P_\std(A_y)\label{eq:Distribution_Identified}
\end{align}
Consequently, we can express $\theta_1^*$ the expected value of $Y$ among the group $G=1$ in the intervention arm as:
\begin{align}
\theta_1^*=\E_\std\bigl[\E_1^*[\mu_1(Z,N,A_z,A_y)\big|A_y]\bigr].\label{eq:ID_Result_Intervention_Arm}
\end{align}
where $\mu_1(Z,N,A_z,A_y)$, the conditional mean function $\E_1[Y|Z,N,A_z,A_y]$, is averaged over $\P_1^*(Z,N,A_z|A_y)$ the distribution of $(Z,N,A_z|A_y)$ in the intervention arm (see (\ref{eq:Distribution_Counterfactual})). The intervention conditions on $(A_z,A_y)$, leaving their distribution unchanged, implying that $\P_1(A_z|A_y)$ identifies $\P_1^*(A_z|A_y)$ and that $\P_\std(A_y)$ identifies $\P_\std^*(A_y)$. For our intervention, $\P_1^*(Z,N|A_z,A_y)$ is identified as the product of $\P_0(Z|A_z,A_y)$ and $\P_1(N|A_z,A_y)$ which, due to their importance in the flexible estimation strategies that follow, we refer to as critical densities.
\begin{definition}[Critical Density of $Z$, $\lambda_0^Z$]
The critical density of $Z$, denoted as $\lambda_0^Z$, among the $G=1$ group under the intervention arm is equivalent to $P_0(Z|A_z,A_y)$, where $Z$ is independent of $N$ given $(A_z,A_y)$. \label{eq:lambda_0_Z}
\end{definition}
\begin{definition}[Critical Density of $N$, $\lambda_1^N$]
The critical density of $N$, denoted as $\lambda_1^N$, among the $G=1$ group under the intervention arm is equivalent to $P_1(N|A_z,A_y)$, where $N$ is independent of $Z$ given $(A_z,A_y)$. \label{eq:lambda_1_N}
\end{definition}
The critical densities $\lambda_0^Z$ and $\lambda_1^N$ differ from the distributions of $Z$ and $N$ in the observation arm, $\P_1(Z|A_z,A_y)$ and $\P_1(N|Z,A_z,A_y)$. The conditional independence of the intervention foci $Z$ from non-allowable factors $N$ reflects the hypothetical intervention's removal of unethical ways of allocating $Z$. The conditional equalization of $Z$ across $G$ reflects the hypothetical intervention's removal of disparity in the distribution of $Z$. As we will see, the key challenge in estimating $\theta_1^*$ is the expectation over the \emph{unobserved} density $\P_1^*(Z,N|A_z,A_y)$. 

\subsection{Prelude to Estimation}
Different views of the identification result (\ref{eq:ID_Result_Intervention_Arm}) indicate potential strategies for estimating $\theta_1^*$. As a weighted average, (\ref{eq:ID_Result_Intervention_Arm}) suggests a weighting function that takes covariate data from the $G=1$ group and morphs the observed density $\P_1(Z,N|A_z,A_y)$ to that of $G=1$ in the intervention arm $\P_1^*(Z,N|A_z,A_y)$, and morphs the density of $A_y$ to that of the standard population, using the observed outcome $Y$:
\begin{align}
\theta_1^*=\E_1 \Big[Y\times \frac{\P^*_1(Z,N|A_z,A_y)}{\P_1(Z,N|A_z,A_y)}\times\frac{\P_\std(A_y)}{\P_1(A_y)}\Big]\nonumber. 
\end{align}
As a sequential expectation, (\ref{eq:ID_Result_Intervention_Arm}) suggests averaging the conditional mean function $\mu_1(Z,N,A_z,A_y)$ over $\P_1^*(Z,N,A_z|A_y)$, the $(Z,N,A_z|A_y)$ distribution of $G=1$ in the intervention arm, yielding the $\eta_1^*(A_y)$ model:
\begin{align}
 \eta_1^*(A_y) &\coloneq \E_1^*[\mu(Z,N,A_z,A_y)|A_y].\nonumber
\end{align}
Another route to $\eta_1^*$ is to average $\mu_1$ over $\P_1^*(Z,N|A_z,A_y)$ (i.e., the product of $\lambda_0^Z$ and $\lambda_1^N$) and then over $\P_1^*(A_z|A_y)$. We can also average $\mu_1$ over $\lambda_0^Z$ and then over $\P_1^*(N,A_z|A_y)$. Or, we can sequentially average $\mu_1$ over $\lambda_1^N$, $\lambda_0^Z$, and $\P_1^*(A_z|A_y)$. The average of $\eta_1^*$ over the standard population's $A_y$ distribution estimates $\theta_1^*$:
\begin{align}
\theta_1^*=\E_\std[\eta_1^*(A_y)].\nonumber
\end{align}

These are the core estimation strategies used by the flexible estimators that follow.

\section{Estimation Strategies for the Observation Arm} \label{sec:Estimators_Observation_Arm}
Here, we present three ways to estimate $\theta_g$, the group-specific mean outcome in the observation arm, to estimate the disparity in that arm. They resemble estimators for average treatment effects but are adapted for this descriptive endeavor. The first estimator is based on a re-expression of $\theta_g$ as a weighted average:
\begin{align}
	\theta_g=\E_g \Bigg[ \frac{\P_\std(A_y)}{\P_g(A_y)} Y\Bigg]  \label{eq:Estimator_PW},
\end{align}
which suggests the `standardizing' weighting function
\begin{align}
	\omega_{gT}^{A_y}(A_y) &\coloneq \frac{\P_\std(A_y)}{\P_g(A_y)} \nonumber = \frac{\P(T=1|A_y)}{\P(G=g|A_y)} \times \frac{\P(G=g)}{\P(T=1)}\label{eq:Weight_Ay}.
\end{align} 
This weight morphs the distribution of $A_y$ among the group $G=g$ to that of the standard population, denoted by $T=1$ (see Definition \ref{def:Standard_Population}). 
This justifies estimating $\theta_g$ as the following weighted sample average:
\begin{align}
\hat\theta_g=\frac{\P_{gn}[\hat\omega_{gT}^{A_y}Y]}{\P_{gn}[\hat\omega_{gT}^{A_y}]},
\end{align}
where $\P_{gn}$ is the average taken over the $G=g$ sample. This weighting estimator is closely related to the inverse probability weighting-based estimators of \citet{sato2003marginal} and of \citet{hernan2006estimating}. It can be shown that the weights are proportional to the sampling fractions used in the sampling plan of the target study component that balance $A_y$ to the standard population's distribution of $A_y$ \citep{jackson2025target}. To achieve consistent estimation of $\theta_g$, the weight $\omega_{gT}^{A_y}(A_y)$ needs to be estimated consistently.

The second estimator is based on the expression of $\theta_g$ as:
\begin{align}
    \theta_g&= \E_\std[\eta_g(A_y)] \nonumber \\
	&=\E_\std\bigl[\E_g [Y|A_y]\big] \label{eq:Estimator_SR}
\end{align}
This estimator averages $\eta_g(A_y)$, defined as $\E_g[Y|A_y]$, over $\P_\std(A_y)$. We can thus estimate $\theta_g$ by fitting a model for $\eta_g(A_y)$, obtaining predicted values $\hateta_g$, and averaging them among the standard population, i.e., those with $T=1$. This is essentially an outcome regression, g-computation, or g-formula estimator \citep{snowden2011implementation}. To achieve consistent estimation of $\theta_g$, the model for $\eta_g(A_y)$ must be correctly specified.

The third estimator is the same as (\ref{eq:Estimator_SR}) but fits the model for $\eta_g$ to the sample space of the standard population. This entails weighting the $\eta_g$ model by $\hatomega_{gT}^{A_y}$ (\ref{eq:Weight_Ay}). This is essentially the doubly robust weighted regression estimator first proposed by Marshall Joffe \citep{isenberg2024marshall} and later explained by \citet{gabriel2024inverse}. When the model for $\eta_g$ has the mean recovery property, the estimator is consistent for $\theta_g$ if either $\eta_g(A_y)$ is correctly specified or $\omega_{gT}^{A_y}$ is consistently estimated. For intuition, the mean recovery property ensures that $\E_\std[\eta^{'}_g(A_y)]=\E_\std[\eta_g(A_y)]$ when the fitted model $\eta^{'}_g(A_y)$ is incorrectly specified but the weights $\hatomega_{gT}^{A_y}$ are consistently estimated. Likewise, $\E_\std[\eta^{'}_g(A_y)]=\E_\std[\eta_g(A_y)]$ when the fitted model $\eta^{'}_g(A_y)$ is correctly specified, regardless of the weights. This estimator also solves the influence function for $\theta_g$, which is: 
\begin{align}
\varphi_{\theta_g}(O)=  \frac{\I(G=g)}{\P(G=g)}\frac{\P_{\std}(A_y)}{\P_g(A_y)}[Y-\eta_g(A_y)] 
                      + \frac{T}{\P(T=1)}[\eta_g(A_y)-\theta_g],\label{eq:if-theta.g}
\end{align}
where $O=(T,G,A_y,Y)$. The second term appears when the standard distribution $\P_{\std}(A_y)$ is not known but ascertained from the $T=1$ sample. With group-specific estimates of $\theta_g$, the observed disparity is $\tau=\theta_1-\theta_0$.

\section{Estimation Strategies for the Intervention Arm} \label{sec:Estimators_Intervention_Arm}

\subsection{Overview}
Here, we present a range of strategies for estimating $\theta_g^*$, the outcome mean of the $G=1$ group in the intervention arm, used to estimate the disparity in the intervention arm, defined as $\tau^*=\theta_1^*-\theta_0^*$. By definition of the intervention, $\theta_0^*=\theta_0$, so we can estimate $\theta_0^*$ using the approaches described for the observation arm. Hence, we focus on estimation strategies for $\theta_1^*$. These include a linear estimator, of interest for its simplicity, and two classes of flexible estimators that employ what we call density-modeling and density-bridging strategies. The modeling estimators model one critical density (e.g., $\lambda_0^Z$) and average over the empirical version of the other critical density (e.g., $\lambda_1^N$). Remarkably, the bridging strategies avoid modeling either critical density by exploiting the conditional independence of $Z$ and $N$ among $G=1$ in the intervention arm.\footnote{The bridging estimator construction permits averaging over the empirical critical density, or re-expressing a density ratio involving the critical density as a combination of odds functions, even though its numerator and denominator differ in conditioning sets.} 

We classify the flexible estimators according to whether (a) they use a critical density modeling strategy or a bridging strategy and (b) the modeling or bridging strategy revolves around $Z$ or $N$. This leads to four strategies for flexible estimation: \textit{Z-modeling}, \textit{N-modeling}, \textit{Z-bridging}, and \textit{N-bridging}. For each of the four flexible estimation strategies, we present (a) a \textit{Pure Weighting} (PW) estimator that requires no model for the outcome, (b) a \textit{Sequential Regression (SR)} estimator that only requires sequential models for the outcome and pseudo-outcome\footnote{By ``pseudo-outcome", we mean the predicted value from the previous regression step in the sequential regression procedure.} (c) a \textit{Sequential Weighted Regression estimator (SWR)} that weights the sequential regressions of the SR estimator to target the covariate space of the intervention arm, and (d) a \textit{Regress-then-Weight} (RW) estimator that combines weighting and outcome regressions (see Supplement).

The linear estimator, PW, RW, and SR estimators will consistently estimate $\theta_1^*$ when all of their nuisance components are correctly specified or consistently esimated. In contrast, the SWR estimators are consistent for $\theta_1^*$ in $2^J$ ways, where $J$ refers to the estimator's count of outcome and pseudo-outcome regression steps, permitting some nuisance components to be incorrectly specified or inconsistent. The robustness of each outcome regression step follows from the use of a model form with the mean recovery property and weights that target the covariate space of the intervention arm. Under mean recovery, the outcome regression will always recover the mean outcome in the sample space used to fit the model. Thus, when the weights for that regression step are consistently estimated, the intervention arm's outcome mean is obtained regardless of the outcome regression's specification. Conversely, when the outcome regression is correctly specified, the intervention arm's outcome mean is obtained regardless of the weights. The linear estimator and certain PW, RW, and SR modeling estimators generalize existing proposals by accommodating intervention-allowability. To our knowledge, all of the SWR estimators, and the all of the bridging estimators, are novel proposals.

\subsection{E-OBD: Linear Estimator}
The linear estimator is an extension of the detailed Oaxaca-Blinder Decomposition estimator of \citet{oaxaca1973male} and \citet{blinder1973wage} to incorporate allowability for the intervention foci and the outcome (thus named E-OBD). The approach is of interest in that (i) it closely relates to this canonical decomposition estimator that is widely used by applied researchers \citep{altonji1999race,sen2014using,rahimi2021detailed,gastwirth2020role} (ii) it is simple to implement for categorical and continuous $Z$ without any density modeling, sequential regressions, or construction of artificial samples. The estimator assumes a linear causal model for the outcome $Y$ given the point of intervention $Z$, the non-allowables $N$, and the allowables $(A_z,A_y)$ among the $G=1$ group:
\begin{align}
	\E_1[Y|Z,N,A_z,A_y]=\beta_0+\beta_z Z+\beta_n N+\beta_{a_z} A_z + \beta_{a_y}A_y \label{eq:Y_Model_Linear}.
\end{align}
Because the model is specified among the $G=1$ group, the effect of $Z$ may implicitly vary across groups $G$. However, the model cannot reflect any further heterogeneity of $Z$'s effect across levels of $(N,A_z,A_y)$ and it cannot contain any higher-order terms (e.g., quadratic) for $Z$.\footnote{The outcome model (\ref{eq:Y_Model_Linear}) can contain interactions separately among $N$, separately among $A_z$, and separately among $A_y$.} Under the linear causal model (\ref{eq:Y_Model_Linear}), $\theta_1^*$ can be expressed in a form that resembles a detailed OBD decomposition: a (weighted) outcome mean of the $G=1$ group is shifted by the product of $\beta_z$ and a difference in (weighted) intervention foci means between $G=0$ and $G=1$. Both groups are weighted to mimic the $A_y$ standard distribution via the standardizing weighting function $\omega_{gT}^{A_y}(A_y)$ previously defined (\ref{eq:Weight_Ay}), which incorporates outcome-allowability. To incorporate the intervention-allowability of $A_z$, the $G=0$ group is additionally weighted to mimic the $(A_z|A_y)$ distribution of the $G=1$ group, via a `covariate-shifting' weighting function:
\begin{align}
	\omega_{01}^{A_z}(A_z,A_y)	&\coloneq \frac{\P_1(A_z|A_y)}{\P_0(A_z|A_y)} =\frac{\odds(G=1\text{ vs }0|A_z,A_y)}{\odds(G=1\text{ vs }0|A_y)}\label{eq:Weight_Az}.
\end{align}
Formally, 
\begin{align}
    \theta_1^*&=\E_1[\omega_{1T}^{A_y}(A_y)\,Y]+(\theta_1^*-\theta_1) \nonumber
\end{align}   
where 
\begin{align}
\theta_1^*-\theta_1=\beta_z\Big\{\E_0[\omega_{0T}^{A_y}(A_y)\omega_{01}^{A_z}(A_z,A_y)\,Z]-\E_1[\omega_{1T}^{A_y} (A_y)\,Z]\Big\}. \nonumber
\end{align}
The RHS does not change if we divide each of these expectations by the expectations of the weights, as they are all equal to 1. This justifies the E-OBD linear estimator, which uses weighted sample averages:
\begin{align}
    \hat\theta_{1,\text{E-OBD}}^*\coloneq\frac{\P_{1n}[\hat\omega_{1T}^{A_y}\,Y]}{\P_{1n}[\hat\omega_{1T}^{A_y}]}+\hat\beta_z \Bigg( \frac{\P_{0n}[\hat\omega_{0T}^{A_y}\hat\omega_{01}^{A_z}\,Z]}{\P_{0n}[\hat\omega_{0T}^{A_y}\hat\omega_{01}^{A_z}]}-\frac{\P_{1n}[\hat\omega_{1T}^{A_y}\,Z]}{\P_{1n}[\hat\omega_{1T}^{A_y}]} \Bigg) \label{eq:Estimator_EOBD},
\end{align}
where $\P_{1n}$ and $\P_{0n}$ respectively denote averages on the $G=1$ and $G=0$ samples.
For a categorical $Z$ with $j=1,2,\dots,J$ levels, we fit the linear outcome model (\ref{eq:Y_Model_Linear}) with indicators $Z_j$ for each for each non-referent level of $Z$ (i.e., $Z_j\coloneq\textup{I}(Z=j)$), and modify (\ref{eq:Estimator_EOBD}) to sum over the contributions of each non-referent level, i.e.,
\begin{align}
    \hat\theta_{1,\text{E-OBD}}^*\coloneq\frac{\P_{1n}[\hat\omega_{1T}^{A_y}\,Y]}{\P_{1n}[\hat\omega_{1T}^{A_y}]}+ \sum_j \Bigg\{ \hat\beta_z \Bigg(\frac{\P_{0n}[\hat\omega_{0T}^{A_y}\hat\omega_{01}^{A_z}\,Z_j]}{\P_{0n}[\hat\omega_{0T}^{A_y}\hat\omega_{01}^{A_z}]}-\frac{\P_{1n}[\hat\omega_{1T}^{A_y}\,Z_j]}{\P_{1n}[\hat\omega_{1T}^{A_y}]}\Bigg) \Bigg\} \label{eq:Estimator_EOBD_Categorical},
\end{align}
To consistently estimate $\theta_1^*$, this estimator requires (i) correct specification of the linear outcome model (\ref{eq:Y_Model_Linear}), which assumes absence of interactions beyond $Z$ and $G$ and absence of higher order terms for $Z$, and (ii) consistent estimation of the weights $\omega_{gT}^{A_y}$ (\ref{eq:Weight_Ay}) and $\omega_{01}^{A_z}$ (\ref{eq:Weight_Az}). (An alternate linear estimator is provided alongside proof of this estimator in the Supplement.) This estimator generalizes the linear estimator of \citet{jackson2018decomposition} by accommodating intervention-allowability and categorical points of intervention, but condition (i) is very restrictive. The flexible estimators that follow avoid these constraints.

\subsection{Z-Modeling Estimators}
The Z-modeling estimators require a correctly specified model for $\lambda_0^Z$, the critical density of $Z$, and rely on the empirical version of $\lambda_1^N$, the critical density of $N$. In the \textit{Pure Weighting} (PW) and \textit{Regress then Weight} (RW) estimators, $\lambda_0^Z$ appears in a density ratio used to take a weighted average of the outcome or pseudo-outcome. Whereas in the \textit{Sequential Regression} (SR) estimator, an artificial sample $G=\blacktriangle$ is created by simulating $Z$ from $\lambda_0^Z$. In the \textit{Sequential Weighted Regression} (SWR) estimator, $\lambda_0^Z$ plays two roles: simulating $Z$ and appearing in the weight used to fit the initial outcome model.

\subsubsection{Z-Model-PW: Pure Weighting}
Based on its identifying formula in (\ref{eq:ID_Result_Intervention_Arm}), $\theta_1^*$ can be expressed as a weighted outcome mean in the $G=1$ group:
\begin{align}
    \theta_1^*
    &=\E_1\left[\omega_{11^*}^{Z,N}(Z,N,A_z,A_y)\omega_{1T}^{A_y}(A_y)\,Y\right],\label{eq:ID_Weighting}
\end{align}
where the   `interventional' weighting function
\begin{align}
    \omega_{11^*}^{Z,N}(Z,N,A_z,A_y)
    \coloneq
    \frac{\P_0(Z| A_z,A_y)\P_1(N|A_z,A_y)}{\P_1(Z,N|A_z,A_y)}\label{eq:Weight_ZN_1I}
\end{align}
shifts the observed distributions of $Z$ and $N$ to the critical densities $\lambda_0^Z$ and $\lambda_1^N$. This weighting function can be expressed in many ways that suggest avenues for its estimation. Here we consider the simple expression 
\begin{align}
    \omega_{11^*}^{(Z,N)}(Z,N,A_z,A_y)
    =\frac{\P_0(Z|A_z,A_y)}{\P_1(Z|N,A_z,A_y)} \label{eq:Weight_Z_Model_PW}.
\end{align}
The Z-Model-PW estimator is thus a weighted average of the outcome among the $G=1$ sample:
\begin{align}
    \hat\theta_{1,\text{Z-Model-PW}}^*
    \coloneq\frac{\P_{1n}[\hat\omega_{11^*}^{(Z,N)}\hat\omega_{1T}^{A_y}\,Y]}{\P_{1n}[\hat\omega_{11^*}^{(Z,N)}\hat\omega_{1T}^{A_y}]}.\label{eq:Estimator_Z_Model_PW}
\end{align}
It is consistent for $\theta_1^*$ if the interventional weight $\omega_{11^*}^{(Z,N)}$ and the standardizing weight $\omega_{1T}^{A_y}$ are consistently estimated. The weight expression of this estimator (\ref{eq:Weight_Z_Model_PW}) is identical to the one proposed in \citet{jackson2021meaningful}.


\subsubsection{Z-Model-SR: Sequential Regression}

This estimator is based on a sequential expectation expression of $\theta_1^*$:
\begin{align}
    \theta_1^*
    =\E_\text{std}\Big\{\underbrace{\E_1^*\big(\underbrace{\E_1[Y|Z,N,A_z,A_y]}_{\mu_1(Z,N,A_z,A_y)}\big|A_y\big)}_{\eta_1^*(A_y)}\Big\},\label{eq:SeqExp_Z_Model_SR}
\end{align}
where
\begin{align*}
    \E_1^*[h(Z,N,A_y,Az)| A_y]=\int\!\!\!\int\!\!\!\int h(z,n,a_z,A_y)\P_0(z| a_z,A_y)\P_1(n,a_z|a_z,A_y)dz\,dn\,da_z
\end{align*}
for arbitrary functions $h()$.

The innermost expectation $\E_1[\cdot]$ (denoted by $\mu_1(Z,N,A_z,A_y)$) is estimated by regression of the outcome $Y$ on the $G=1$ sample given $(Z,N,A_z,A_y)$; the middle expectation $\E_1^*(\cdot)$ (denoted by $\eta_1^*(A_y)$) is estimated by regression of the predicted values $\hat\mu_1$ on an artificial sample $G=\blacktriangle$ (defined shortly) given $A_y$, which effectively averages $\hat\mu_1$ over the modeled $\lambda_0^z$ and empirical $\lambda_1^N$; and the outermost expectation $\E_\std\{\cdot\}$ is estimated by averaging the predicted values $\hat\eta_1^*$ on the standard $A_y$ distribution sample.

\begin{definition}[artificial sample $G=\blacktriangle$]
An artificial sample $G=\blacktriangle$ used to estimate $\E_1^*(\cdot)$ has the joint distribution $\P_\blacktriangle(Z,N,A_z\mid A_y)=\P_1^*(Z,N,A_z\mid A_y)=\P_0(Z\mid A_z,A_y)\P_1(N,A_z\mid A_y)$.\label{def:Sample_Z_Model} 
\end{definition}

We can construct an artificial sample $G=\blacktriangle$ as follows: 1) Among the $G=0$ group, fit a model for $\lambda_0^{Z}$; 2) Take the $G=1$ sample, delete the $Z$ values, duplicate the sample many times, and stack the duplicates; 3) Draw a simulated value of $Z$ from $\hat{\lambda}_0^{Z}$.

The Z-Model-SR estimator is consistent for $\theta_1^*$ if both of the $\mu_1$ and $\eta_1^*$ models are correctly specified. This estimator generalizes \citet{sudharsanan2021educational} by accommodating intervention-allowability. It generalizes \citet{fairlie2005extension} by replacing its sampling step with draws from $\hat\lambda_0^Z$ and standardizing over outcome-allowables $A_y$.

\subsubsection{Z-Model-SWR: Sequential Weighted Regression} \label{subsubsec:Estimator_Z_Model_SWR}
The Z-Model-SWR estimator of $\theta_1^*$ is based on the same sequential expectation (\ref{eq:SeqExp_Z_Model_SR}) as Z-Model-SR but fits the $\mu_1$ and $\eta_1^*$ models to the sample space of the intervention arm. This entails: 1) weighting the $\mu_1$ model by $\hatomega_{11^*}^{(Z,N)}\hatomega_{1T}^{A_y}$ (\ref{eq:Weight_Z_Model_PW}) and (\ref{eq:Weight_Ay}); 2) weighting the $\eta_1^*$ model by $\hatomega_{1T}^{A_y}$ (\ref{eq:Weight_Ay}). The Z-Model-SWR estimator is consistent for $\theta_1^*$ if 1) the model for $\lambda_0^Z$ is correctly specified \emph{and} 2) the following nuisance components are correctly specified (in the case of outcome regressions) or consistently estimated (in the case of weighting functions): (2-i) $\mu_1$ or $\omega_{11^*}^{(Z,N)}\omega_{1T}^{A_y}$; and (2-ii) $\eta_1^*$ or $\omega_{1T}^{A_y}$. This amounts to four avenues for consistent estimation of $\theta_1^*$. Because the Z-Model-SWR estimator requires correct specification of the $\lambda_0^Z$ model, we distinguish it as partially robust.

\subsection{N-Modeling Estimators}
When modeling $\lambda_0^Z$ is difficult (e.g., when $Z$ is continuous), an alternate strategy is to model $\lambda_1^N$. This motivates the N-modeling estimators, which require a correctly specified model for $\lambda_1^N$, the critical density for $N$, and rely on the empirical version of $\lambda_0^Z$, the critical density for $Z$. In the \textit{Pure Weighting} (PW) and \textit{Regress then Weight} (RW) estimators, $\lambda_1^N$ appears in a density ratio used to take a weighted average of the outcome or pseudo-outcome. Whereas, in the \textit{Sequential Regression} (SR) estimator, an artificial sample $G=\triangle$ is created by simulating $N$ from $\lambda_1^N$. In the \textit{Sequential Weighted Regression} (SWR) estimator, $\lambda_1^N$ plays two roles: simulating $N$ and appearing in the weight used to fit the initial outcome model. 

\subsubsection{N-Model-PW: Pure Weighting}
This estimator is based on the same weighted-outome-mean expression of $\theta_1^*$ in (\ref{eq:ID_Weighting}) that is the basis of Z-Model-PW. Here, we factorize the interventional weighting function $\omega_{11^*}^{(Z,N)}$ into two terms $\omega_{11^*}^N$ and $\omega_{11^*}^Z$:
\begin{align}
    \omega_{11^*}^{(Z,N)}(Z,N,A_z,A_y)
    &=\underbrace{\frac{\P_1(N|A_z,A_y)}{\P_1(N|Z,A_z,A_y)}}_{=:\omega_{11^*}^N(Z,N,A_z,A_y)}\underbrace{\frac{\P_0(Z|A_z,A_y)}{\P_1(Z|A_z,A_y)}}_{=:\omega_{11^*}^Z(Z,A_z,A_y)}.\label{eq:Weight_N_Model_PW}
\end{align}
$\lambda_1^N$ appears in the first term $\omega_{11^*}^N$, and $\lambda_0^Z$ appears in the second term $\omega_{11^*}^Z$, which can be re-expressed as
\begin{align*}
    \omega_{11^*}^Z(Z,A_z,A_y)
    =\frac{\text{odds}(G=0~\text{vs}~1\mid Z,A_z,A_y)}{\text{odds}(G=0~\text{vs}~1\mid A_z,A_y)},
\end{align*}
suggesting a way to estimate $\omega_{11^*}^{(Z,N)}$ by modeling $\lambda_1^N$, but using a ratio of odds functions to avoid modeling $\lambda_0^Z$. The N-Model-PW estimator is the following weighted average of the outcome among the $G=1$ sample:
\begin{align}
    \hat\theta_{1,\text{N-Model-PW}}^*
    \coloneq\frac{\P_{1n}[\hat\omega_{11^*}^{(Z,N)}\hat\omega_{1T}^{A_y}\,Y]}{\P_{1n}[\hat\omega_{11^*}^{(Z,N)}\hat\omega_{1T}^{A_y}]},
\end{align}
where the estimate of the interventional weight $\omega_{11^*}^{(Z,N)}$ is provided by the product of $\hat\omega_{11^*}^N$ and $\hat\omega_{11^*}^Z$. It is consistent for $\theta_1^*$ if the interventional weight $\omega_{11^*}^{(Z,N)}$ and the standardizing weight $\omega_{1T}^{A_y}$ are consistently estimated.


\subsubsection{N-Model-SR: Sequential Regression}

This estimator is based on the following sequential expectation expression of $\theta_1^*$:
\begin{align}
    \theta_1^*=\E_\text{std}\bigg[\underbrace{\E_1\Big\{\underbrace{\E_1^*\big(\underbrace{\E_1[Y| Z,N,A_z,A_y]}_{\mu_1(Z,N,A_z,A_y)}\big| A_z,A_y\big)}_{\kappa_1^*(A_z,A_y)}\Big| A_y\Big\}}_{\eta_1^*(A_y)}\bigg],\label{eq:SeqExp_N_Model_SR}
\end{align}
where
\begin{align*}
    \E_1^*[h(Z,N,A_z,A_y)\mid A_z,A_y]=\int\!\!\!\int h(z,n,A_z,A_y)\P_1(n\mid A_z,A_y)\P_0(z\mid A_z,A_y)dn\,dz
\end{align*}
for arbitrary functions $h()$.

The innermost expectation $\E_1[\cdot]$ (denoted by $\mu_1(Z,N,A_z,A_y))$ is estimated by regressing the outcome $Y$ on $(Z,N,A_z,A_y)$ among the $G=1$ sample; the next expectation $\E_1^*(\cdot)$ (denoted by $\kappa_1^*(A_z,A_y)$) is then estimated by regression of predicted values $\hat\mu_1$ on $(A_z,A_y)$ among an artificial sample $G=\triangle$ (explained shortly), which effectively averages $\hat\mu_1$ over the empirical $\lambda_0^z$ and modeled $\lambda_1^N$; the third expectation $\E_1\{\cdot\}$ (denoted by $\eta_1^*(A_y)$) is then estimated by regression of predicted values $\hat\kappa_1^*$ on $A_y$ among the $G=1$ sample; and finally the outermost expectation $\E_\std[\cdot]$ is estimated by averaging the predicted values $\hat\eta_1^*$ on the standard $A_y$ distribution sample.
\begin{definition}[artificial sample $G=\triangle$]
An artificial sample $G=\triangle$ used to estimate $\E_1^*(\cdot)$ has the joint distribution $\P_\triangle(Z,N\mid A_z,A_y)=\P_1^*(Z,N\mid A_z,A_y)=\P_0(Z\mid A_z,A_y)\P_1(N\mid A_z,A_y)$.\label{def:Sample_N_Model}
\end{definition}
We can construct an artificial sample $G=\triangle$ as follows: 1) Among the $G=1$ group, fit a model for $\lambda_1^{N}$; 2) Take the $G=0$ sample, delete the $N$ values, duplicate the sample many times, and stack the duplicates; 3) Draw a simulated value of $N$ from $\hat{\lambda}_1^{N}$.

This estimator requires correct specification of the $\mu_1$, $\kappa_1^*$, and $\eta_1^*$ models to achieve consistent estimation of $\theta_1^*$. This estimator generalizes \citet{park2024estimation} by accommodating intervention-allowability.

\subsubsection{N-Model-SWR: Sequential Weighted Regression} \label{subsubsec:Estimator_N_Model_SWR}
The N-Model-SWR estimator of $\theta_1^*$ is based on the same sequential expectation as that of N-Model-SR  (\ref{eq:SeqExp_N_Model_SR}) but instead fits the $\mu_1$, $\kappa_1^*$, and $\eta_1^*$ models to the sample space of the intervention arm. This entails: 1) weighting the $\mu_1$ model by $\hatomega_{11^*}^{(Z,N)}\hatomega_{1T}^{A_y}$ (\ref{eq:Weight_N_Model_PW}) and (\ref{eq:Weight_Ay}); 2) weighting the $\kappa_1^*$ model by $\hatomega_{01}^{A_z}\hatomega_{0T}^{A_y}$ (\ref{eq:Weight_Az}) and (\ref{eq:Weight_Ay}); 3) weighting the $\eta_1^*$ model by $\hatomega_{1T}^{A_y}$ (\ref{eq:Weight_Ay}). The N-Model-SWR estimator is consistent for $\theta_1^*$ if 1) the density model $\lambda_1^N$ is correctly specified \emph{and} 2) the following nuisance components are correctly specified (in the case of outcome regressions) or consistently estimated (in the case of weighting functions): (i) $\mu_1$ or $\omega_{11^*}^{(Z,N)}\omega_{1T}^{A_y}$; (ii) $\kappa_1^*$ or $\omega_{01}^{A_z}\omega_{0T}^{A_y}$;  and (iii) $\eta_1^*$ or $\omega_{1T}^{A_y}$.  This amounts to eight avenues for consistent estimation of $\theta_1^*$. Because the N-Model-SWR estimator requires correct specification of $\lambda_1^N$, we distinguish it as partially robust.

\subsection{Z-bridging Estimators} \label{subsec:Estimators_Z_Bridge}
Sometimes $\lambda_0^Z$ and $\lambda_1^N$ may be difficult to model (e.g., both $Z$ and $N$ are continuous). Z-bridging estimators, which avoid modeling either density, are based on an artificial population group with certain properties.  
\begin{definition}[artificial group $G=\blackdiamond$]
The artificial group $G=\blackdiamond$ has a joint distribution $\P_\blackdiamond(Z,N,A_z,A_y)=\P_\blackdiamond(Z|N,A_z,A_y)\P_\blackdiamond(N,A_z,A_y)$ where $\P_\blackdiamond(N,A_z,A_y)=\P_1(N,A_z,A_y)$ and $\P_\blackdiamond(Z|N,A_z,A_y)=\P_\blackdiamond(Z)$. The distribution $\P_\blackdiamond(Z)$ is any arbitrary distribution of $Z$ that satisfies a support condition: the support of this distribution covers a subset of the $G=0$ group's support of $Z$, where the subset results from needing only to consider values that $Z$ takes in this group when $(A_z,A_y)$ values fall within the $G=1$ group's support of these variables. \label{def:Group_Z_Bridge}
\end{definition}

Overall, the $G=\blackdiamond$ group has the $(N,A_z,A_y)$ distribution of the $G=1$ group, and a distribution of $Z$ that does not vary by $(N,A_z,A_y)$. Together, these properties permit expressions of $\theta_1^*$ that do not involve $\lambda_0^Z$ or $\lambda_1^N$.

Although the $G=\blackdiamond$ group is not a real population subgroup, there are multiple ways to construct a $G=\blackdiamond$ sample that is i.i.d. with the distributional properties of the $G=\blackdiamond$ group outlined in Definition \ref{def:Group_Z_Bridge}. When $Z$ is discrete, we can construct the artificial $G=\blackdiamond$ sample as follows: 1) Determine the set of unique values $z$ in the support $\mathcal{Z}$ of $Z$ in the $G=0$ sample; 2) Take the $G=1$ sample, create a duplicate for each unique value of $z$, where each duplicate's $Z$ is assigned a unique value $z$, retain the original values of $(N,A_z,A_y)$, and stack the duplicates. When $Z$ is continuous, we can replace steps 1 and 2 by stacking many duplicates of the $G=1$ sample, retaining the original $(N,A_z,A_y)$ values, deleting the $Z$ values and either sampling $Z$ from the $G=0$ sample with replacement or by drawing $Z$ from a uniform distribution that covers $\mathcal{Z}$ the support of $Z$ in the $G=0$ sample.

In the \textit{Pure Weighting} (PW) and \textit{Regress then Weight} (RW) estimators, the artificial $G=\blackdiamond$ sample allows density ratios involving $\lambda_0^Z$ or $\lambda_1^N$ to be expressed purely in terms of odds functions. In the \textit{Sequential Regression} (SR) estimator, the artificial sample $G=\blackdiamond$ bridges the outcome regression function $\mu_1$ to the empirical version of $\lambda_0^Z$. In the \textit{Sequential Weighted Regression} (SWR), the artificial sample $G=\blackdiamond$ plays both roles simultaneously. 

\subsubsection{Z-Bridge-PW: Pure Weighting}

The Z-Bridge-PW estimator is based on the same weighted-outome-mean expression of $\theta_1^*$ in (\ref{eq:ID_Weighting}) that is the basis of Z-Model-PW. We consider the interventional weighting function $\omega_{11^*}^{(Z,N)}$ as before but now modify its expression in N-Model-PW (\ref{eq:Weight_N_Model_PW}). Specifically, we consider its component $\omega_{11^*}^N$, and replace its numerator $\P_1(N|A_z,A_y)$ with $\P_\blackdiamond(N|A_z,A_y)=\P_\blackdiamond(N|Z,A_z,A_y)$, and label the component as $\omega_{1\blackdiamond}^N$:
\begin{align}
    \omega_{11^*}^{(Z,N)}(Z,N,A_z,A_y)
    &=\underbrace{\frac{\P_\blackdiamond(N|Z, A_z,A_y)}{\P_1(N| Z,A_z,A_y)}}_{=:\omega_{1\blackdiamond}^N(Z,N,A_z,A_y)}\underbrace{\frac{\P_0(Z|A_z,A_y)}{\P_1(Z|A_z,A_y)}}_{=:\omega_{11^*}^Z(Z,A_z,A_y)}.\label{eq:Weight_Z_Bridge_PW_2}
\end{align}
where $\lambda_1^N$ appears in $\omega_{1\blackdiamond}^N$ and $\lambda_0^Z$ appears in $\omega_{11^*}^Z$. Both pieces can be expressed as in odds terms:
\begin{align*}
    \omega_{1\blackdiamond}^N(Z,N,A_z,A_y)=\frac{\text{odds}(G=\blackdiamond~\text{vs}~1\mid N,Z,A_z,A_y)}{\text{odds}(G=\blackdiamond~\text{vs}~1\mid Z,A_z,A_y)}
\end{align*}
\begin{align*}
    \omega_{11^*}^Z(Z,A_z,A_y)
    =\frac{\text{odds}(G=0~\text{vs}~1\mid Z,A_z,A_y)}{\text{odds}(G=0~\text{vs}~1\mid A_z,A_y)},
\end{align*}
providing an alternate way to estimate $\omega_{11^*}^{(Z,N)}$ without modeling $\lambda_1^N$ or $\lambda_0^Z$ using ratios of odds functions.\footnote{The expression $\odds(G=\blackdiamond~vs~g^\prime|\cdot)$, where $g^\prime$ is a fixed value of $G\in(1,0)$, is evaluated in a dataset where $G=\blackdiamond$ data is stacked on $G=g^\prime$ data that is cloned to the same degree as $G=\blackdiamond$ data (see Definition \ref{def:Group_Z_Bridge}) but retains its original $Z$ values.\label{note:odds_bridge_Z}} Putting the two pieces together and simplifying yields the following expression:
\begin{align}
\omega_{11^*}^Z(Z,N,A_z,A_y)=\underbrace{\odds(G=\blackdiamond~vs~1|Z,N,A_z,A_y)\frac{\odds(G=1~vs~0|A_z,A_y)}{\odds(G=\blackdiamond~vs~0|Z,A_z,A_y}}_{\omega_{1\blackdiamond}^N(Z,N,A_z,A_y)\omega_{11^*}^Z(Z,A_z,A_y)}\nonumber
\end{align}

Thus, the Z-Bridge-PW estimator is the following weighted average of the outcome among the $G=1$ sample:
\begin{align}
    \hat\theta_{1,\text{Z-Bridge-PW}}^*
    \coloneq\frac{\P_{1n}[\hat\omega_{11^*}^{(Z,N)}\hat\omega_{1T}^{A_y}\,Y]}{\P_{1n}[\hat\omega_{11^*}^{(Z,N)}\hat\omega_{1T}^{A_y}]}.
\end{align}
where the estimate of the interventional weight $\omega_{11^*}^{(Z,N)}$ is provided by the product of $\hat\omega_{1\blackdiamond}^N$ and $\hat\omega_{11^*}^Z$. It is consistent for $\theta_1^*$ if the interventional weight $\omega_{11^*}^{(Z,N)}$ and the standardizing weight $\omega_{1T}^{A_y}$ are consistently estimated.

\subsubsection{Z-Bridge-SR: Sequential Regression}

This estimator is based on the following expression of $\theta_1^*$ as a sequential expectation
\begin{align*}
    \theta_1^*
    &=\E_\text{std}\Bigg(\E_1\bigg[\E_0\Big\{\E_1^*\big(\E_1[Y| Z,N,A_z,A_y]\big| Z,A_z,A_y\big)\Big| A_z,A_y\Big\}\bigg| A_y\bigg]\Bigg).
\end{align*}
The middle expectation $\E_1^*(\cdot)$ is taken over the $(N|Z,A_z,A_y)$ distribution of the $G=1$ group in the intervention arm   
\begin{align*}
    \E_1^*[h(Z,N,A_z,A_y)| Z,A_z,A_y]
    &=\int h(Z,n,A_z,A_y)\P_1^*(n|Z,A_z,A_y)dn
\end{align*}
for arbitrary functions $h()$.
Because $\P_\blackdiamond(N|Z,A_z,A_y)=\P_1^*(N|Z,A_z,A_y)$, it follows that this expectation could be equivalently taken over $\P_\blackdiamond(N|Z,A_z,A_y)$ of the $G=\blackdiamond$ group, i.e.,  
$
\E_1^*[h(Z,N,A_z,A_y)|Z,A_z,A_y]=\E_\blackdiamond[h(Z,N,A_z,A_y)|Z,A_z,A_y].$
This implies that we can re-express $\theta_1^*$ in terms of the $G=\blackdiamond$ group
\begin{align}
    \theta_1^*
    =\E_\text{std}\Bigg(\underbrace{\E_1\bigg[\underbrace{\E_0\Big\{\underbrace{\E_\blackdiamond\big(\underbrace{\E_1[Y| Z,N,A_z,A_y]}_{\mu_1(Z,N,A_z,A_y)}\big| Z,A_z,A_y\big)}_{\zeta_1^*(Z,A_z,A_y)}\Big| A_z,A_y\Big\}}_{\kappa_1^*(A_z,A_y)}\bigg| A_y\bigg]}_{\eta_1^*(A_y)}\Bigg).\label{eq:Seq_Exp_Z_Bridge_SR}
\end{align}
Based on this expression, the Z-Bridge-SR estimator involves a sequence of regressions. The innermost expectation $\E_1[\cdot]$ (denoted by $\mu_1(Z,N,A_z,A_y)$) is estimated by regressing the outcome $Y$ on $(Z,N,A_z,A_y)$ among the $G=1$ sample; the next expectation $\E_\blackdiamond(\cdot)$ (denoted $\zeta_1^*(Z,A_z,A_y)$) is then estimated by regressing the predicted values $\hat\mu_1$ on $(Z,A_z,A_y)$ among the artifical $G=\blackdiamond$ sample, which effectively averages $\hat\mu_1$ over the empirical $\lambda_1^N$; the next expecataion $\E_0\{\cdot\}$ (denoted $\kappa_1^*(A_z,A_y)$) is estimated by regressing $\hat\zeta_1^*$ on $(A_z,A_y)$ among the $G=0$ sample, which effectively averages $\hat\zeta_1^*$ over the empirical $\lambda_0^Z$; the next epxectation $\E_1[\cdot]$ (denoted $\eta_1^*(A_y)$) is estimated by regressing $\hat\kappa_1^*$ on $A_y$ among the $G=1$ sample. Finally, the outermost expectation $\E_\std(\cdot)$ is estimated by averaging the predicted values $\hat\eta_1^*$ on the standard $A_y$ sample. This estimator requires correct specification of the $\mu_1$, $\zeta_1^*$, $\kappa_1^*$, and $\eta_1^*$ models to achieve consistent estimation of $\theta_1^*$.

\subsubsection{Z-Bridge-SWR: Sequential Weighted Regression} \label{subsubsec:Estimator_Z_Bridge_SWR}

To construct a Z-bridging SWR estimator, we need a new weighting function $\omega_{\blackdiamond1^*}^Z(Z,A_z,A_y)$ 
\begin{align}
    \omega_{\blackdiamond1^*}^Z(Z,A_z,A_y) &\coloneq  \frac{\P_0(Z|A_z,A_y)}{\P_\blackdiamond(Z|A_z,A_y)}= \frac{\odds(G=0\text{ vs }\blackdiamond|Z,A_z,A_y)}{\odds(G=0\text{ vs }\blackdiamond|A_z,A_y)}\label{eq:Weight_Z_Bridge_RW}.
\end{align}
This weight takes the artificial sample $G=\blackdiamond$ and morphs its distribution of $Z$, i.e., $\P_\blackdiamond(Z|A_z,A_y)$, to $\lambda_0^Z$. When expressed in odds terms (see note \ref{note:odds_bridge_Z}), it avoids a model for $\lambda_0^Z$.


The Z-Bridge-SWR estimator of $\theta_1^*$ is based on the same sequential expectation as Z-Bridge-SR (\ref{eq:Seq_Exp_Z_Bridge_SR}) but instead fits the $\mu_1$, $\zeta_1^*$, $\kappa_1^*$, and $\eta_1^*$ models to the sample space of the intervention arm. This entails: 1) weighting the $\mu_1$ model by $\hat\omega_{11^*}^{(Z,N)}\hat\omega_{1T}^{A_y}$ (\ref{eq:Weight_Z_Bridge_PW_2}) and (\ref{eq:Weight_Ay}); 2) weighting the $\zeta_1^*$ model by $\hat\omega_{\blackdiamond1^*}^Z\hat\omega_{1T}^{A_y}$ (\ref{eq:Weight_Z_Bridge_RW}) and (\ref{eq:Weight_Ay}); 3) weighting the $\kappa_1^*$ model by $\hatomega_{01}^{A_z}\hatomega_{0T}^{A_y}$ (\ref{eq:Weight_Az}) and (\ref{eq:Weight_Ay}); 4) weighting the $\eta_1^*$ model by $\hatomega_{1T}^{A_y}$ (\ref{eq:Weight_Ay}). Provided that the artificial $G=\blackdiamond$ sample is properly constructed, the Z-Bridge-SWR estimator is consistent for $\theta_1^*$ if the following nuisance components are correctly specified (in the case of outcome regressions) or consistently estimated (in the case of weighting functions): (i) $\mu_1$ or $\omega_{11^*}^{(Z,N)}\omega_{1T}^{A_y}$; (ii) $\zeta_1^*$ or $\omega_{\blackdiamond1^*}^Z\omega_{1T}^{A_y}$; (iii) and $\kappa_1^*$ or $\omega_{01}^{A_z}\omega_{0T}^{A_y}$; and (iv) $\eta_1^*$ or $\omega_{1T}^{A_y}$. This amounts to sixteen avenues for consistent estimation of $\theta_1^*$. The Z-Bridge-SWR estimator is fully robust because it does not require any particular model to be correct.

\subsection{N-bridging Estimators} \label{subsec:Estimators_N_Bridge}
The Z-bridging estimators avoid modeling either critical density but are more complex than the Z- and N-modeling estimators. The N-bridging estimators also avoid modeling $\lambda_0^Z$ and $\lambda_1^N$ but are simpler. Like the Z-bridging estimators, they are based on an artificial population group, albeit one with different properties.
\begin{definition}[artificial group $G=\diamond$]
The artificial group $G=\diamond$ has a joint distribution $\P_\diamond(Z,N,A_z,A_y)=\P_\diamond(Z,A_z,A_y)\P_\diamond(N|Z,A_z,A_y)$ where $\P_\diamond(Z,A_z,A_y)=\P_0(Z,A_z,A_y)$ and $\P_\diamond(N|Z,A_z,A_y)=\P_\diamond(N)$. The distribution $\P_\diamond(N)$ is any distribution that covers the $G=1$ group's support of $N$.
\label{def:Group_N_Bridge}
\end{definition}

Overall, the $G=\diamond$ group has the $(Z,A_z,A_y)$ distribution of the $G=0$ group, and a distribution of $N$ that does not vary by $(Z,A_z,A_y)$. Together, these properties permit expressions of $\theta_1^*$ that do not involve $\lambda_0^Z$ or $\lambda_1^N$.

Though the $G=\diamond$ group is not real, we can construct an i.i.d. $G=\diamond$ sample with its properties as follows. When $N$ is discrete, we can use the following procedure: 1) Determine the set of unique values $N$ in the support $\mathcal{N}$ of $N$ in the $G=1$ sample; 2) Take the $G=0$ sample, create a duplicate for each unique value $n$, where each duplicate's $N$ is assigned a unique value $n$, retain the original values of $(Z,A_z,A_y)$, and stack the duplicates. When $N$ is continuous or multivariate, we can replace steps 1 and 2 by stacking many duplicates of the $G=0$ sample, retaining the original $(Z,A_z,A_y)$ values, and deleting the $N$ values and either sampling $N$ from the $G=1$ sample with replacement or drawing $N$ from a uniform distribution that covers $\mathcal{N}$ the support of $N$ in the $G=1$ sample.

In the \textit{Pure Weighting} (PW) and \textit{Regress then Weight} (RW) estimators, the artificial $G=\diamond$ sample  allows density ratios involving $\lambda_0^Z$ or $\lambda_1^N$ to be modeled using ratios of odds functions. In the \textit{Sequential Regression} (SR) estimator, the artificial $G=\diamond$ sample bridges the outcome regression function $\mu_1$ to the empirical version of $\lambda_1^N$. In \textit{Sequential Weighted Regression} (SWR), the artificial sample $G=\diamond$ plays both roles simultaneously.

\subsubsection{N-Bridge-PW: Pure Weighting}
The N-Bridge-PW estimator is based on same weighted-outome mean expression of $\theta_1^*$ in (\ref{eq:ID_Weighting}) that is the basis of Z-Model-PW. Here, we modify the interventional weighting function's expression used in Z-Model-PW (\ref{eq:Weight_Z_Model_PW}) by replacing its numerator $\P_0(Z|A_z,A_y)$ with $\P_\diamond(Z|A_z,A_y)=\P_\diamond(Z|N,A_z,A_y)$, and label this form as $\omega_{1\diamond}^Z$. 
\begin{align}
    \omega_{11^*}^{(Z,N)}(Z,N,A_z,A_y)
    =\underbrace{\frac{\P_\diamond(Z| N,A_z,A_y)}{\P_1(Z| N,A_z,A_y)}}_{=:\omega_{1\diamond}^Z(Z,N,A_z,A_y)},\label{eq:Weight_N_Bridge_PW}
\end{align}
where $\lambda_0^Z$ appears in the numerator. Expressing the weighting function $\omega_{1\diamond}^Z$ in terms of odds\footnote{The expression $\odds(G=\diamond~vs~1|\cdot)$ is evaluated in a dataset where the $G=\diamond$ data is stacked on the $G=1$ data that is cloned to the same degree as the $G=\diamond$ data (see Definition \ref{def:Group_N_Bridge}) but retains its original $N$ values.\label{note:odds_bridge_N}} 
\begin{align*}
    \omega_{1\diamond}^Z(Z,N,A_z,A_y)=\frac{\text{odds}(G=\diamond~\text{vs}~1\mid Z,N,A_z,A_y)}{\text{odds}(G=\diamond~\text{vs}~1\mid N,A_z,A_y)}
\end{align*}
suggests a simple way to estimate $\omega_{11^*}^{(Z,N)}$ without modeling $\lambda_0^Z$ or $\lambda_1^N$. Thus, the N-Bridge-PW estimator is the following weighted average of the outcome among the $G=1$ sample:
\begin{align}
    \hat\theta_{1,\text{N-Bridge-PW}}^*
    \coloneq\frac{\P_{1n}[\hat\omega_{11^*}^{(Z,N)}\hat\omega_{1T}^{A_y}\,Y]}{\P_{1n}[\hat\omega_{11^*}^{(Z,N)}\hat\omega_{1T}^{A_y}]},
\end{align}
where the estimate of the interventional weight $\omega_{11^*}^{(Z,N)}$ is provided by $\hat\omega_{1\diamond}^Z$. It is consistent for $\theta_1^*$ if the interventional weight $\omega_{11^*}^{(Z,N)}$ and the standardizing weight $\omega_{1T}^{A_y}$ are consistently estimated. 


\subsubsection{N-Bridge-SR: Sequential Regression}
This estimator is based on the following sequential expectation expression of $\theta_1^*$
\begin{align*}
\theta_1^*=\E_\text{std}\Big[\E_1\Big\{\E_1^*\big(\E_1[Y| Z,N,A_z,A_y]\big| N,A_z,A_y\big)\Big| A_y\Big\}\Big].
\end{align*}
The middle expectation $\E_1^*(\cdot)$ is taken over the $(Z|N,A_z,A_y)$ distribution of the $G=1$ group in the intervention arm
\begin{align*}
    \E_1^*[h(Z,N,A_z,A_y)| N,A_z,A_y]=\int h(z,N,A_z,A_y)\P_1^*(z| N,A_z,A_y)dz
\end{align*}
for arbitrary functions $h()$. Because $\P_\diamond(Z|N,A_z,A_y)=\P_1^*(Z|N,A_z,A_y)$, it follows that this expectation could be equivalently taken over $\P_\diamond(Z|N,A_z,A_y)$ of the $G=\diamond$ group, i.e., $\E_1^*[h(Z,N,A_z,A_y)| N,A_z,A_y]=\E_\diamond[h(Z,N,A_z,A_y)| N,A_z,A_y]$. This implies that we can re-express $\theta_1^*$ in terms of the $G=\diamond$ group
\begin{align}
    \theta_1^*=\E_\text{std}\bigg[
    \underbrace{\E_1\Big\{
    \underbrace{\E_\diamond\big(
    \underbrace{\E_1[Y| Z,N,A_z,A_y]}_{\mu_1(Z,N,A_z,A_y)}
    \big| N,A_z,A_y\big)}_{\nu_1^*(N,A_z,A_y)}
    \Big| A_y\Big\}}_{\eta_1^*(A_y)}
    \bigg].\label{eq:Seq_Exp_N_Bridge_SR}
\end{align}
Based on this expression, the N-Bridge-SR estimator involves a sequence of regressions. The innermost expectation $\E_1[\cdot]$ (denoted by $\mu_1(Z,N,A_z,A_y)$) is estimated by regressing the outcome $Y$ on $(Z,N,A_z,A_y)$ among the $G=1$ sample; the next expectation $\E_\diamond(\cdot)$ (denoted by $\nu_1^*(N,A_z,A_y)$) is then estimated by regressing the predicted values $\hat\mu_1$ on $(N,A_z,A_y)$ among the $G=\diamond$ sample, which effectively averages $\hat\mu_1$ over the empriical $\lambda_0^Z$; the next expectation $\E_1\{\cdot\}$ (denoted by $\eta_1^*(A_y)$) is estimated by regressing the predicted values $\hat\nu_1^*$ on $A_y$ among the $G=1$ sample. Finally, the outermost expectation $\E_\std[\cdot]$ is estimated by averaging the predicted values $\hat\eta_1^*$ on the standard $A_y$ sample. This estimator requires correct specification of the $\mu_1$, $\nu_1^*$, and $\eta_1^*$ models to consistently estimate $\theta_1^*$.

\subsubsection{N-Bridge-SWR: Sequential Weighted Regression} \label{subsubsec:Estimator_N_Bridge_SWR}

To construct a N-bridging SWR estimator, we need a new weighting function $\omega_{\diamond1^*}^{(N,A_z)}(N,A_z,A_y)$:
\begin{align}
    \omega_{\diamond1^*}^{(N,A_z)}(N,A_z,A_y) &\coloneq\frac{\P_1(N|A_z,A_y)}{\P_\diamond(N|A_z,A_y)}\frac{\P_1(A_z|A_y)}{\P_\diamond(A_z|A_y)}      =\frac{\odds(G=1\text{ vs }\diamond|N,A_z,A_y)}{\odds(G=1\text{ vs }\diamond|A_y)}.\label{eq:Weight_N_Bridge_RW}
\end{align}
This weight takes the artificial sample $G=\diamond$ and morphs its joint distribution of $(N,A_z)$, $\P_\diamond(N,A_z|A_y)$, to that of the $G=1$ sample, $\P_1(N,A_z|A_y)$. When expressed in odds terms (see note \ref{note:odds_bridge_N}), it avoids a model for $\lambda_1^N$.

The N-Bridge-SWR estimator of $\theta_1^*$ is based on the same sequential expectation as N-Bridge-SR (\ref{eq:Seq_Exp_N_Bridge_SR}) but instead fits the $\mu_1$, $\nu_1^*$, and $\eta_1^*$ models to the sample space of the intervention arm. This entails: 1) weighting the $\mu_1$ model by $\hatomega_{11^*}^{(Z,N)}\hatomega_{1T}^{A_y}$ (\ref{eq:Weight_N_Bridge_PW}) and (\ref{eq:Weight_Ay}); 2) weighting the $\nu_1^*$ model by $\hatomega_{\diamond1^*}^{(N,A_z)}\hatomega_{0T}^{A_y}$ (\ref{eq:Weight_N_Bridge_RW}) and (\ref{eq:Weight_Ay}); 3) weighting the $\eta_1^*$ model by $\hatomega_{1T}^{A_y}$ (\ref{eq:Weight_Ay}). Provided that the artificial $G=\diamond$ sample is properly constructed, the N-Bridge-SWR estimator is consistent for $\theta_1^*$ if the following nuisance components are correctly specified (in the case of outcome regressions) or consistently estimated (in the case of weighting functions): (i) $\mu_1$ or $\omega_{11^*}^{(Z,N)}\omega_{1T}^{A_y}$, and (ii) $\nu_1^*$ or $\omega_{\diamond1^*}^{(N,A_z)}\omega_{1T}^{A_y}$, and (iii) $\eta_1^*$ or $\omega_{1T}^{A_y}$. This amounts to eight avenues for consistent estimation of $\theta_1^*$. The N-Bridge-SWR is fully robust because it does not require any particular model to be correct.

\subsection{Settings where Simpler Estimation Strategies Emerge} \label{subsec:SimplifiedEstimators}

When no covariates are allowable (i.e., both $A_z$ and $A_y$ are empty), the linear estimator (E-OBD; (\ref{eq:Estimator_EOBD})) and the Z-Model-SR estimator (based on (\ref{eq:SeqExp_Z_Model_SR})), respectively, recover causal versions of the Oaxaca-Blinder and Fairlie decompositions described in Section \ref{sec:SDA}. When no covariates are intervention-allowable but some are outcome-allowable (i.e., $A_z$ is empty but $A_y$ is not), the steps involving $\kappa_1^*$ in the N-Model-SR, N-Model-SWR, Z-Bridge-SR, and Z-Bridge-SWR estimators can be skipped and the covariate-shifting weights $\omega_{01}^{A_z}$ reduce to one. When no covariates are outcome-allowable but some are intervention-allowable (i.e., $A_y$ is empty but $A_z$ is not), the steps involving $\eta_1^*$ can be skipped and the standardizing weights $\omega_{gT}^{A_y}$ reduce to one. $\eta_1^*$ is also obviated and the standardizing weight $\omega_{1T}^{A_y}$ reduces to one when $A_y$ is non-empty but $G=1$ is the standard population. When all covariates are allowable in some way (i.e., $N$ is empty but one of $A_z$ or $A_y$ or both are non-empty), the steps involving $\lambda_1^N$ (N-modeling estimators) and the steps that involve $\nu_1^*$ and assigning or sampling values of $N$ (N-bridging estimators) can be skipped. In fact, the following paired estimators coincide with one another within each pair: N-Model-PW and N-Bridge-PW, N-Model-SR and N-Bridge-SR, N-Model-SWR and N-Bridge-SWR, and N-Model-RW and N-Bridge-RW, showing that in this case the N-modeling and N-bridging strategies collapse into a single strategy.

\subsection{Extension to Multivariate Points of Intervention}\label{MultivariateIntervention}

Each estimation strategy readily extends to the setting where $\mathbf{Z}=\{Z_1,\dots,Z_J\}$ is a vector of intervention points. Here, we focus on a joint intervention to shift the distribution of $\mathbf{Z}$ to $\P_0(\mathbf{Z}|A_z,A_y)$ under a common set of intervention-allowables $A_z$. We simply modify our identification assumptions \ref{def:Assumption_Ay_Overlap}, \ref{def:Assumption_Innocuous_Sampling}, \ref{def:Assumption_Az_Overlap}, \ref{def:Assumption_Z_Exchangeability}, \ref{def:Assumption_Z_Positivity} and \ref{def:Assumption_Z_Consistency} to hold with respect to $\mathbf{Z}$, with the the key constraint that no component of $\mathbf{Z}$ may affect $(N,A_z,A_y)$. Then we replace $Z$ with $\mathbf{Z}$ in our identification formula (\ref{eq:ID_Result_Intervention_Arm}), in the outcome regressions (i.e., (\ref{eq:Y_Model_Linear}), $\mu_1$, and $ \zeta_1$) and in the components of the weighting functions. When the distribution of $\mathbf{Z}$ is modeled, as for the Z-Modeling estimators, one can model the factored joint distribution, e.g., $\P(\mathbf{Z}|\cdot)=\P(Z_1|\cdot)\times\P(Z_2|Z_1,\cdot)\times\dots\times\P(Z_J|Z_1,Z_2,\dots,Z_{J-1},\cdot)$. For the linear estimator (E-OBD), we can use (\ref{eq:Estimator_EOBD_Categorical}) to obtain $\theta_1^*$, where now $Z_j$ refers to a specific point of intervention. For the Z-bridging estimators, one can sample $\mathbf{Z}$ from the $G=0$ group or a multivariate uniform distribution to build the artificial sample $G=\blackdiamond$. For the N-modeling and N-bridging estimators, one need only to replace $Z$ with $\mathbf{Z}$ in the nuisance components.

\subsection{Perspectives from the Influence function} \label{eq:Influence_Function_Counterfactual_Arm}

The influence function for $\theta_1^*$, derived in the Supplement, suggests the estimation strategies we have presented.
\begin{align}
\varphi_{}^*(O)=\varphi^*_{I}(O)+\varphi_{II}^*(O)+\varphi_{III}^*(O)+\varphi_{IV}^*(O),\label{eq:if-theta.1*}
\end{align}
where 
\begin{align*}
\varphi^*_{\I}(O)&=\frac{G}{\P(G=1)}\omega_{1T}^{A_y}(A_y)\omega_{11^*}^{(Z,N)}(Z,N,A_z,A_y)[Y-\mu_1(Z,N,A_z,A_y)] \nonumber \\
\varphi_{\II}^*(O)&=\frac{1-G}{\P(G=0)}\omega_{0T}^{A_y}(A_y)\omega_{01}^{A_z}(A_z,A_y)[\zeta_1^*(Z,A_z,A_y)-\kappa_1^*(A_z,A_y)] \nonumber\\
\varphi_{\III}^*(O)&=\frac{G}{\P(G=1)}\omega_{1T}^{A_y}(A_y)[\nu_1^*(N,A_z,A_y)-\eta_1^*(A_y)] \nonumber \\
\varphi_{\IV}^*(O)&=\frac{T}{\P(T=1)}[\eta_1^*(A_y)-\theta_1^*],
\end{align*}
and $O=(T,G,A_y,A_z,N,Z,Y)$. All PW estimators, and the Z-Model-SR and Z-Model-SWR estimators are reflected in the term $\varphi^*_{\I}(O)$. The N-modeling and Z-bridging SR and SWR estimiators, as well as the linear estimator E-OBD are reflected in the term $\varphi^*_{\II}(O)$. The N-bridging SR and SWR estimators are reflected in the term $\varphi^*_{\III}(O)$. The last term $\varphi^*_{\IV}(O)$ appears when $\P_{\std}(A_y)$ is known through the data, which is a sample.

\section{Diagnostics} \label{sec:Diagnostics}

Given the critical role of $\lambda_0^Z$ and $\lambda_1^N$ in the modeling estimators, we propose the following diagnostics:
\begin{alignat}{1}
	\P_{\square n}[b_0(Z,A_z,A_y)] &= \P_{0n}[Y] \label{eq:DiagnosticModelLambda} \\
	\P_{\blacksquare n}[b_1(N,A_z,A_y)] &= \P_{1n}[Y] \label{eq:DiagnosticModelDelta}
\end{alignat}
where $b_0(\cdot)$ and $b_1(\cdot)$ are \emph{working} models (i.e., not necessarily correctly specified models) for the conditional mean of $Y$ (i.e., $b_0(Z,A_z,A_y)\equiv \E_0[Y|Z,A_z,A_y]$ and $b_1(N,A_z,A_y)\equiv \E_1[Y|N,A_z,A_y])$ fit by GLMs with canonical links (to leverage mean recovery). The artificial sample $G=\square$ is created by taking the $G=0$ sample and drawing $Z$ from the model for $\lambda_0^z$. The artificial sample $G=\blacksquare$ is created by taking the $G=1$ sample and drawing $N$ from the model for $\lambda_1^N$. (\ref{eq:DiagnosticModelLambda}) evaluates the model for $\lambda_0^Z$ for use in $Z$-modeling estimators, and (\ref{eq:DiagnosticModelDelta}) evaluates the model for $\lambda_1^N$ for use in the $N$-modeling estimators. We want the LHS sample mean to return the appropriate sample mean on the RHS. Because the working models $b_0$ and $b_1$ exhibit mean recovery, this occurs when the critical density is correctly specified but not otherwise. (\ref{eq:DiagnosticModelLambda}) and (\ref{eq:DiagnosticModelDelta}) each isolate the quality of a particular critical density model because they are expressed in terms of $Z$ or $N$ but not both.

Many of the estimators use weighting functions. As density ratios, when consistently estimated, each expression of the interventional weights $\hat{\omega}^{(Z,N)}_{11^*}$ ((\ref{eq:Weight_Z_Model_PW}), (\ref{eq:Weight_N_Model_PW}), (\ref{eq:Weight_Z_Bridge_PW_2}), and (\ref{eq:Weight_N_Bridge_PW})) will have a mean of one among the $G=1$ sample; the weights $\hat{\omega}^Z_{\blackdiamond 1^*}$ and $\hat{\omega}^{(N,A_z)}_{\diamond 1^*}$ used in the second steps of the bridging SWR estimators will, respectively, have a mean of one among the artificial samples $G=\blackdiamond$ and $G=\diamond$; the covariate-shifting weight $\hat{\omega}^{A_z}_{01}$ (\ref{eq:Weight_Az}) will have a mean of one among the $G=0$ sample; and the standardizing weight $\hat{\omega}^{A_y}_{gT}$ (\ref{eq:Weight_Ay}) will have a mean of one among the full sample that combines the $G=1$ and $G=0$ samples. These diagnostics are satisfied when weights are properly constructed. One can also assess the weights' distribution. However, moment diagnostics are not sufficient because they can be satisfied in pathological cases (e.g., all observations' weights equal one).

A further diagnostic is to check whether the weights or their components balance the distribution of their targeted covariates to that of their intended target population (see Table \ref{tab:WeightProperties} for details). For each expression of the interventional weights $\hat{\omega}^{(Z,N)}_{11^*}$ ((\ref{eq:Weight_Z_Model_PW}), (\ref{eq:Weight_N_Model_PW}), (\ref{eq:Weight_Z_Bridge_PW_2}), and (\ref{eq:Weight_N_Bridge_PW})), provided that (\ref{eq:DiagnosticModelLambda}) holds, we can assess whether they balance the $Z$ distribution of the $G=1$ sample according to that of the artificial sample $G=\blacktriangle$ (i.e., the $G=1$ sample where $Z$ is randomly drawn from a model for $\lambda_0^Z$). This idea is in the same spirit of the balance proposals discussed in \citet{nguyen2023causal}. A diagnostic to check the weights ability to balance covariate means according to their intended target distributions, as recommended in \citet{chattopadhyay2020balancing}), is
\begin{equation}
   \frac{\sum_i\hat{\omega}_i\times \textup{I}(analytic_i)X_i}{\sum_i \hat{\omega}_i\times \textup{I}(analytic_i)}=\frac{\sum_i\textup{I}(target_i) X_i}{\sum_i \textup{I}(target_i)} \label{eq:TargetBalance}
\end{equation}
where $X_i$ is a covariate for individual $i$, $\textup {I}(\cdot)$ is the indicator function, and $analytic_i$ and $target_i$ refer to the individual's membership in the analytic and target populations, respectively. To compare across covariates, (\ref{eq:TargetBalance}) could be divided by $X$'s \emph{unweighted} standard deviation in the target population.

\section{Simulation Study} \label{sec:Simulation_Study}

\subsection{Design}

We performed a simulation study to examine the consistency and efficiency for $\theta_1^*$ at large (n=5,000) and small (n=500) sample sizes, for continuous and binary outcomes. We also examined the robustness of SWR estimators under mis-specified nuisance models. We omitted estimators for $\theta_g$, the group-specific mean in the observation arm, whose statistical properties are known \citep{lunceford2004stratification,gabriel2024inverse}.

The data generating model (DGM) for the observed data was informed by the National Health and Nutrition Examination Survey \citep{nhanes2017} and the clinical literature \citep{fontil2015simulating}, adapting the procedure of \citet{chang2024importance}. We generated systolic blood pressure $Y$ at follow-up as a continuous outcome, and dichotomized it as $W$ for a binary outcome. These outcomes were dependent on covariates considered outcome-allowable (age $A_{y1}$ and sex $A_{y2}$; continuous and binary), intervention-allowable (baseline systolic blood pressure $A_{z1}$ and diabetes $A_{z2}$; continuous and binary), and non-allowable (educational attainment $N_1$ and insurance status $N_2$; both binary). The point of intervention $Z$ (treatment intensification) and race $G$ were both binary. The DGM contained dependencies between these variables, as well as heterogeneous effects across levels of $G$. In the model for the outcome $Y$, the effect of treatment intensification $Z$ depended on systolic blood pressure $A_{z1}$ and group membership $G$. Our DGM reflected standard clinical practice where antihypertensive treatment is only intensified when blood pressure is elevated (e.g., systolic $\geq$ 140 mm Hg), which is a deterministic relationship in the true model for $Z$. 

As described in section \ref{subsec:HypotheticalDesign}, we cast our estimand in a study design where (i) each social group $G=1$ and $G=0$ is enrolled so that its outcome-allowables are distributed as in a standard population denoted $T=1$ (ii) after enrollment, persons are randomized to the observation arm (that leaves $Z$ alone) or the intervention arm (that modifies the distribution of $Z$ to follow $\P_0(Z|A_zA_y)$ when $G=1$). We chose the $G=1$ group to represent the standard population (i.e., $G=1 \Longrightarrow T=1$). We obtained the true values of $\theta_1$ numerically by (i) modifying the distributions of $A_y$ and $Z$ in the DGM to reflect their distributions in both arms of this hypothetical study design (ii) generating a large sample for each arm ($n=4\times 10^9$) and (iii) taking group-specific means in the sample generated for each arm. (See the Supplemental Material for further details on the DGM).

We specified all estimators' nuisance models with correct models following the DGM, flexible models including higher order terms and two-way interactions between selected covariates, and incorrect models excluding higher order terms and key covariates $A_{Y1}$, $A_{Z1}$, and $N_2$. For the interventional weights $\omega_{11^*}^{(Z,N)}$, we used (\ref{eq:Weight_Z_Model_PW}) for Z-Model-PW/SWR, (\ref{eq:Weight_N_Model_PW}) for N-Model-PW/SWR, (\ref{eq:Weight_Z_Bridge_PW_2}) for Z-Bridge-PW/SWR, and (\ref{eq:Weight_N_Bridge_PW}) for N-Bridge-PW/SWR, with (\ref{eq:Weight_N_Model_PW}), (\ref{eq:Weight_Z_Bridge_PW_2}), and (\ref{eq:Weight_N_Bridge_PW}) expressed in odds terms. For the large sample size (n=5,000) with correct or flexible nuisance models for $\omega_{11^*}^{(Z,N)}$, we respected the true model for $Z$ by estimating $\omega_{11^*}^{(Z,N)}$ among those with $A_{Z1}\geq140$ and setting  $\omega_{11^*}^{(Z,N)}$ to $1$ when $A_{Z1}<140$, and by setting $Z=0$ in the artificial sample $G=\blackdiamond$ when $A_{Z1}<140$. For the bridging estimators, we created the artifical samples by assignment (see sections \ref{subsec:Estimators_Z_Bridge} and \ref{subsec:Estimators_N_Bridge}). We ran 1,000 simulations and estimated the bias, standard error (SE), root mean squared error (RMSE), and 95\% confidence interval coverage (via the non-parametric bootstrap with 1,000 replicates) for each scenario.

\subsection{Results}

Table 1 shows the estimator performance when all nuisance models are correct or flexible. See Table \ref{tab:AllEstimators_Omit} for their performance when all nuisance models are incorrect, and Table \ref{tab:AllEstimators_CorrFlex_RW} for RW estimators. The estimators were consistent for $\theta_1^*$ across sample size and outcome types under correct or flexible nuisance models. The linear E-OBD estimator performed well for the continuous outcome (despite its inability to specify the $Z\times A_z$ interaction in the true $Y$ model), but suffered degraded consistency and coverage for the binary outcome.

At the smaller sample size, comparing estimators by their status as PW, RW, SR, or SWR estimators, the SWR estimators were the least biased, while the SR estimators were slightly more biased but had coverage closer to the nominal rate. The RW were the most biased, followed by the PW estimators, both of which had anti-conservative coverage. Comparing estimators by their status as modeling or bridging the $Z$ or $N$ distribution, the modeling strategies were less efficient with anti-conservative coverage. The N-bridging estimators dominated all other strategies in their respective classes, achieving the lowest bias and most accurate coverage. The Z-Bridge-SWR and N-Bridge-SWR estimators had the best performance across all estimators.

Table 2 shows the robustness of the SWR estimators under various scenarios of nuisance model mis-specification that nonetheless achive consistent estimation (as outlined in sections \ref{subsubsec:Estimator_Z_Model_SWR}, \ref{subsubsec:Estimator_N_Model_SWR}, \ref{subsubsec:Estimator_Z_Bridge_SWR}, and \ref{subsubsec:Estimator_N_Bridge_SWR}). Our results confirm the robustness properties of the SWR estimators across sample size and outcome types. Comparing across estimators, the N-Model SWR estimator was the most biased and showed worse coverage as dependence on consistent weight estimation increased. Overall, the Z-Model-SWR and N-Bridge-SWR estimators dominated the Z-Bridge-SWR estimator with lower bias and more accurate coverage.

\section{Data Application} \label{sec:DataApplication}

We analyzed data from 79,898 primary care visits of 25,287 adult patients seen in the Johns Hopkins Community Physicians (JHCP) health system from 2018 to 2022. The cohort was assembled to emulate a target trial  where each trial’s “time zero” was anchored on bi-weekly calendar intervals. Patients could enroll in a trial (including multiple trials) whenever they met the following eligibility criteria: age 18 years or older, self-reported race as measured in the electronic medical record as “Black” or “White”, current office visit with elevated blood pressure (BP) readings (systolic BP$\geq$ 140 mm Hg or diastolic BP $\geq$ 90 mm Hg), no office visit with elevated BP readings in the prior two weeks but 2+ primary care office visits with elevated BP readings in the prior year, recorded diagnosis of primary hypertension in the prior year, were not pregnant or diagnosed with end stage renal disease, liver disease, dementia, or metastatic tumor in the prior year, and all allowable covariates were measured by the time of the current visit defining enrollment into the trial. We considered age and sex as outcome allowable $A_y$. We considered body mass index, chronic kidney disease, diabetes, cardiovascular disease, congestive heart failure, baseline numbers of antihypertensive medications and medication classes, current systolic and diastolic BP as intervention-allowable $A_z$, as these reflect relevant factors used by clinical guidelines to determine antihypertensive medication regimens \citep{whelton20182017}. As non-allowable covariates $N$ we considered the number of prior primary care visits, type of health insurance, and the national percentile of the area deprivation index, all measured at baseline. Our point of intervention $Z$ was treatment intensification, which was measured as an initiation, increase in dose, or addition of a new class of an antihypertensive medication within the two weeks following the current visit, the date of the visit marking the start of follow-up. We measured our outcomes systolic BP $Y$ and uncontrolled hypertension $W$ (systolic BP $\geq$ 140 mm Hg or diastolic BP $\geq$ 90 mm Hg) by taking BP measures from the closest visit to the 6-month follow-up mark between 4.5 and 7.5 months. Further details on the dataset and its construction are available in \citet{meche2026target}.

\subsection{Analysis}

We defined the Black population as the historically disadvantaged group $G=1$ and the standard population $T=1$, and the White population as the historically advantaged group $G=0$. We used the weighting estimator of section \ref{sec:Estimators_Observation_Arm} to measure the disparity in antihypertensive treatment intensification, adjusting for all allowable covariates. We estimated $\theta_1^*$ the mean systolic BP and proportion of uncontrolled hypertension for $G=1$ in the intervention arm of our target trial using the estimators described in section \ref{sec:Estimators_Intervention_Arm}. We estimated $\theta_1$ and $\theta_0$ the group-specific mean outcomes of the observation arm using the estimators described in section \ref{sec:Estimators_Observation_Arm}, and calculated the reduced disparity for each outcome. When implementing these estimators, we included cubic splines for age and systolic BP, and (except for E-OBD) interactions between treatment intensification and (i) treatment status at baseline as well as (ii) number of concomitant treatment classes at baseline. We used multinomial models to estimate the factored joint distribution of the non-allowables, after discretizing the area deprivation index distribution into 10 evenly spaced categories across the response scale. To estimate the average means and disparities over 2018-2022, we adjusted for calendar time (via cubic splines) by including it along with outcome-allowable covariates $A_y$ as discussed in \citet{jackson2025target}. For the interventional weights $\omega_{11^*}^{(Z,N)}$, we used (\ref{eq:Weight_Z_Model_PW}) for Z-Model-PW/SWR, (\ref{eq:Weight_N_Model_PW}) for N-Model-PW/SWR, (\ref{eq:Weight_Z_Bridge_PW_2}) for Z-Bridge-PW/SWR, and (\ref{eq:Weight_N_Bridge_PW}) for N-Bridge-PW/SWR, with (\ref{eq:Weight_N_Model_PW}), (\ref{eq:Weight_Z_Bridge_PW_2}), and (\ref{eq:Weight_N_Bridge_PW}) expressed in odds terms. For the bridging estimators, we created the artificial samples by sampling from uniform distributions of $Z$ or $N$. We addressed missing outcomes using the methods described in the Supplement. Because patients could enroll for each trial where they met eligibility, we accounted for within-person correlations by estimating 95\% confidence intervals obtained via a non-parametric cluster bootstrap \citep{field2007bootstrapping} that sampled persons with replacement and retained all observations of sampled persons. We assessed the quality of our nuisance models using the diagnostics outlined in section \ref{sec:Diagnostics}.

\subsection{Results}

Black persons received slightly less treatment intensification than White patients by 1.9\% (95\% CI 1.2\%, 2.7\%), after accounting for outcome- and intervention-allowable covariates. At follow up, Black and White patients had an average systolic BP of 142 and 139 mm Hg, respectively, with an observed disparity of 2.8 mm Hg (95\%CI 2.4, 3.3), after accounting for outcome-allowable covariates. With respect to uncontrolled hypertension, 53\% and 48\% of Black and White patients, respectively, had uncontrolled hypertension, with an observed disparity of 5.5\% (95\%CI 4.3\%, 6.6\%). Table \ref{tab:Data_Application} shows the estimated results and 95\% confidence intervals for outcome mean in the intervention arm of the envisioned target trial, along with the reduced disparity and the residual disparity. The diagnostics for the critical density models of $\lambda_0^Z$ and $\lambda_1^N$ (data not shown) and the weights (see Supplement Figures \ref{fig:Weight_Ay_1T}-\ref{fig:Weight_C_0_G0}) suggested that the nuisance models for the estimators were adequately specified. Interestingly, the varaibility of the interventional weights $\omega^{(Z,N)}_{11^*}$ was substantially greater for the forms used in the modeling strategies compared to the forms used in the bridging strategies (see Table \ref{tab:Weight_Distribution}).  The results were similar across all estimators, showing outcome means of of 142 mm Hg systolic BP and a proportion with uncontrolled hypertension of 53\% among Black patients under the hypothetical intevention, and a negligible reduction in the observed disparity in systolic BP and in uncontrolled hypertension. The estimators’ confidence interval width were wider for the linear and modeling estimators. The small effect size likely stems from the small disparity in treatment intensification. Overall, the results suggest that, in this healthcare system, an intervention to remove disparities in treatment intensification would have little impact on hypertension control disparities.

\section{Conclusion} \label{sec:Conclusion}

In this paper, we toured estimation strategies for CDA, including existing methods and introduced novel estimators that overcome key challenges, and provided diagnostics for their implementation. We examined estimators that model either the distribution of $Z$ (the point of intervention) or $N$ (the non-allowable covariates). We also introduced estimators that build `bridge' samples to avoid modeling the distributions of $Z$ or $N$. For each estimator class, we presented or developed pure weighting (PW), sequential regression (SR), and multiply robust Weighted SR (SWR) estimators (and non-robust Regress-then-Weight estimators (RW), see the Supplement). We showed that these estimators are consistent and that they relate to the influence function of our estimand's identifying formula. We also showed that certain estimators, under simplifying circumstances, reduce to causal versions of statistical decomposition estimators that are widely used in the applied research community.

Here, we provide some guidance for choosing between estimators. The simplicity of the linear estimator is valuable pedagogically, but its assumptions are overly restrictive. Of the more flexible estimators, the Z-Modeling estimators are very intuitive, and require fitting the fewest number of models. They may be advantageous when one has adequate substantive knowledge to correctly model $\lambda_0^Z$, the interventional distribution of the estimand (the critical density of $Z$), which may be limited to scenarios where $Z$ is binary. If chosen, we recommend using the diagnostic for this distribution (\ref{eq:DiagnosticModelLambda}) to support their implementation. We do not prioritize the use of the N-modeling estimators (over their alternatives) whose inferior performance likely reflects the challenge of modeling the joint distribution of covariates. Both bridging estimators avoid modeling any density, but the N-bridging estimators require fewer models, and in our simulation results they dominated every other estimator in their respective subclass. Of the SWR estimators, the N-bridging estimator had the best overall performance across scenarios. The N-bridging estimators therefore offer a very attractive option for routine estimation, especially when dealing with continuous or multivariate points of intervention, but the Z-bridging estimators remain a viable option. Within a given estimation strategy, if the weights can be estimated well and are not highly variable, the SWR estimators could serve as a first-line option. Otherwise, the SR estimators would be the next logical choice. The PW estimators or their RW counterparts can be used when the outcome is difficult to model or one wishes to fit a specialized model such as Cox proportional hazards or quantile regression.   

Our multiply robust SWR estimators do not make use of machine learning tools that are advocated for when covariates or their interactions are high-dimensional. They could be incorporated via Cross-validated Targeted Minimum Loss Estimation \citep{van2011cross,diaz2020machine}, which satisfies the mean recovery condition for outcome models. Alternatively, augmented balancing weights that leverage covariate balance conditions \citep{chattopadhyay2020balancing,ben2021balancing}, or emerging techniques for density ratio estimation \citep{hines2025learning}, could be used to estimate the weights directly. Each of these proposals can accommodate machine learning but this is left for future work. Importantly, unlike current AIPW-based proposals for CDA (\citep{lundberg2024gap,yu2025nonparametric,park2025causal}), our SWR estimators readily generalize to continuous and multivariate points of intervention. Furthermore, the sensitivity of machine-learning enabled AIPW estimators to random-seed selection, their performance in finite samples, and the optimal choice of folds in their required sample splitting and cross-validation procedures remains an active area of research \citep{chernozhukov2018double,naimi2023challenges,schader2024don,naimi2024pseudo,zivich2024commentary,ahrens2025introduction}. The robustness of our SWR estimators can be enhanced by guiding their implementation with our proposed diagnostics for modeled densities and weighting functions. Moreover, we provided different forms of the interventional weights $\omega_{11^*}^{(Z,N)}$ used to fit the initial $\mu_1$ outcome regressions. As the forms rely on different nuisance models, they can be swapped across SWR estimators, offering additional avenues for robustness.

\citet{steiner2024robust} discusses proper variable selection for multiply robust estimators. We advise users to include the same respective set of covariates for outcome-allowables $A_y$, intervention-allowables $A_z$, and non-allowables $N$ in whichever nuisance models they appear. If investigators vary $N$, perhaps out of necessity such as in high-dimensional settings, they should ensure that the covariate set for $N$ used in the outcome regressions can overcome any biases not addressed by the covariate set for $N$ used in the weighting functions and vice versa. See \citet{steiner2024robust} for further details. Sensitivity analysis for unmeasured confounding in CDA is also important avenue to assess robustness to variable selection \citep{park2023sensitivity,shen2025calibrated}.

Regarding statistical inference, all of our estimators are M-estimators \citep{stefanski2002calculus}, and therefore are asymptotically normal and are amenable to bootstrap procedures, including ones developed for clustered data \citep{field2007bootstrapping} or very large datasets \citep{kleiner2014scalable}. Our simulation results suggest that the standard errors for our multiply robust estimators, when obtained using bootstrap methods or M-estimation theory, may also enjoy a multiple robustness property \citep{shook2025double,wu2025doublerobustestimatorcausalinference}. They may also achieve the non-parametric efficiency bound, as weighted regression estimators for the average treatment effect do generally \citep{gabriel2024inverse}. Proof of these conjectures for our estimators is left for future work.

\section{Author Contributions}
Conception (JWJ); development of estimators (JWJ \& TQN); derivation of influence function \& proof of SWR estimator robustness (TQN); simulation study (JWJ \& TC); data application (JWJ \& AM); drafting of manuscript (JWJ); editing for critical scientific content (TC, AM, TQN).

\section{Funding Statement}
Research reported in this publication was supported by the National Heart, Lung, and Blood Institute of the NIH and the NIH Office of the Director under awards K01HL145320 and R01HL169956.

\newpage

\bibliography{refs}

\newpage

\section{Tables}

\begin{table}[htpb!]
  \caption{Estimator performance under correct or flexible nuisance models by outcome type and sample size.}
  \centering
    \begin{tabular}{lrrrrrrrrr}
\hline    
    \multicolumn{10}{c}{Sample Size = 5,000} \\
\hline    
          & \multicolumn{4}{c}{Continuous ($\theta_1^*=132.606$)} &       & \multicolumn{4}{c}{Binary ($\theta_1^*=.245$)} \\
          & \multicolumn{1}{l}{Bias} & \multicolumn{1}{l}{SE} & \multicolumn{1}{l}{RMSE} & \multicolumn{1}{l}{Coverage} &       & \multicolumn{1}{l}{Bias} & \multicolumn{1}{l}{SE} & \multicolumn{1}{l}{RMSE} & \multicolumn{1}{l}{Coverage} \\
\hline          
    Linear (E-OBD)$^\dagger$ & 0.020 & 0.273 & 0.273 & 0.940 &       & 0.014 & 0.010 & 0.018 & 0.717 \\
    Pure Weighting (PW) &       &       &       &       &       &       &       &       &  \\
    \multicolumn{1}{r}{Z-Model-PW} & 0.013 & 0.269 & 0.270 & 0.937 &       & 0.001 & 0.010 & 0.010 & 0.956 \\
    \multicolumn{1}{r}{N-Model-PW} & -0.004 & 0.270 & 0.270 & 0.936 &       & 0.000 & 0.010 & 0.010 & 0.959 \\
    \multicolumn{1}{r}{Z-Bridge-PW} & -0.019 & 0.269 & 0.269 & 0.932 &       & -0.001 & 0.010 & 0.010 & 0.950 \\
    \multicolumn{1}{r}{N-Bridge-PW} & -0.003 & 0.270 & 0.270 & 0.940 &       & 0.000 & 0.010 & 0.010 & 0.959 \\
    \multicolumn{10}{l}{Sequential Regressions (SR)} \\
    \multicolumn{1}{r}{Z-Model-SR} & 0.013 & 0.269 & 0.269 & 0.942 &       & 0.000 & 0.010 & 0.010 & 0.961 \\
    \multicolumn{1}{r}{N-Model-SR} & -0.086 & 0.267 & 0.281 & 0.927 &       & -0.004 & 0.010 & 0.010 & 0.936 \\
    \multicolumn{1}{r}{Z-Bridge-SR} & -0.017 & 0.267 & 0.267 & 0.941 &       & 0.000 & 0.010 & 0.010 & 0.953 \\
    \multicolumn{1}{r}{N-Bridge-SR} & -0.018 & 0.267 & 0.267 & 0.939 &       & 0.000 & 0.010 & 0.010 & 0.953 \\
    \multicolumn{10}{l}{Sequential Weighted Regressions (SWR)} \\
    \multicolumn{1}{r}{Z-Model-SWR} & 0.014 & 0.272 & 0.272 & 0.940 &       & 0.001 & 0.010 & 0.010 & 0.964 \\
    \multicolumn{1}{r}{N-Model-SWR} & -0.062 & 0.270 & 0.277 & 0.931 &       & -0.003 & 0.010 & 0.010 & 0.950 \\
    \multicolumn{1}{r}{Z-Bridge-SWR} & 0.008 & 0.270 & 0.270 & 0.936 &       & 0.001 & 0.010 & 0.010 & 0.953 \\
    \multicolumn{1}{r}{N-Bridge-SWR} & 0.007 & 0.270 & 0.270 & 0.937 &       & 0.001 & 0.010 & 0.010 & 0.954 \\
\hline
    \multicolumn{10}{c}{Sample Size = 500} \\
\hline
          & \multicolumn{4}{c}{Continuous ($\theta_1^*=132.606$)} &       & \multicolumn{4}{c}{Binary ($\theta_1^*=.245$)} \\
          & \multicolumn{1}{l}{Bias} & \multicolumn{1}{l}{SE} & \multicolumn{1}{l}{RMSE} & \multicolumn{1}{l}{Coverage} &       & \multicolumn{1}{l}{Bias} & \multicolumn{1}{l}{SE} & \multicolumn{1}{l}{RMSE} & \multicolumn{1}{l}{Coverage} \\
\hline          
    Linear (E-OBD)$^\dagger$ & 0.009 & 0.838 & 0.837 & 0.954 &       & 0.013 & 0.032 & 0.035 & 0.937 \\
    \multicolumn{10}{l}{Pure Weighting (PW)} \\
    \multicolumn{1}{r}{Z-Model-PW} & 0.198 & 1.654 & 1.665 & 0.983 &       & -0.008 & 0.048 & 0.049 & 0.982 \\
    \multicolumn{1}{r}{N-Model-PW} & -0.139 & 0.918 & 0.929 & 0.990 &       & 0.002 & 0.037 & 0.037 & 0.983 \\
    \multicolumn{1}{r}{Z-Bridge-PW} & -0.202 & 0.824 & 0.848 & 0.948 &       & 0.002 & 0.032 & 0.032 & 0.964 \\
    \multicolumn{1}{r}{N-Bridge-PW} & -0.025 & 0.870 & 0.870 & 0.969 &       & 0.001 & 0.034 & 0.034 & 0.957 \\
    \multicolumn{10}{l}{Sequential Regressions (SR)} \\
    \multicolumn{1}{r}{Z-Model-SR} & 0.030 & 0.838 & 0.838 & 0.964 &       & 0.002 & 0.033 & 0.033 & 0.960 \\
    \multicolumn{1}{r}{N-Model-SR} & -0.127 & 0.816 & 0.825 & 0.949 &       & -0.004 & 0.032 & 0.032 & 0.955 \\
    \multicolumn{1}{r}{Z-Bridge-SR} & -0.050 & 0.809 & 0.810 & 0.952 &       & -0.001 & 0.032 & 0.032 & 0.951 \\
    \multicolumn{1}{r}{N-Bridge-SR} & -0.050 & 0.809 & 0.810 & 0.952 &       & -0.001 & 0.032 & 0.032 & 0.949 \\
    \multicolumn{10}{l}{Sequential Weighted Regressions (SWR)} \\
    \multicolumn{1}{r}{Z-Model-SWR} & 0.030 & 0.855 & 0.855 & 0.964 &       & 0.003 & 0.035 & 0.035 & 0.970 \\
    \multicolumn{1}{r}{N-Model-SWR} & -0.093 & 0.833 & 0.838 & 0.975 &       & -0.005 & 0.034 & 0.034 & 0.980 \\
    \multicolumn{1}{r}{Z-Bridge-SWR} & -0.020 & 0.819 & 0.819 & 0.953 &       & -0.001 & 0.033 & 0.033 & 0.961 \\
    \multicolumn{1}{r}{N-Bridge-SWR} & -0.020 & 0.822 & 0.822 & 0.955 &       & -0.001 & 0.033 & 0.033 & 0.956 \\

\hline    
    \multicolumn{10}{l}{\footnotesize{$^\dagger$The true outcome model contained a $Z\times A_z$ interaction. The outcome model in Linear (E-OBD) omitted $Z\times A_z$ interaction.}} \\
    \end{tabular}%
\label{tab:AllEstimators_CorrFlex}
\end{table}%

\begin{table}[htpb!]
  \caption{Robustness of SWR estimators under nuisance specification scenarios by outcome and sample size.}
  \centering
      \begin{tabular}{rllllllllrrrr}
      \hline
      \multicolumn{13}{c}{Sample Size = 5,000} \\
      \hline
          & \multicolumn{8}{c}{Nuisance Specification Scenario} & \multicolumn{2}{c}{Con. ($\theta_1^*=132.606$)} & \multicolumn{2}{c}{Bin. ($\theta_1^*=.245$)} \\
          & $\mu$ & $\zeta$ & $\nu$ & $\kappa$ & $\omega^{(Z,N)}_{11^*}$ & $\omega^{Z}_{\blackdiamond 1^*}$ & $\omega^{(N,A_z)}_{\diamond 1^*}$ & $\omega^{A_z}_{01}$ & \multicolumn{1}{l}{Bias} & \multicolumn{1}{l}{Coverage} & \multicolumn{1}{l}{Bias} & \multicolumn{1}{l}{Coverage} \\
     \hline
    \multicolumn{1}{r}{Z-Model-SWR} & \cmark &       &       &       & \xmark &       &       &       & 0.014 & 0.940 & 0.001 & 0.960 \\
          & \xmark &       &       &       & \cmark &       &       &       & 0.013 & 0.944 & 0.000 & 0.956 \\
          &       &       &       &       &       &       &       &       &       &       &       &  \\
    \multicolumn{1}{r}{N-Model-SWR} & \cmark &       &       & \cmark & \xmark &       & & \xmark        & -0.087 & 0.929 & -0.004 & 0.940 \\
          & \cmark &       &       & \xmark & \xmark &       & & \cmark        & -0.158 & 0.899 & -0.007 & 0.908 \\
          & \xmark &       &       & \cmark & \cmark &       & & \xmark        & -0.266 & 0.854 & -0.011 & 0.840 \\
          & \xmark &       &       & \xmark & \cmark &       & & \cmark        & -0.334 & 0.788 & -0.013 & 0.774 \\
          &       &       &       &       &       &       &       &       &       &       &       &  \\
    \multicolumn{1}{r}{Z-Bridge-SWR} & \cmark & \cmark &       & \cmark & \xmark & \xmark &       & \xmark & -0.016 & 0.937 & 0.000 & 0.950 \\
          & \cmark & \cmark &       & \xmark & \xmark & \xmark &       & \cmark & -0.088 & 0.929 & -0.003 & 0.938 \\
          & \cmark & \xmark &       & \cmark & \xmark & \cmark &       & \xmark & -0.079 & 0.933 & 0.006 & 0.911 \\
          & \cmark & \xmark &       & \xmark & \xmark & \cmark &       & \cmark & -0.149 & 0.902 & 0.003 & 0.943 \\
          & \xmark & \cmark &       & \cmark & \cmark & \xmark &       & \xmark & -0.024 & 0.934 & -0.001 & 0.954 \\
          & \xmark & \cmark &       & \xmark & \cmark & \xmark &       & \cmark & -0.095 & 0.926 & -0.004 & 0.931 \\
          & \xmark & \xmark &       & \cmark & \cmark & \cmark &       & \xmark & -0.028 & 0.933 & -0.001 & 0.954 \\
          & \xmark & \xmark &       & \xmark & \cmark & \cmark &       & \cmark & -0.097 & 0.928 & -0.004 & 0.931 \\
          &       &       &       &       &       &       &       &       &       &       &       &  \\
    \multicolumn{1}{r}{N-Bridge-SWR} & \cmark &       & \cmark &       & \xmark &       & \xmark &       & -0.010 & 0.943 & 0.000 & 0.954 \\
          & \cmark &       & \xmark &       & \xmark &       & \cmark &       & 0.058 & 0.940 & -0.002 & 0.948 \\
          & \xmark &       & \cmark &       & \cmark &       & \xmark &       & 0.007 & 0.944 & 0.000 & 0.959 \\
          & \xmark &       & \xmark &       & \cmark &       & \cmark &       & 0.004 & 0.940 & -0.004 & 0.939 \\
          &       &       &       &       &       &       &       &       &       &       &       &  \\
    \hline
    \multicolumn{13}{c}{Sample Size = 500} \\
   \hline

          & \multicolumn{8}{c}{Nuisance Specification Scenario} & \multicolumn{2}{c}{Con. ($\theta_1^*=132.606$)} & \multicolumn{2}{c}{Bin. ($\theta_1^*=.245$)} \\
          & $\mu$ & $\zeta$ & $\nu$ & $\kappa$ & $\omega^{(Z,N)}_{11^*}$ & $\omega^{Z}_{\blackdiamond 1^*}$ & $\omega^{(N,A_z)}_{\diamond 1^*}$ & $\omega^{A_z}_{01}$ & \multicolumn{1}{l}{Bias} & \multicolumn{1}{l}{Coverage} & \multicolumn{1}{l}{Bias} & \multicolumn{1}{l}{Coverage} \\
    \hline
    \multicolumn{1}{r}{Z-Model-SWR} & \cmark &       &       &       & \xmark &       &       &       & 0.029 & 0.962 & 0.004 & 0.966 \\
          & \xmark &       &       &       & \cmark &       &       &       & -0.171 & 0.974 & -0.017 & 0.950 \\
          &       &       &       &       &       &       &       &       &       &       &       &  \\
    \multicolumn{1}{r}{N-Model-SWR} & \cmark &       &       & \cmark & \xmark &       & & \xmark        & -0.121 & 0.954 & -0.005 & 0.959 \\
          & \cmark &       &       & \xmark & \xmark &       & & \cmark        & -0.209 & 0.956 & -0.008 & 0.964 \\
          & \xmark &       &       & \cmark & \cmark &       & & \xmark        & -0.265 & 0.989 & -0.005 & 0.977 \\
          & \xmark &       &       & \xmark & \cmark &       & & \cmark        & -0.371 & 0.983 & -0.008 & 0.979 \\
          &       &       &       &       &       &       &       &       &       &       &       &  \\
    \multicolumn{1}{r}{Z-Bridge-SWR} & \cmark & \cmark &       & \cmark & \xmark & \xmark &       & \xmark & -0.051 & 0.952 & -0.001 & 0.952 \\
          & \cmark & \cmark &       & \xmark & \xmark & \xmark &       & \cmark & -0.140 & 0.956 & -0.005 & 0.957 \\
          & \cmark & \xmark &       & \cmark & \xmark & \cmark &       & \xmark & -0.207 & 0.950 & 0.008 & 0.956 \\
          & \cmark & \xmark &       & \xmark & \xmark & \cmark &       & \cmark & -0.319 & 0.944 & 0.004 & 0.961 \\
          & \xmark & \cmark &       & \cmark & \cmark & \xmark &       & \xmark & -0.082 & 0.961 & 0.006 & 0.954 \\
          & \xmark & \cmark &       & \xmark & \cmark & \xmark &       & \cmark & -0.193 & 0.952 & 0.002 & 0.958 \\
          & \xmark & \xmark &       & \cmark & \cmark & \cmark &       & \xmark & -0.095 & 0.961 & 0.006 & 0.961 \\
          & \xmark & \xmark &       & \xmark & \cmark & \cmark &       & \cmark & -0.206 & 0.956 & 0.002 & 0.961 \\
          &       &       &       &       &       &       &       &       &       &       &       &  \\
    \multicolumn{1}{r}{N-Bridge-SWR} & \cmark &       & \cmark &       & \xmark &       & \xmark &       & -0.046 & 0.954 & -0.001 & 0.950 \\
          & \cmark &       & \xmark &       & \xmark &       & \cmark &       & 0.046 & 0.969 & -0.002 & 0.966 \\
          & \xmark &       & \cmark &       & \cmark &       & \xmark &       & 0.070 & 0.954 & 0.008 & 0.953 \\
          & \xmark &       & \xmark &       & \cmark &       & \cmark &       & 0.030 & 0.968 & 0.002 & 0.953 \\
    \hline 
    \multicolumn{13}{l}{\footnotesize{Nuisance specification denoted by \cmark ~(flexible) or \xmark ~(omitted interactions, higher order terms, and certain covariates in $(N_2,A_{z1},A_{y1})$).}} \\    
    \multicolumn{13}{l}{\footnotesize{Flexible specification was used for $\eta_1^*$ and $\omega_{gT}^{A_y}$ (all estimators) and $\lambda_0^Z$ (for Z-Model-SWR) and $\lambda_1^N$ (for N-Model-SWR).}} \\

    \end{tabular}%
  \label{tab:SWR_Robustness}%
\end{table}%

\begin{table}[htbp!]
  \centering
    \caption{Counterfactual mean and reduced disparity in the JHCP health system over 2018-2022, by estimator}
    \begin{tabular}{lrrrrr}
\hline    
& \multicolumn{2}{c}{Systolic Blood Pressure$^\dagger$} & & \multicolumn{2}{c}{Uncontrolled Hypertension (\%){$^\ddagger$}} \\
          & Counterfactual Mean & Reduced Disparity &       & Counterfactual Mean & Reduced Disparity  \\
\hline          
    Linear (E-OBD) & 141.8 (141.5 ,  142.1) & -0.1 (-0.2 ,  0.1)&       & 53.5 (52.7 ,  54.3) & 0.0 (-0.3 ,  0.3)  \\
    \multicolumn{6}{l}{Pure Weighting (PW)} \\
    \multicolumn{1}{r}{Z-Model-PW} & 141.8 (141.4 ,  142.1) & -0.1 (-0.2 ,  0.1) &       & 53.2 (52.4 ,  54.0) & -0.2 (-0.6 ,  0.2) \\
    \multicolumn{1}{r}{N-Model-PW} & 141.8 (141.5 ,  142.1) & 0.0 (-0.1 ,  0.1) &       & 53.4 (52.6 ,  54.1) & 0.0 (-0.3 ,  0.2)  \\
    \multicolumn{1}{r}{Z-Bridge-PW} & 141.8 (141.5 ,  142.1) & 0.0 (-0.1 ,  0.0) &       & 53.4 (52.6 ,  54.1) & -0.1 (-0.2 ,  0.0)  \\
    \multicolumn{1}{r}{N-Bridge-PW} & 141.8 (141.5 ,  142.1) & -0.1 (-0.1 ,  0.0) &       & 53.3 (52.6 ,  54.0) & -0.1 (-0.2 ,  0.0)  \\
    \multicolumn{6}{l}{Sequential Regressions (SR)} \\
    \multicolumn{1}{r}{Z-Model-SR} & 141.8 (141.4 ,  142.1) & -0.1 (-0.1 ,  0.0) &       & 53.3 (52.5 ,  54.0) & -0.1 (-0.2 ,  0.0)  \\
    \multicolumn{1}{r}{N-Model-SR} & 141.8 (141.4 ,  142.1) & 0.0 (-0.1 ,  0.0) &        & 53.3 (52.5 ,  54.0) & -0.1 (-0.3 ,  0.1)  \\
    \multicolumn{1}{r}{Z-Bridge-SR} & 141.8 (141.4 ,  142.1) & -0.1 (-0.1 ,  0.0)  &       & 53.3 (52.5 ,  54.0) & -0.1 (-0.2 ,  0.0)  \\
    \multicolumn{1}{r}{N-Bridge-SR} & 141.8 (141.4 ,  142.1) & -0.1 (-0.1 ,  0.0)  &       & 53.3 (52.5 ,  54.0) & -0.1 (-0.2 ,  0.0)  \\
    \multicolumn{6}{l}{Sequential Weighted Regressions (SWR)} \\
    \multicolumn{1}{r}{Z-Model-SWR} & 141.8 (141.4 ,  142.1) & 0.0 (-0.2 ,  0.1)  &       & 53.3 (52.5 ,  54.1) & -0.1 (-0.5 ,  0.3)  \\
    \multicolumn{1}{r}{N-Model-SWR} & 141.8 (141.4 ,  142.1) & -0.1 (-0.2 ,  0.0)  &       & 53.2 (52.4 ,  53.9) & -0.2 (-0.6 ,  0.1) \\
    \multicolumn{1}{r}{Z-Bridge-SWR} & 141.8 (141.5 ,  142.1) & -0.1 (-0.1 ,  0.0)  &       & 53.3 (52.5 ,  54.0) & -0.1 (-0.2 ,  0.0)  \\
    \multicolumn{1}{r}{N-Bridge-SWR} & 141.8 (141.5 ,  142.1) & 0.0 (-0.1 ,  0.0)  &       & 53.4 (52.6 ,  54.1) & 0.0 (-0.1 ,  0.1) \\
\hline    
\multicolumn{6}{l}{\footnotesize{{$^\dagger$}The observed disparity for systolic blood pressure was 2.8 mm Hg (95\%CI 2.4, 3.3)}} \\
\multicolumn{6}{l}{\footnotesize{{$^\ddagger$}The observed disparity for uncontrolled hypertension was 5.5\% (95\%CI 4.3\%, 6.6\%)}} 
\end{tabular}%
\label{tab:Data_Application}%
\end{table}%

\clearpage

\appendix

\renewcommand{\thetable}{S\arabic{table}}
\renewcommand{\thefigure}{S\arabic{figure}}
\setcounter{table}{0}

\section{Supplemental Material}

\subsection{Regress-then-Weight (RW) Estimators}

Another weighted average targeting $\theta_1^*$ suggests a weighting function that takes covariate data from the $G=0$ group and morphs the observed density $\P_0(Z,N,A_z|A_y)$ to that of the $G=1$ group in the intervention arm $\P_1^*(Z,N,A_z|A_y)$, and morphs the density of $A_y$ to that of the standard population, using $\mu_1(Z,N,A_z,A_y)$: 
\begin{align}
\theta_1^*=\E_0 \Big(\E_1[Y|Z,N,A_z,A_y]\times \frac{\P^*_1(Z,N,A_z|A_y)}{\P_0(Z,N,A_z|A_y)}\times\frac{\P_\std(A_y)}{\P_0(A_y)}\Big).
\end{align}
This approach forms the basis for some of the regress-then-weight approaches we now present. Unlike the Sequential Weighted Regression (SWR) estimators that also combine weighting and outcome regressions, these estimators are only consistent for $\theta_1^*$ when all of their nusiance models are correctly specified and weighting functions are consistently estimated. 

\subsubsection{Z-Model-RW}

A Z-Model-RW estimator of $\theta_1^*$ can be based on (\ref{eq:ID_Weighting}), taking a weighted mean among the $G=1$ group but replacing $Y$ with $\mu_1(Z,N,A_z,A_y)$. Thus, a Z-Model-RW estimator can be adapted from (\ref{eq:Estimator_Z_Model_PW}). It is consistent for $\theta_1^*$ if the model for $\mu_1$ is correctly specified and weights $\omega_{11^*}^{(Z,N)}$ (\ref{eq:Weight_Z_Model_PW}) and $\omega_{1T}^{A_y}$ (\ref{eq:Weight_Ay}) are consistently estimated.

Alternatively, as noted earlier, $\theta_1^*$ can be expressed as a weighted mean of $\mu_1(Z,N,A_z,A_y)$ in the $G=0$ group:
\begin{align}
    \theta_1^*
    &=\E_0\left[\omega_{01^*}^{(Z,N,A_z)}(Z,N,A_z,A_y)\omega_{0T}^{A_y}(A_y)\,\mu_1(Z,N,A_z,A_y)\right],\label{eq:ID_Weighting_Alt}
\end{align}
where the weighting function
\begin{align}
    \omega_{01^*}^{(Z,N,A_z)}(Z,N,A_z,A_y) &\coloneq \frac{\P_0(Z|A_z,A_y)\times\P_1(N|A_z,A_y)\times\P_1(A_z|A_y)}{\P_0(Z,N,A_z|A_y)}, 
\end{align}
shifts the distribution of $(Z,N,A_z)$ to that of the $G=1$ group in the intervention arm. This weighting function can be expressed in many ways. Here we use a factorization involving two pieces, $\omega_{01^*}^Z$ and $\omega_{01^*}^{(N,A_Z)}$
\begin{align}
    \omega_{01^*}^{(Z,N,A_z)}(Z,N,A_z,A_y) &= \underbrace{\frac{\P_0(Z|A_z,A_y)}{\P_0(Z|N,A_z,A_y)}}_{:=\omega_{01^*}^{Z}(Z,N,A_z,A_y)}
     \underbrace{\frac{\P_1(N|A_z,A_y)}{\P_0(N|A_z,A_y)}\frac{\P_1(A_z|A_y)}{\P_0(A_z|A_y)}}_{:=\omega_{01^*}^{(N,A_z)}(N,A_z,A_y)}\label{eq:Weight_Z_Model_RW}.
\end{align}
where $\lambda_0^Z$ appears in $\omega_{01^*}^{Z}$ and $\lambda_1^N$ appears in $\omega_{01^*}^{(N,A_Z)}$. Noting that $\omega_{01^*}^{(N,A_Z)}$ can be expressed as a ratio of odds functions
\begin{align}
\omega_{01^*}^{(N,A_z)}(N,A_z,A_y)= \frac{\odds(G=1\text{ vs }0|N,A_z,A_y)}{\odds(G=1\text{ vs }0|A_y)}, 
\end{align}
suggests a way to estimate $\omega_{01^*}^{(Z,N,A_z)}$ without modeling $\lambda_1^N$. Thus, this alternate Z-Model-RW estimator is the following weighted average of $\mu_1(Z,N,A_z,A_y)$ among the $G=0$ sample:
\begin{align}
    \hat\theta_{1,\text{Z-Model-RW-alt}}^*
    \coloneq\frac{\P_{0n}[\hat\omega_{01^*}^{(Z,N,A_z)}\hat\omega_{0T}^{A_y}\hat\mu_1(Z,N,A_z,A_y)]}{\P_{0n}[\hat\omega_{01^*}^{(Z,N,A_z)}\hat\omega_{0T}^{A_y}]}.
\end{align}
where the estimate of $\omega_{01^*}^{(Z,N,A,z)}$ is provided by the product of $\hat\omega_{01^*}^{Z}$ and $\hat\omega_{01^*}^{(N,A_Z)}$. It is consistent for $\theta_1^*$ if the model for $\mu_1$ is correctly specified and the weights $\omega_{01^*}^{(Z,N,A_z)}$ (\ref{eq:Weight_Z_Model_RW}) and $\omega_{0T}^{A_y}$ (\ref{eq:Weight_Ay}) are consistently estimated.

\subsubsection{N-Model-RW}

This estimator is based on the same weighted-outcome mean expression of $\theta_1^*$ in (\ref{eq:ID_Weighting_Alt}) that is the basis of the Z-Model-RW-alt. Here, we factorize the weighting function $\omega_{01^*}^{(Z,N,A_z)}$ into two terms $\omega_{01^*}^N$ and $\omega_{01}^{A_z}$ 
\begin{align}
    \omega_{01^*}^{(Z,N,A_z)}(Z,N,A_z,A_y) &=
    \underbrace{\frac{\P_1(N|A_z,A_y)}{\P_0(N|Z,A_z,A_y)}}_{:=\omega_{01^*}^N(Z,N,A_z,A_y)}\underbrace{\frac{\P_1(A_z|A_y)}{\P_0(A_z|A_y)}}_{:=\omega_{01}^{A_z}(A_z,A_y)} \label{eq:Weight_N_Model_RW}
\end{align}
where $\omega_{01}^{A_z}(A_z,A_y)$ is the same weighting function in (\ref{eq:Weight_Az}) which can be expressed in terms of odds. Note that this expression for $\omega_{01^*}^{(Z,N,A_z)}$ obviates the need to model $\lambda_0^Z$ because it cancels out in the factorization.

The N-Model-RW estimator is the following weighted average of $\mu_1(Z,N,A_z,A_y)$ among the $G=0$ sample:
\begin{align}
    \hat\theta_{1,\text{N-Model-RW}}^*
    \coloneq\frac{\P_{0n}[\hat\omega_{01^*}^{(Z,N,A,z)}\hat\omega_{0T}^{A_y}\hat\mu_1(Z,N,A_z,A_y)]}{\P_{0n}[\hat\omega_{01^*}^{(Z,N,A_z)}\hat\omega_{0T}^{A_y}]}.
\end{align}
where the estimate of $\omega_{01^*}^{(Z,N,A,z)}$ is provided by the product of $\hat\omega_{01^*}^{N}$ and $\hat\omega_{01^*}^{A_Z}$. It is consistent for $\theta_1^*$ if the model for $\mu_1$ is correctly specified and the weights $\omega_{01^*}^{(Z,N,A_z)}$ (\ref{eq:Weight_N_Model_RW}) and $\omega_{0T}^{A_y}$ (\ref{eq:Weight_Ay}) are consistently estimated.


\subsubsection{Z-Bridge-RW}
It is possible to represent $\theta_1^*$ as a weighted mean of $\mu_1(Z,N,A_z,A_y)$ among the artifical group $G=\blackdiamond$:
\begin{align}
    \theta_1^*
    &\coloneq\E_\blackdiamond\left[\omega_{\blackdiamond1^*}^Z(Z,A_z,A_y)\omega_{1T}^{A_y}(A_y)\mu_1(Z,N,A_z,A_y)\right],\label{eq:ID_Weighting_Z_Bridge}
\end{align}
where the weighting function
\begin{align}
    \omega_{\blackdiamond1^*}^Z(Z,A_z,A_y) 
    \coloneq\frac{\P_0(Z|A_z,A_y)}{\P_\blackdiamond(Z|A_z,A_y)}                                     
    = \frac{\odds(G=0\text{ vs }\blackdiamond|Z,A_z,A_y)}{\odds(G=0\text{ vs }\blackdiamond|A_z,A_y)}.\tag{\ref{eq:Weight_Z_Bridge_RW}}
\end{align}
The weighting function $\omega_{\blackdiamond1^*}^Z$ is the same one used to weight the $\zeta_1^*$ model in Z-Bridge-SWR. That it can be re-expressed as a ratio of odds functions (see note \ref{note:odds_bridge_Z}) does away with the need to model $\lambda_0^Z$. 

The Z-Bridge-RW estimator is the following weighted average of $\mu_1(Z,N,A_z,A_y)$ among the $G=\blackdiamond$ sample:
\begin{align}
    \hat\theta_{1,\text{Z-Bridge-RW}}^*
    \coloneq\frac{\P_{\blackdiamond n}[\hat\omega_{\blackdiamond1^*}^Z\hat\omega_{1T}^{A_y}\hat\mu_1(Z,N,A_z,A_y)]}{\P_{\blackdiamond n}[\hat\omega_{\blackdiamond1^*}^Z\hat\omega_{1T}^{A_y}]}.
\end{align}
It is consistent for $\theta_1^*$ if the model for $\mu_1$ is correctly specified and weights $\omega_{\blackdiamond1^*}^Z$ (\ref{eq:Weight_Z_Bridge_RW}) and $\omega_{1T}^{A_y}$ (\ref{eq:Weight_Ay}) are consistently estimated.

\subsubsection{N-Bridge-RW}
It is possible to represent $\theta_1^*$ as a weighted mean of $\mu_1(Z,N,A_z,A_y)$ among the artifical group $G=\diamond$:
\begin{align}
    \theta_1^*
    &=\E_\diamond\left[\omega_{\diamond1^*}^{(N,A_z)}(N,A_z,A_y)\omega_{0T}^{A_y}(A_y)\mu_1(Z,N,A_z,A_y)\right],\label{eq:ID_Weighting_N_Bridge}
\end{align}
with the weighting function
\begin{align}
    \omega_{\diamond1^*}^{(N,A_z)}(N,A_z,A_y)
    \coloneq\frac{\P_1(N|A_z,A_y)}{\P_\diamond(N|A_z,A_y)}\frac{\P_1(A_z|A_y)}{\P_\diamond(A_z|A_y)}  
    =\frac{\odds(G=1\text{ vs }\diamond|N,A_z,A_y)}{\odds(G=1\text{ vs }\diamond|A_y)}\nonumber.\tag{\ref{eq:Weight_N_Bridge_RW}}
\end{align}
The weighting function $\omega_{\diamond1^*}^{(N,A_z)}$ is the same one that is used to weight the $\nu_1^*$ model in N-Bridge-SWR. That it can be re-expressed as a ratio of odds functions (see note \ref{note:odds_bridge_N}) does away with the need to model $\lambda_1^N$.

The N-Bridge-RW estimator is the following weighted average of $\mu_1(Z,N,A_z,A_y)$ among the $G=\diamond$ sample:
\begin{align}
    \hat\theta_{1,\text{N-Bridge-RW}}^*
    \coloneq\frac{\P_{\diamond n}[\hat\omega_{\diamond1^*}^{(N,A_z)}\hat\omega_{0T}^{A_y}\hat\mu_1(Z,N,A_z,A_y)]}{\P_{\diamond n}[\hat\omega_{\diamond1^*}^{(N,A_z)}\hat\omega_{0T}^{A_y}]}.
\end{align}
It is consistent for $\theta_1^*$ if the model for $\mu_1$ is correctly specified and the weights $\omega_{\diamond1^*}^{(N,A_z)}$ (\ref{eq:Weight_N_Bridge_RW}) and $\omega_{0T}^{A_y}$ (\ref{eq:Weight_Ay}) are consistently estimated.

\newpage

\subsection{Extension to Missing Outcome Data}

To address missing outcome data, we first assume that (i) the potential outcome $Y(Z)$ is independent of outcome missingness status $C$ ($1$ yes, $0$ no) given the point of intervention $Z$, group $G$, and all baseline covariates $(N,A_z,A_y)$ (conditional exchangeability with respect to missing outcome data) and (ii) there are no covariate strata where all persons are loss to follow-up (positivity with respect to compete data). Under these additional assumptions, $\theta_g$ and $\theta_1^*$ can be identified by the PW estimators by incorporating the following inverse probability of censoring weights, which are slight modifications of their usual form in other contexts \citep{cain2009inverse,jackson2024evaluating}:
\begin{align}
	\omega_{g}^{C}(C,Z,N,A_z,A_y) \coloneq \frac{\P_g(C=0)}{\P_g(C=0|Z,N,A_z,A_y)}. \label{eq:Weight_IPCW}
\end{align}
Under these assumptions, $\theta_g$ can also be identified as
\begin{align}
\theta_g &=\E_\std[\eta_g(A_y)] \nonumber \\
         &=\E_\std[\E_g[\mu_1(Z,N,A_z,A_y,C=0)]|A_y],   
\end{align}
where essentially $\mu_1$ is fit among those without missing outcome data. $\theta_1^*$ can be obtained in SR and SWR estimators by fitting $\mu_1$ among those without missing outcome data, i.e., $\mu_1(Z,N,A_z,A_y,C=0)$ rather than $\mu_1(Z,N,A_z,A_y)$. In each SWR estimator, $\mu_1(Z,N,A_z,A_y,C=0)$ is fit by incorporating $\omega_{g}^{C}$ (\ref{eq:Weight_IPCW}) into the weights. Then each SWR estimator remains consistent when either the $\mu_1$ model is correct or the weights (now inclusive of $\omega_{g}^{C}$) are consistently estimated, along with its other robustness conditions. For the RW estimators, under these assumptions, it suffices to fit $\mu_1(Z,N,A_z,A_y,C=0)$ rather than $\mu_1(Z,N,A_z,A_y)$ as the conditional mean outcome model.

\newpage

\subsection{Data Generating Model for the Simulation Study}

The observed data $\mathcal{O}$ ($n$=500 or 5,000) were generated as:
\begin{singlespace}
\begin{equation*}
\begin{aligned}
\mathcal{O}	&=(G,A_{y1},A_{y2},N_1,N_2,A_{z1},A_{z2},Z,Y,W) \\ 
\mathcal{O} & \sim \bigl(\textup{B}(\pi_G),\textup{B}(\pi_{A_{y1}}),\textup{N}(\mu_{A_{y2}},6),\textup{B}(\pi_{N_1}),\textup{B}(\pi_{N_2}),\textup{B}(\pi_{A_{z1}}),\textup{N}(\mu_{A_{z2}},12),\textup{B}(\pi_Z),\textup{N}(\mu_Y,5),\textup{I}(Y\geq 140)\bigr) \\
\pi_G		&=.4 \\
\pi_{A_{y1}}&=\textup{expit}(-0.3-0.29G) \\
\mu_{A_{y2}}&=64-1.5G-1.1A_{y1} \\
\pi_{N_1}	&=\textup{expit}\bigl((-.6+.3G)+(.67-1.09G)A_{y1}+(.006-.012G)A_{y2}\bigr) \\
\pi_{N_2}	&=\textup{expit}\bigl((.39-.69G)+(-.096+.616G)A_{y1}+(-.0087+.0107G)A_{y2}+(.79+1.01G)N_1\bigr) \\
\pi_{A_{z1}}&=\textup{expit}\bigl((-2.6+.6G)+(.41-.24G)A_{y1}+.02A_{y2}+(.08-.16G)N_1 \\ & +(.11-.06G)N_2\bigr) \\
\mu_{A_{z2}}&=(115+16G)+(-1.78+3.78G)A_{y1}+(.3-.15G)A_{y2}+(-.95-1.05G)N_1+(1.58-8.58G)N_2 \\&+(-1.12+2.22G)A_{z1} \\
\pi_Z		&=I(A_{z2} \geq 140)\times\textup{expit}\bigl(\textup{logit}(.25)+(.04-.01G) A_{y1}+(.01-.009G)A_{y2}+(.0005-.0002G)A_{y2}^2 \\&+(.45-.05G)N_1+(.35-.05G)N_2+(.1-.04G)A_{z1}+(.008-.004G)(A_{z2}-150) \\ &+(.00005-.00001G)(A_{z2}-150)^2\bigr) \\ 
\mu_Y		&=(0.3+.05G) A_{y1}+(0.015-.005G)A_{y2}+.0005A_{y2}^2+(-1.2+.05G)N_1+(-1.2+.05G)N_2 \\&+(.9+.05G)A_{z1}+.95A_{z2}+.0002A_{z2}^2+\Bigl(10.5+.12(A_{z2}-150)\times\bigl(G\times.8+(1-G)\bigr)\Bigr)Z
\end{aligned}
\end{equation*}
\end{singlespace}

For data in the observation and intervention arms of the target trial (our estimand), we modified the distribution of $A_y$ by replacing $\pi_{A_{y1}}$ with $(-0.3-0.29)$ and $\mu_{A_{y2}}$ with $(64-1.5-1.1A_{y1})$. To modify the distribution of $Z$, among the $G=1$ group, $Z$ was realized under $\lambda_0^z\equiv\pi_{\tilde{Z}}$ the interventional distribution of $Z$, defined as: 
\begin{singlespace}
\begin{equation*}
\begin{aligned}
\pi_{\tilde{Z}}&=\textup{I}(A_{z2}\geq140)\textup{expit}\bigl(\textup{logit}(.25)+.04A_{y1}+.01A_{y2}+.0005A_{y2}^2+.45N_1^++.35N_2^++.1A_{z1} \\&+.008(A_{z2}-150)+.00005(A_{z2}-150)^2\bigr) \\
N_1^+ &\sim B(\pi_{N_1}^+)\text{ where }\pi_{N_1}^+=\textup{expit}(-.6+.67A_{y1}+.006A_{y2}) \\
N_2^+ &\sim B(\pi_{N_2}^+)\text{ where }\pi_{N_2}^+=\textup{expit}(.39-.096A_{y1}-.0087A_{y2}+.79N_1 )
\end{aligned}
\end{equation*}
\end{singlespace}
Under this definition of $\pi_{\tilde{Z}}$, among $G=1$ in the intervention arm, $Z$ follows $\P_0(Z|A_z,A_y)$, is dependent on $(A_z,A_y)$, and is conditionally independent of $N$ given $(A_z,A_y)$.

\newpage

\subsection{Supplemental Tables}

\begin{landscape}

\begin{table}[htbp]
\caption{Properties and target covariate distributions of weighting functions used in CDA estimators}
    \centering
    \renewcommand{\arraystretch}{1.5}
    \begin{tabular}{>{\arraybackslash}p{0.65in}>{\arraybackslash}p{2.0in}>{\arraybackslash}p{2.50in}>{\arraybackslash}p{0.9in}
    >{\arraybackslash}p{0.9in}>{\centering\arraybackslash}p{.40in}>{\centering\arraybackslash}p{.55in}}
    \hline
    Weight & Expression & Component & Actual & Target & Sample$^\dag$ weighted & Sample$^\dag$ emulated \\
    \hline
    $\omega_{gT}^{A_y}$ &  (\ref{eq:Weight_Ay}) &                                                                   & $\P_g(A_y)$                       & $\P_\std(A_y)$                    & $G=g$ & $T=1$ \\
    $\omega_{10}^{A_z}$ &  (\ref{eq:Weight_Az})                    & $\odds(G=1~vs~0|A_z,A_y)$                     & $\P_0(A_z,A_y)$                   & $\P_1(A_z,A_y)$                   & $G=0$ & $G=1$ \\   
    $\omega_{10}^{A_z}$ &  (\ref{eq:Weight_Az})                    & $\odds(G=1~vs~0|A_y)$                         & $\P_0(A_y)$                       & $\P_1(A_y)$                       & $G=0$ & $G=1$ \\        
    $\omega^{(Z,N)}_{11^*}$ & (\ref{eq:Weight_Z_Model_PW}), (\ref{eq:Weight_N_Model_PW}), (\ref{eq:Weight_Z_Bridge_PW_2}),\text{ and }(\ref{eq:Weight_N_Bridge_PW}) &                                              & $\P_1(Z,N,A_z,A_y)$                    & $ \P_\blacktriangle(Z,N,A_z,A_y)$                   & $G=1$  & $G=\blacktriangle$ \\
    $\omega_{11^*}^{(Z,N)}$ & (\ref{eq:Weight_Z_Bridge_PW_2})   & $\omega_{1\blackdiamond}^{N}:~\odds(G=\blackdiamond~vs~1|Z,N,A_z,A_y)$       & $\P_1(Z,N,A_z,A_y)$                 & $\P_\blackdiamond(Z,N,A_z,A_y)$     & $G=1$ & $G=\blackdiamond$ \\
    $\omega_{11^*}^{(Z,N)}$ & (\ref{eq:Weight_Z_Bridge_PW_2})   & $\omega_{11^*}^{Z}:~\odds(G=\blackdiamond~vs~0|Z,A_z,A_y)$                     & $\P_0(Z,A_z,A_y)$                   & $\P_\blackdiamond(Z,A_z,A_y)$                   & $G=0$ & $G=\blackdiamond$ \\
    $\omega_{11^*}^{(Z,N)}$ & (\ref{eq:Weight_Z_Bridge_PW_2})   & $\omega_{11^*}^{Z}:~\odds(G=1~vs~0|A_z,A_y)$                     & $\P_0(A_z,A_y)$                   & $\P_1(A_z,A_y)$                   & $G=0$ & $G=1$ \\
    $\omega_{11^*}^{(Z,N)}$ & (\ref{eq:Weight_N_Bridge_PW})     & $\omega_{1\diamond}^{Z}:~\odds(G=\diamond~vs~1|Z,N,A_z,A_y)$          & $\P_1(Z,N,A_z,A_y)$               & $\P_\diamond(Z,N,A_z,A_y)$        & $G=1$ & $G=\diamond$ \\
    $\omega_{11^*}^{(Z,N)}$ & (\ref{eq:Weight_N_Bridge_PW})     & $\omega_{1\diamond}^{Z}:~\odds(G=\diamond~vs~1|N,A_z,A_y)$            & $\P_1(N,A_z,A_y)$                 & $\P_\diamond(N,A_z,A_y)$          & $G=1$ & $G=\diamond$ \\    
    $\omega_{\blackdiamond 1^*}^Z$ & (\ref{eq:Weight_Z_Bridge_RW})    & $\odds(G=0~vs~\blackdiamond|Z,A_z,A_y)$       & $\P_\blackdiamond(Z,A_z,A_y)$     & $\P_0(Z,A_z,A_y)$                 & $G=\blackdiamond$ & $G=0$ \\    
    $\omega_{\blackdiamond 1^*}^Z$ & (\ref{eq:Weight_Z_Bridge_RW})    & $\odds(G=0~vs~\blackdiamond|A_z,A_y)$         & $\P_\blackdiamond(A_z,A_y)$       & $\P_0(A_z,A_y)$                   & $G=\blackdiamond$ & $G=0$ \\
    $\omega_{\diamond 1^*}^{(N,A_z)}$ & (\ref{eq:Weight_N_Bridge_RW})     & $\odds(G=1~vs~\diamond|N,A_z,A_y)$            & $\P_\diamond(N,A_z,A_y)$          & $\P_1(N,A_z,A_y)$                 & $G=\diamond$ & $G=1$ \\       
    $\omega_{\diamond 1^*}^{(N,A_z)}$ & (\ref{eq:Weight_N_Bridge_RW})     & $\odds(G=1~vs~\diamond|A_y)$                  & $\P_\diamond(A_y)$                & $\P_1(A_y)$                       & $G=\diamond$ & $G=1$ \\       
    $\omega_{01^*}^{(Z,N,A_z)}$ &  (\ref{eq:Weight_Z_Model_RW}) and (\ref{eq:Weight_N_Model_RW})       &                                       & $\P_0(Z,N,A_z|A_y)$                   & $\P_\blacktriangle(Z,N,A_z|A_y)$      & $G=0$ & $G=\blacktriangle$ \\
    \hline
   \multicolumn{7}{l}{\footnotesize{$^\dag$The artificial samples referenced here are specified in Definition \ref{def:Sample_Z_Model} ($G=\blacktriangle$), Definition \ref{def:Group_Z_Bridge} ($G=\blackdiamond$), and Definition \ref{def:Group_N_Bridge} ($G=\diamond$)}}\\
   \multicolumn{7}{l}{\footnotesize{$\omega_{gT}^{A_y}$ is the standardizing weight (\ref{eq:Weight_Ay}) which appears in E-OBD and all PW, RW, and SWR estimators}}\\
   \multicolumn{7}{l}{\footnotesize{$\omega_{10}^{A_z}$ is the covariate shifting weight (\ref{eq:Weight_Az}) which appears in E-OBD and the N-modeling SWR estimator}}\\
   \multicolumn{7}{l}{\footnotesize{$\omega_{11^*}^{(Z,N)}$ is the interventional weight which has several forms (\ref{eq:Weight_Z_Model_PW}) (\ref{eq:Weight_N_Model_PW}) (\ref{eq:Weight_Z_Bridge_PW_2}) and (\ref{eq:Weight_N_Bridge_PW}) that appear in PW and SWR estimators}}\\
    \multicolumn{7}{l}{\footnotesize{$\omega_{\blackdiamond 1^*}^Z$ (\ref{eq:Weight_Z_Bridge_RW}) and $\omega_{\diamond 1^*}^{(N,A_z)}$ (\ref{eq:Weight_N_Bridge_RW}) appear in bridging SWR estimators and bridging RW estimators}} \\
   \multicolumn{7}{l}{\footnotesize{$\omega_{01^*}^{(Z,N,A_z)}$ has two forms (\ref{eq:Weight_Z_Model_RW}) and (\ref{eq:Weight_N_Model_RW}) which appear in modeling RW estimators}} \\\\
    \end{tabular}
    \label{tab:WeightProperties}
\end{table}

\end{landscape}

\newpage

where 
\begin{align}
	\omega_{gT}^{A_y}(A_y) &\coloneq \frac{\P_\std(A_y)}{\P_g(A_y)} \tag{\ref{eq:Weight_Ay}} \\
    &= \frac{\P(T=1|A_y)}{\P(G=g|A_y)} \times \frac{\P(G=g)}{\P(T=1)},\nonumber\\\nonumber\\
	\omega_{01}^{A_z}(A_z,A_y) \tag{\ref{eq:Weight_Az}}	
    &\coloneq \frac{\P_1(A_z|A_y)}{\P_0(A_z|A_y)} \\
    &=\frac{\odds(G=1\text{ vs }0|A_z,A_y)}{\odds(G=1\text{ vs }0|A_y)},\nonumber\\\nonumber\\
    \omega_{11^*}^{(Z,N)}(Z,N,A_z,A_y) \tag{\ref{eq:Weight_Z_Model_PW}}
    &=\frac{\P_0(Z|A_z,A_y)}{\P_1(Z|N,A_z,A_y)},\\\nonumber\\
    \omega_{11^*}^{(Z,N)}(Z,N,A_z,A_y)\tag{\ref{eq:Weight_N_Model_PW}}
    &=\underbrace{\frac{\P_1(N|A_z,A_y)}{\P_1(N|Z,A_z,A_y)}}_{=:\omega_{11^*}^N(Z,N,A_z,A_y)}\underbrace{\frac{\P_0(Z|A_z,A_y)}{\P_1(Z|A_z,A_y)}}_{=:\omega_{11^*}^Z(Z,A_z,A_y)} \\
    &=\frac{\P_1(N|A_z,A_y)}{\P_1(N|Z,A_z,A_y)}\frac{\text{odds}(G=0~\text{vs}~1\mid Z,A_z,A_y)}{\text{odds}(G=0~\text{vs}~1\mid A_z,A_y)}, \nonumber\\\nonumber\\
    \omega_{11^*}^{(Z,N)}(Z,N,A_z,A_y)
    &=\underbrace{\frac{\P_\blackdiamond(N|Z, A_z,A_y)}{\P_1(N| Z,A_z,A_y)}}_{=:\omega_{1\blackdiamond}^N(Z,N,A_z,A_y)}\underbrace{\frac{\P_0(Z|A_z,A_y)}{\P_1(Z|A_z,A_y)}}_{=:\omega_{11^*}^Z(Z,A_z,A_y)},\tag{\ref{eq:Weight_Z_Bridge_PW_2}} \\
    &=\underbrace{\odds(G=\blackdiamond~vs~1|Z,N,A_z,A_y)\frac{\odds(G=1~vs~0|A_z,A_y)}{\odds(G=\blackdiamond~vs~0|Z,A_z,A_y}}_{\omega_{1\blackdiamond}^N(Z,N,A_z,A_y)\omega_{11^*}^Z(Z,A_z,A_y)},\nonumber\\\nonumber\\
    \omega_{11^*}^{(Z,N)}(Z,N,A_z,A_y) \tag{\ref{eq:Weight_N_Bridge_PW}}
    &=\underbrace{\frac{\P_\diamond(Z| N,A_z,A_y)}{\P_1(Z| N,A_z,A_y)}}_{=:\omega_{1\diamond}^Z(Z,N,A_z,A_y)} \\
    &=\frac{\text{odds}(G=\diamond~\text{vs}~1\mid Z,N,A_z,A_y)}{\text{odds}(G=\diamond~\text{vs}~1\mid N,A_z,A_y)}, \nonumber\\\nonumber\\
    \omega_{\blackdiamond1^*}^Z(Z,A_z,A_y) \tag{\ref{eq:Weight_Z_Bridge_RW}} 
    &\coloneq\frac{\P_0(Z|A_z,A_y)}{\P_\blackdiamond(Z|A_z,A_y)} \\                                    
    &= \frac{\odds(G=0\text{ vs }\blackdiamond|Z,A_z,A_y)}{\odds(G=0\text{ vs }\blackdiamond|A_z,A_y)},\nonumber\\\nonumber\\
    \omega_{\diamond1^*}^{(N,A_z)}(N,A_z,A_y) \tag{\ref{eq:Weight_N_Bridge_RW}}
    &\coloneq\frac{\P_1(N|A_z,A_y)}{\P_\diamond(N|A_z,A_y)}\frac{\P_1(A_z|A_y)}{\P_\diamond(A_z|A_y)} \\ 
    &=\frac{\odds(G=1\text{ vs }\diamond|N,A_z,A_y)}{\odds(G=1\text{ vs }\diamond|A_y)},\nonumber\\\nonumber\\
    \omega_{01^*}^{(Z,N,A_z)}(Z,N,A_z,A_y) \tag{\ref{eq:Weight_Z_Model_RW}}
    &= \underbrace{\frac{\P_0(Z|A_z,A_y)}{\P_0(Z|N,A_z,A_y)}}_{:=\omega_{01^*}^{Z}(Z,N,A_z,A_y)}
     \underbrace{\frac{\P_1(N|A_z,A_y)}{\P_0(N|A_z,A_y)}\frac{\P_1(A_z|A_y)}{\P_0(A_z|A_y)}}_{:=\omega_{01^*}^{(N,A_z)}(N,A_z,A_y)}\\
     &=\frac{\P_0(Z|A_z,A_y)}{\P_0(Z|N,A_z,A_y)}\frac{\odds(G=1\text{ vs }0|N,A_z,A_y)}{\odds(G=1\text{ vs }0|A_y)},\nonumber\\\nonumber\\
    \omega_{01^*}^{(Z,N,A_z)}(Z,N,A_z,A_y)\tag{\ref{eq:Weight_N_Model_RW} }
    &=\underbrace{\frac{\P_1(N|A_z,A_y)}{\P_0(N|Z,A_z,A_y)}}_{:=\omega_{01^*}^N(Z,N,A_z,A_y)}\underbrace{\frac{\P_1(A_z|A_y)}{\P_0(A_z|A_y)}}_{:=\omega_{01}^{A_z}(A_z,A_y)}\\
    &=\frac{\P_1(N|A_z,A_y)}{\P_0(N|Z,A_z,A_y)}\frac{\odds(G=1\text{ vs }0|A_z,A_y)}{\odds(G=1\text{ vs }0|A_y)}.\nonumber\\\nonumber
\end{align}

Note also the inverse probability of censoring weights

\begin{align}
	\omega_{g}^{C}(C,Z,N,A_z,A_y) \coloneq \frac{\P_g(C=0)}{\P_g(C=0|Z,N,A_z,A_y)}.\tag{\ref{eq:Weight_IPCW}}
\end{align}

\newpage

\begin{table}[htbp]
  \caption{Estimator performance under incorrect nuisance models by outcome type and sample size.}
  \centering
    \begin{tabular}{lrrrrrrrrr}
\hline    
    \multicolumn{10}{c}{Sample Size = 5,000} \\
\hline    
          & \multicolumn{4}{c}{Continuous ($\theta_1^*=132.606$)} &       & \multicolumn{4}{c}{Binary ($\theta_1^*=.245$)} \\
          & \multicolumn{1}{l}{Bias} & \multicolumn{1}{l}{SE} & \multicolumn{1}{l}{RMSE} & \multicolumn{1}{l}{Coverage} &       & \multicolumn{1}{l}{Bias} & \multicolumn{1}{l}{SE} & \multicolumn{1}{l}{RMSE} & \multicolumn{1}{l}{Coverage} \\
\hline          
    Linear (E-OBD) & 1.077 & 0.280 & 1.113 & 0.026 &       & 0.044 & 0.010 & 0.045 & 0.006 \\
    Pure Weighting (PW) &       &       &       &       &       &       &       &       &  \\
    \multicolumn{1}{r}{Z-Model-PW} & 1.151 & 0.284 & 1.186 & 0.012 &       & 0.044 & 0.010 & 0.045 & 0.006 \\
    \multicolumn{1}{r}{N-Model-PW} & 1.215 & 0.283 & 1.247 & 0.007 &       & 0.045 & 0.010 & 0.046 & 0.007 \\
    \multicolumn{1}{r}{Z-Bridge-PW} & 1.067 & 0.286 & 1.105 & 0.030 &       & 0.043 & 0.010 & 0.044 & 0.007 \\
    \multicolumn{1}{r}{N-Bridge-PW} & 1.135 & 0.282 & 1.169 & 0.013 &       & 0.045 & 0.010 & 0.046 & 0.006 \\
    \multicolumn{10}{l}{Regress-then-Weight (RW)} \\
    \multicolumn{1}{r}{Z-Model-RW} & 1.100 & 0.282 & 1.136 & 0.022 &       & 0.043 & 0.010 & 0.045 & 0.006 \\
    \multicolumn{1}{r}{N-Model-RW} & 1.077 & 0.281 & 1.113 & 0.028 &       & 0.042 & 0.010 & 0.043 & 0.014 \\
    \multicolumn{1}{r}{Z-Bridge-RW} & 1.077 & 0.280 & 1.113 & 0.027 &       & 0.042 & 0.010 & 0.043 & 0.012 \\
    \multicolumn{1}{r}{N-Bridge-RW} & 1.801 & 0.281 & 1.823 & 0.000 &       & 0.060 & 0.011 & 0.061 & 0.000 \\
    \multicolumn{10}{l}{Sequential Regressions (SR)} \\
    \multicolumn{1}{r}{Z-Model-SR} & 1.077 & 0.290 & 1.116 & 0.030 &       & 0.042 & 0.010 & 0.043 & 0.015 \\
    \multicolumn{1}{r}{N-Model-SR} & 0.643 & 0.282 & 0.702 & 0.354 &       & 0.026 & 0.010 & 0.028 & 0.291 \\
    \multicolumn{1}{r}{Z-Bridge-SR} & 1.078 & 0.280 & 1.113 & 0.026 &       & 0.042 & 0.010 & 0.043 & 0.012 \\
    \multicolumn{1}{r}{N-Bridge-SR} & 1.078 & 0.280 & 1.113 & 0.026 &       & 0.042 & 0.010 & 0.043 & 0.012 \\
    \multicolumn{10}{l}{Sequential Weighted Regressions (SWR)} \\
    \multicolumn{1}{r}{Z-Model-SWR} & 1.129 & 0.293 & 1.166 & 0.021 &       & 0.042 & 0.011 & 0.044 & 0.015 \\
    \multicolumn{1}{r}{N-Model-SWR} & 0.780 & 0.290 & 0.832 & 0.194 &       & 0.029 & 0.011 & 0.031 & 0.187 \\
    \multicolumn{1}{r}{Z-Bridge-SWR} & 1.089 & 0.281 & 1.124 & 0.022 &       & 0.042 & 0.010 & 0.043 & 0.014 \\
    \multicolumn{1}{r}{N-Bridge-SWR} & 1.078 & 0.280 & 1.114 & 0.026 &       & 0.042 & 0.010 & 0.043 & 0.012 \\
\hline
    \multicolumn{10}{c}{Sample Size = 500} \\
\hline
          & \multicolumn{4}{c}{Continuous ($\theta_1^*=132.606$)} &       & \multicolumn{4}{c}{Binary ($\theta_1^*=.245$)} \\
          & \multicolumn{1}{l}{Bias} & \multicolumn{1}{l}{SE} & \multicolumn{1}{l}{RMSE} & \multicolumn{1}{l}{Coverage} &       & \multicolumn{1}{l}{Bias} & \multicolumn{1}{l}{SE} & \multicolumn{1}{l}{RMSE} & \multicolumn{1}{l}{Coverage} \\
\hline          
    Linear (E-OBD) & 1.082 & 0.866 & 1.386 & 0.741 &       & 0.044 & 0.032 & 0.054 & 0.724 \\
    \multicolumn{10}{l}{Pure Weighting (PW)} \\
    \multicolumn{1}{r}{Z-Model-PW} & 1.159 & 0.881 & 1.456 & 0.721 &       & 0.044 & 0.033 & 0.055 & 0.745 \\
    \multicolumn{1}{r}{N-Model-PW} & 1.218 & 0.872 & 1.498 & 0.702 &       & 0.045 & 0.033 & 0.056 & 0.728 \\
    \multicolumn{1}{r}{Z-Bridge-PW} & 1.066 & 0.889 & 1.388 & 0.767 &       & 0.043 & 0.032 & 0.054 & 0.752 \\
    \multicolumn{1}{r}{N-Bridge-PW} & 1.144 & 0.872 & 1.438 & 0.717 &       & 0.045 & 0.032 & 0.056 & 0.704 \\
    \multicolumn{10}{l}{Regress-then-Weight (RW)} \\
    \multicolumn{1}{r}{Z-Model-RW} & 1.107 & 0.873 & 1.410 & 0.734 &       & 0.043 & 0.032 & 0.054 & 0.729 \\
    \multicolumn{1}{r}{N-Model-RW} & 1.081 & 0.874 & 1.390 & 0.752 &       & 0.042 & 0.032 & 0.053 & 0.754 \\
    \multicolumn{1}{r}{Z-Bridge-RW} & 1.082 & 0.868 & 1.387 & 0.745 &       & 0.042 & 0.032 & 0.053 & 0.748 \\
    \multicolumn{1}{r}{N-Bridge-RW} & 1.810 & 0.854 & 2.001 & 0.443 &       & 0.061 & 0.033 & 0.069 & 0.565 \\
    \multicolumn{10}{l}{Sequential Regressions (SR)} \\
    \multicolumn{1}{r}{Z-Model-SR} & 1.072 & 0.899 & 1.399 & 0.760 &       & 0.041 & 0.032 & 0.053 & 0.765 \\
    \multicolumn{1}{r}{N-Model-SR} & 0.647 & 0.890 & 1.100 & 0.889 &       & 0.026 & 0.033 & 0.041 & 0.903 \\
    \multicolumn{1}{r}{Z-Bridge-SR} & 1.083 & 0.866 & 1.386 & 0.740 &       & 0.042 & 0.032 & 0.053 & 0.754 \\
    \multicolumn{1}{r}{N-Bridge-SR} & 1.083 & 0.866 & 1.386 & 0.740 &       & 0.042 & 0.032 & 0.053 & 0.753 \\
    \multicolumn{10}{l}{Sequential Weighted Regressions (SWR)} \\
    \multicolumn{1}{r}{Z-Model-SWR} & 1.133 & 0.909 & 1.453 & 0.759 &       & 0.042 & 0.034 & 0.054 & 0.773 \\
    \multicolumn{1}{r}{N-Model-SWR} & 0.771 & 0.887 & 1.175 & 0.852 &       & 0.028 & 0.033 & 0.044 & 0.891 \\
    \multicolumn{1}{r}{Z-Bridge-SWR} & 1.099 & 0.876 & 1.405 & 0.745 &       & 0.042 & 0.033 & 0.053 & 0.761 \\
    \multicolumn{1}{r}{N-Bridge-SWR} & 1.085 & 0.866 & 1.388 & 0.740 &       & 0.042 & 0.032 & 0.053 & 0.750 \\
\hline    
    \multicolumn{10}{l}{\footnotesize{All were nuisance models specified without interactions, higher order terms, and certain covariates $(N_2,A_{z1},A_{y1})$.}} \\
    \multicolumn{10}{l}{\footnotesize{The true outcome model contained a $Z\times A_z$ interaction.}} \\
    \multicolumn{10}{l}{\footnotesize{The lesser bias of N-Model-PW/RW/SWR stems from modeling both $N_1$ and $N_2$ for $\omega_{11^*}^{N}$ rather than only modeling $N_1$.}}
    \end{tabular}%
\label{tab:AllEstimators_Omit}
\end{table}%

\begin{table}[htpb]
  \caption{Performance of Regress-then-Weight Estimators under correct or flexible nuisance models by outcome type and sample size.}
  \centering
    \begin{tabular}{lrrrrrrrrr}
\hline    
    \multicolumn{10}{c}{Sample Size = 5,000} \\
\hline    
          & \multicolumn{4}{c}{Continuous ($\theta_1^*=132.606$)} &       & \multicolumn{4}{c}{Binary ($\theta_1^*=.245$)} \\
          & \multicolumn{1}{l}{Bias} & \multicolumn{1}{l}{SE} & \multicolumn{1}{l}{RMSE} & \multicolumn{1}{l}{Coverage} &       & \multicolumn{1}{l}{Bias} & \multicolumn{1}{l}{SE} & \multicolumn{1}{l}{RMSE} & \multicolumn{1}{l}{Coverage} \\
\hline          
    \multicolumn{1}{r}{Z-Model-RW} & 0.013 & 0.269 & 0.269 & 0.944 &       & 0.000 & 0.010 & 0.010 & 0.955 \\
    \multicolumn{1}{r}{N-Model-RW} & -0.086 & 0.297 & 0.310 & 0.946 &       & -0.003 & 0.011 & 0.011 & 0.948 \\
    \multicolumn{1}{r}{Z-Bridge-RW} & -0.087 & 0.267 & 0.281 & 0.932 &       & 0.005 & 0.010 & 0.011 & 0.918 \\
    \multicolumn{1}{r}{N-Bridge-RW} & 0.090 & 0.283 & 0.297 & 0.933 &       & -0.001 & 0.010 & 0.010 & 0.949 \\
\hline
    \multicolumn{10}{c}{Sample Size = 500} \\
\hline
          & \multicolumn{4}{c}{Continuous ($\theta_1^*=132.606$)} &       & \multicolumn{4}{c}{Binary ($\theta_1^*=.245$)} \\
          & \multicolumn{1}{l}{Bias} & \multicolumn{1}{l}{SE} & \multicolumn{1}{l}{RMSE} & \multicolumn{1}{l}{Coverage} &       & \multicolumn{1}{l}{Bias} & \multicolumn{1}{l}{SE} & \multicolumn{1}{l}{RMSE} & \multicolumn{1}{l}{Coverage} \\
\hline          
    \multicolumn{1}{r}{Z-Model-RW} & 0.204 & 1.665 & 1.677 & 0.979 &       & -0.009 & 0.047 & 0.048 & 0.979 \\
    \multicolumn{1}{r}{N-Model-RW} & -0.314 & 0.966 & 1.015 & 0.977 &       & -0.011 & 0.038 & 0.040 & 0.976 \\
    \multicolumn{1}{r}{Z-Bridge-RW} & -0.285 & 0.829 & 0.876 & 0.943 &       & 0.005 & 0.031 & 0.031 & 0.964 \\
    \multicolumn{1}{r}{N-Bridge-RW} & 0.121 & 0.911 & 0.919 & 0.970 &       & 0.000 & 0.035 & 0.035 & 0.969 \\
\hline    
    \multicolumn{10}{l}{\footnotesize{$^\dagger$The true outcome model contained a $Z\times A_z$ interaction.}} 
    \end{tabular}%
\label{tab:AllEstimators_CorrFlex_RW}
\end{table}%

\begin{table}[htbp]
  \centering
  \caption{Distribution of Weights used by CDA Estimators}
    \begin{tabular}{lrrr}
    \hline
    Weight & \multicolumn{1}{c}{Mean} & \multicolumn{1}{c}{Min} & \multicolumn{1}{c}{Max} \\
    \hline
    $\omega_{gT}^{A_y}$ & 0.998 & 0.318 & 2.637 \\ \\
    $\omega_{10}^{A_z}$ & 1.001 & 0.169 & 9.267 \\ \\
    $\omega_{11^*}^{(Z,N)}$ (\ref{eq:Weight_Z_Model_PW}) & 1.009 & 0.158 & 24.305 \\ \\
    $\omega_{11^*}^{(Z,N)}$ (\ref{eq:Weight_N_Model_PW}$^\dag$)  & 0.998 & 0.129 & 12.343 \\ \\
    $\omega_{11^*}^{(Z,N)}$ (\ref{eq:Weight_Z_Bridge_PW_2}$^\dag$)  & 1.004 & 0.596 & 1.573 \\ \\
    $\omega_{11^*}^{(Z,N)}$ (\ref{eq:Weight_N_Bridge_PW}$^\dag$)  & 1.000     & 0.566 & 1.699 \\ \\
    $\omega_{\blackdiamond 1^*}^Z $ (\ref{eq:Weight_Z_Bridge_RW})$^\dag$ & 1.016 & 0.000     & 2.869 \\ \\
    $\omega_{\diamond 1^*}^{(N,A_z)} $ (\ref{eq:Weight_N_Bridge_RW})$^\dag$ & 0.967 & 0.068 & 10.256 \\ \\
    $\omega_{1}^{C}$ & 1.011 & 0.089 & 9.682 \\ \\
    $\omega_{0}^{C}$ & 1.006 & 0.081 & 8.958 \\
    \hline
    \multicolumn{4}{l}{\footnotesize{$^\dag$Estimated using expressions involving odds}}
    \end{tabular}%
  \label{tab:Weight_Distribution}%
\end{table}%

\clearpage

\newpage

\subsection{Supplemental Figures}

\begin{figure}[h]
    \centering
    \includegraphics[width=0.5\linewidth]{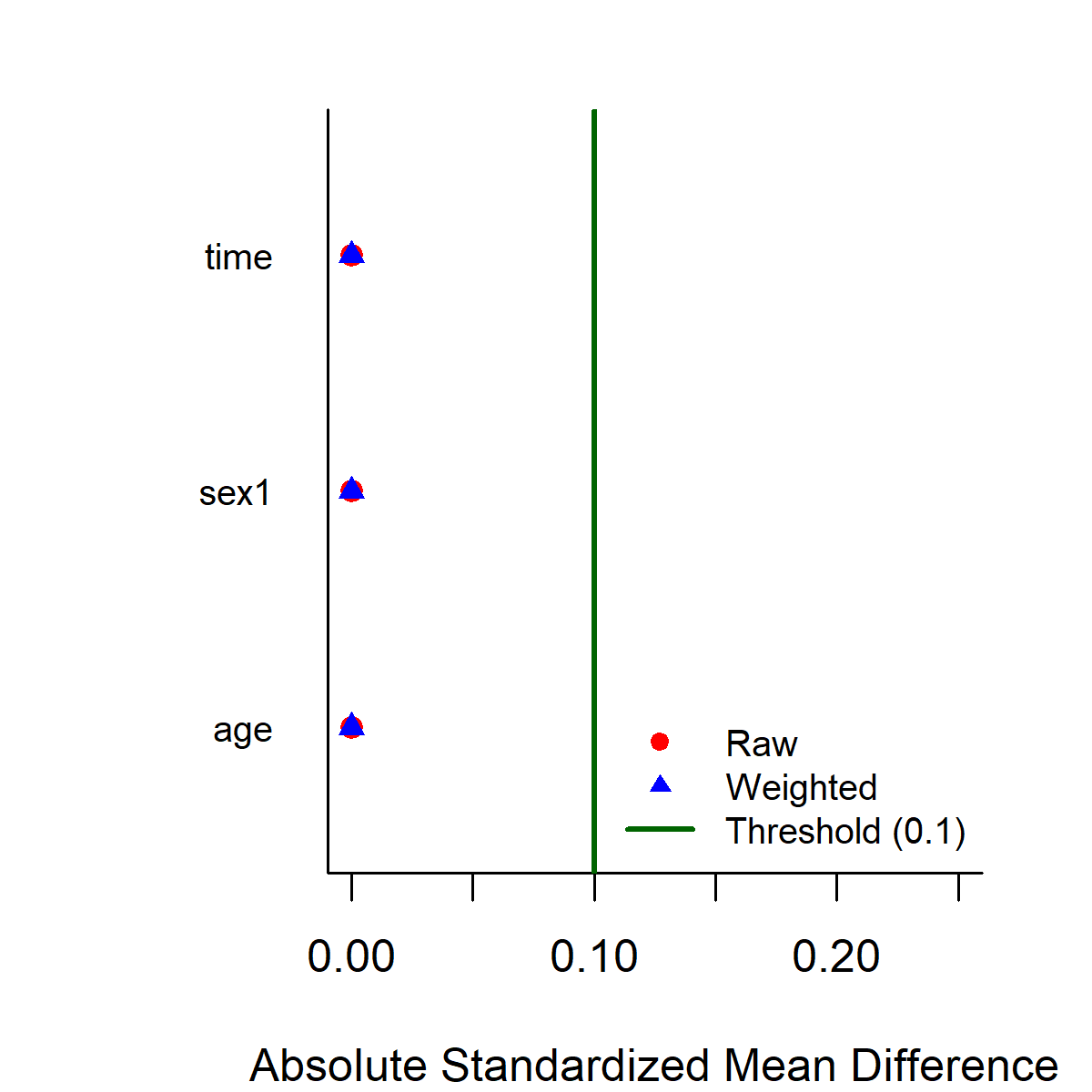}
    \caption{Target Covariate Balance for $\omega_{1T}^{A_y}$} (\ref{eq:Weight_Ay})
    \label{fig:Weight_Ay_1T}
\end{figure}

\begin{figure}[h]
    \centering
    \includegraphics[width=0.5\linewidth]{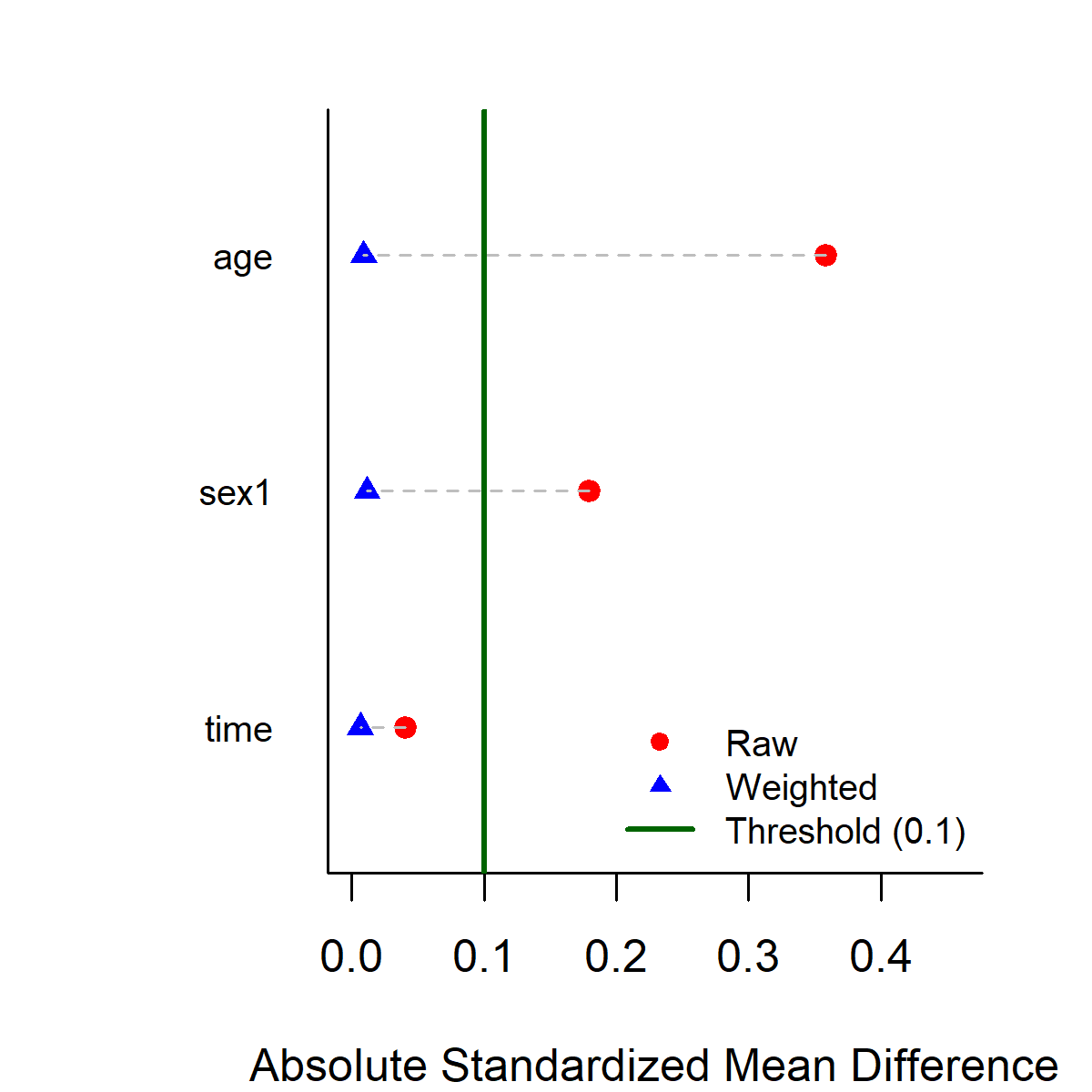}
    \caption{Target Covariate Balance for $\omega_{0T}^{A_y}$}  (\ref{eq:Weight_Ay})
    \label{fig:Weight_Ay_0T}
\end{figure}

\begin{figure}[h]
    \centering
    \includegraphics[width=0.5\linewidth]{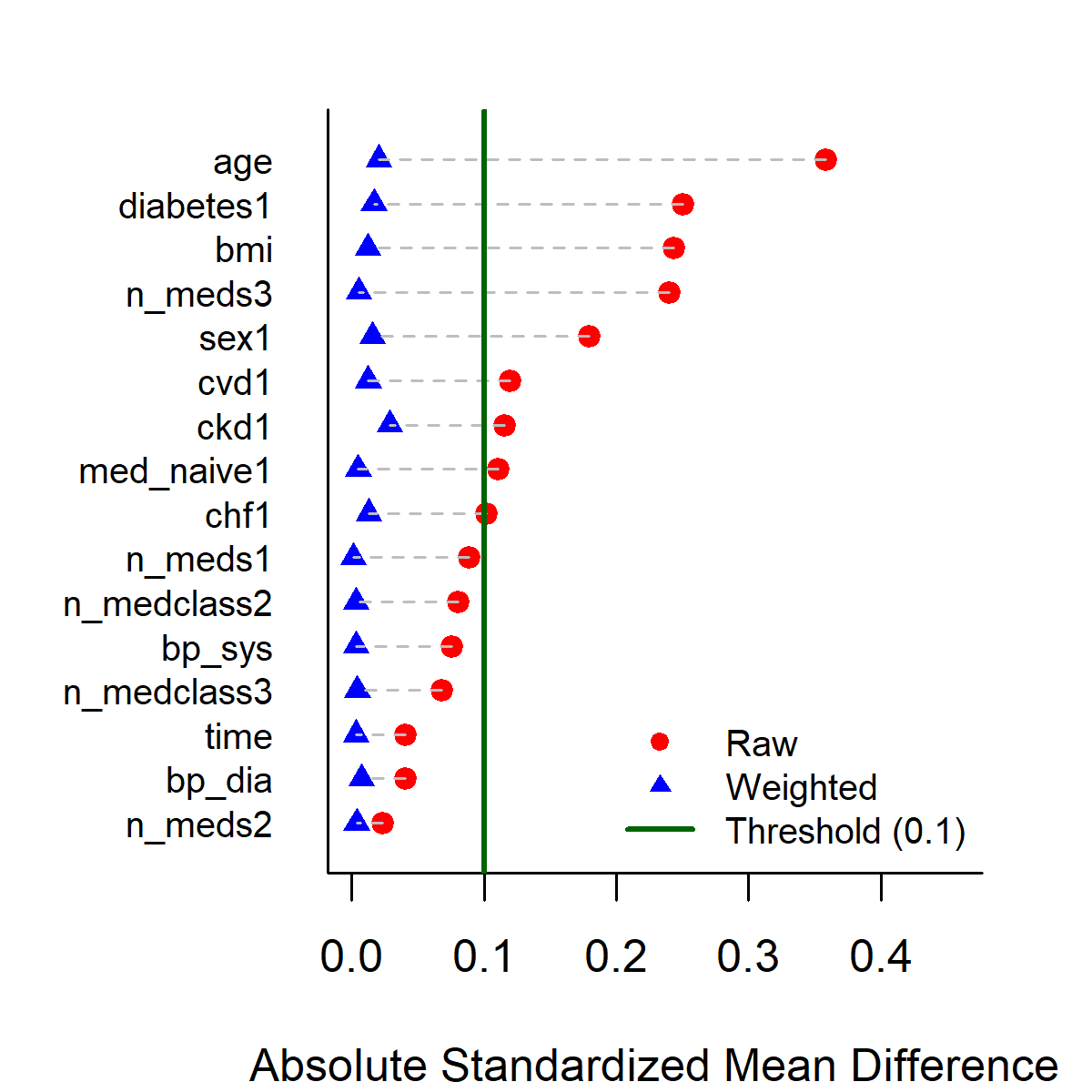}
    \caption{Target Covariate Balance for $\omega_{0T}^{A_z}$ (\ref{eq:Weight_Az}) as expressed in odds terms (numerator)}
    \label{fig:Weight_Az_01_part_i}
\end{figure}

\begin{figure}[h]
    \centering
    \includegraphics[width=0.5\linewidth]{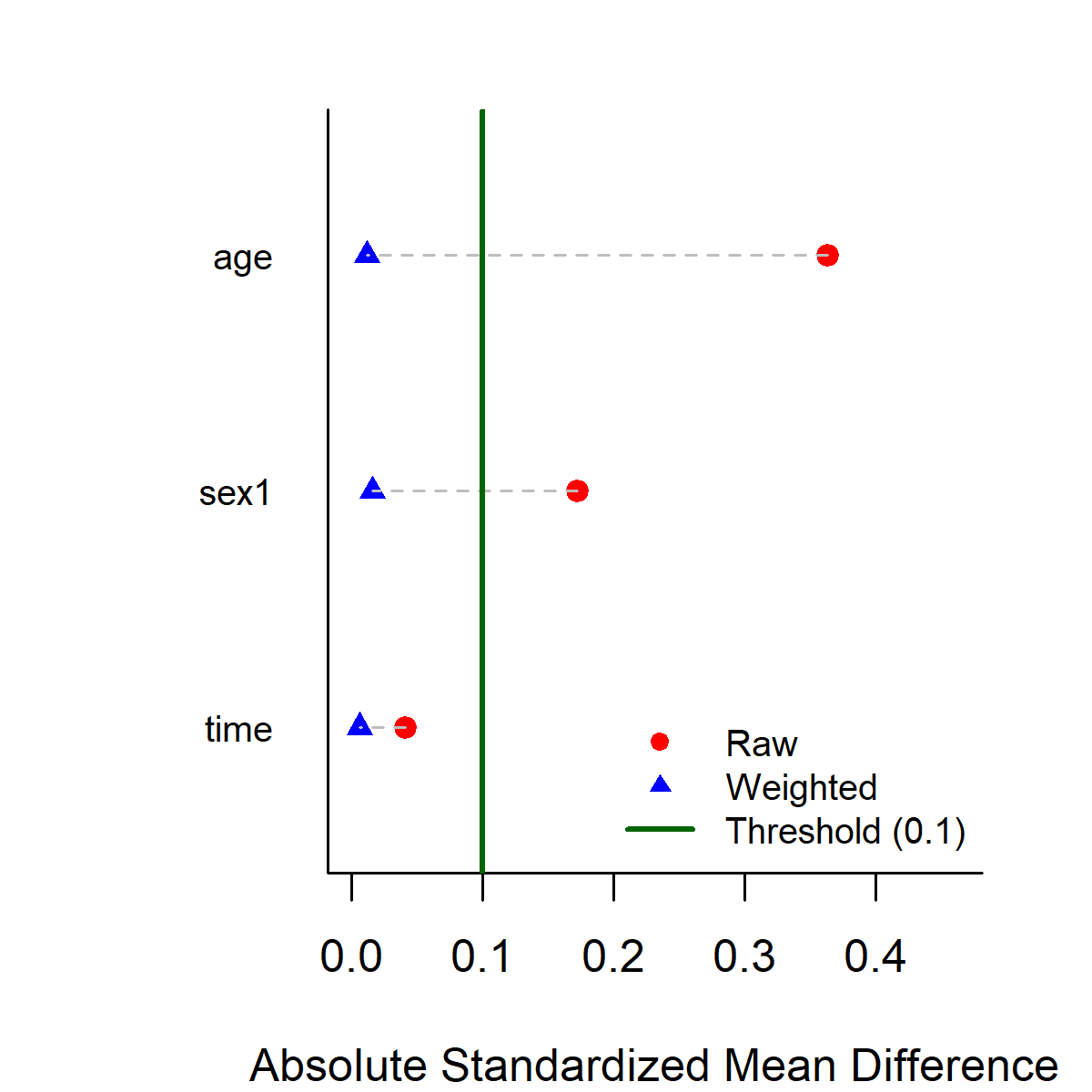}
    \caption{Target Covariate Balance for $\omega_{0T}^{A_z}$ (\ref{eq:Weight_Az}) as expressed in odds terms (denominator)}
    \label{fig:Weight_Az_01_part_ii}
\end{figure}

\begin{figure}[h]
    \centering
    \includegraphics[width=0.5\linewidth]{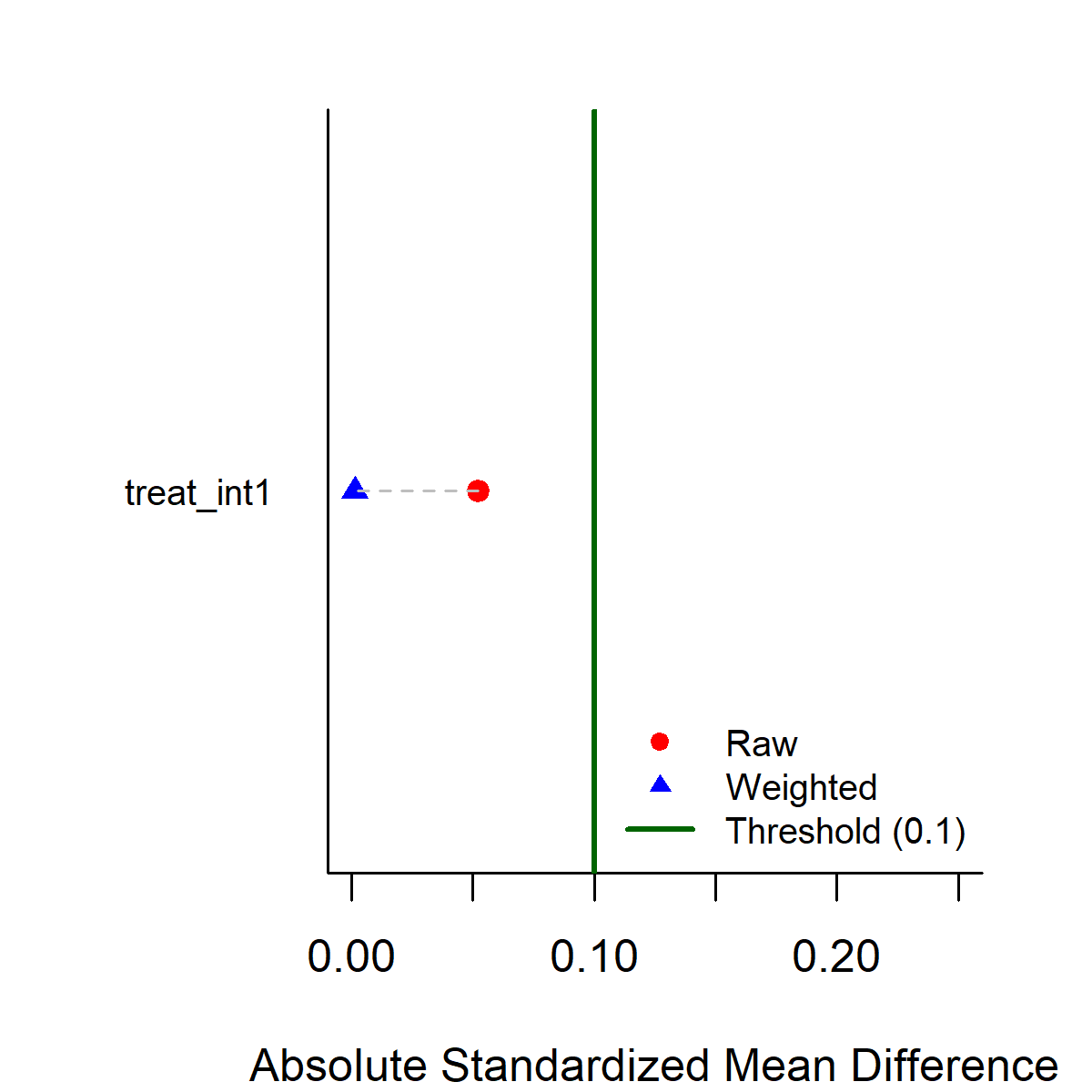}
    \caption{Target Covariate Balance for $\omega_{11^*}^{(Z,N)}$ (\ref{eq:Weight_Z_Model_PW})}
    \label{fig:Weight_ZN_1I_ZM}
\end{figure}

\begin{figure}[h]
    \centering
    \includegraphics[width=0.5\linewidth]{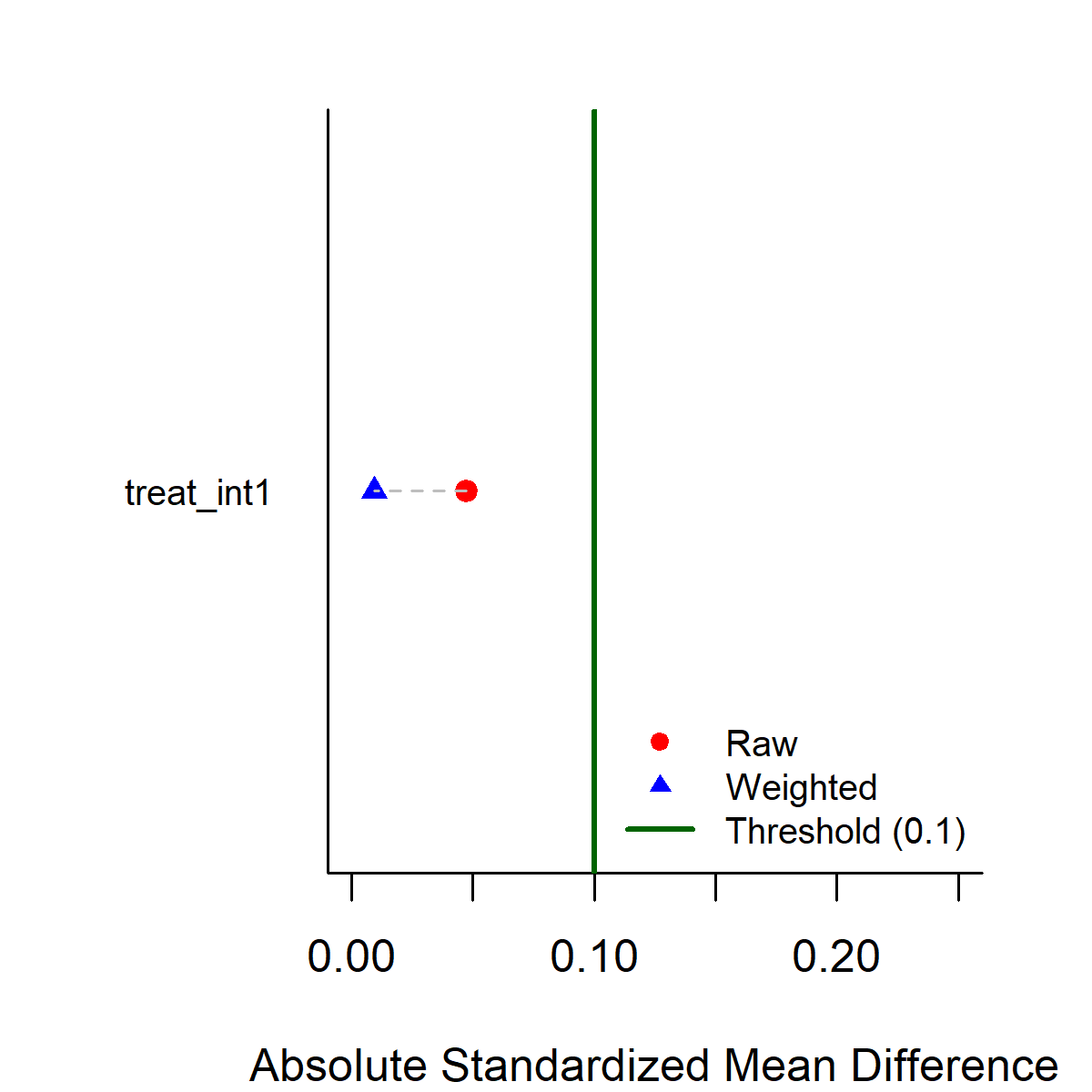}
    \caption{Target Covariate Balance for $\omega_{11^*}^{(Z,N)}$ (\ref{eq:Weight_N_Model_PW})}
    \label{fig:Weight_ZN_1I_NM}
\end{figure}

\begin{figure}[h]
    \centering
    \includegraphics[width=0.5\linewidth]{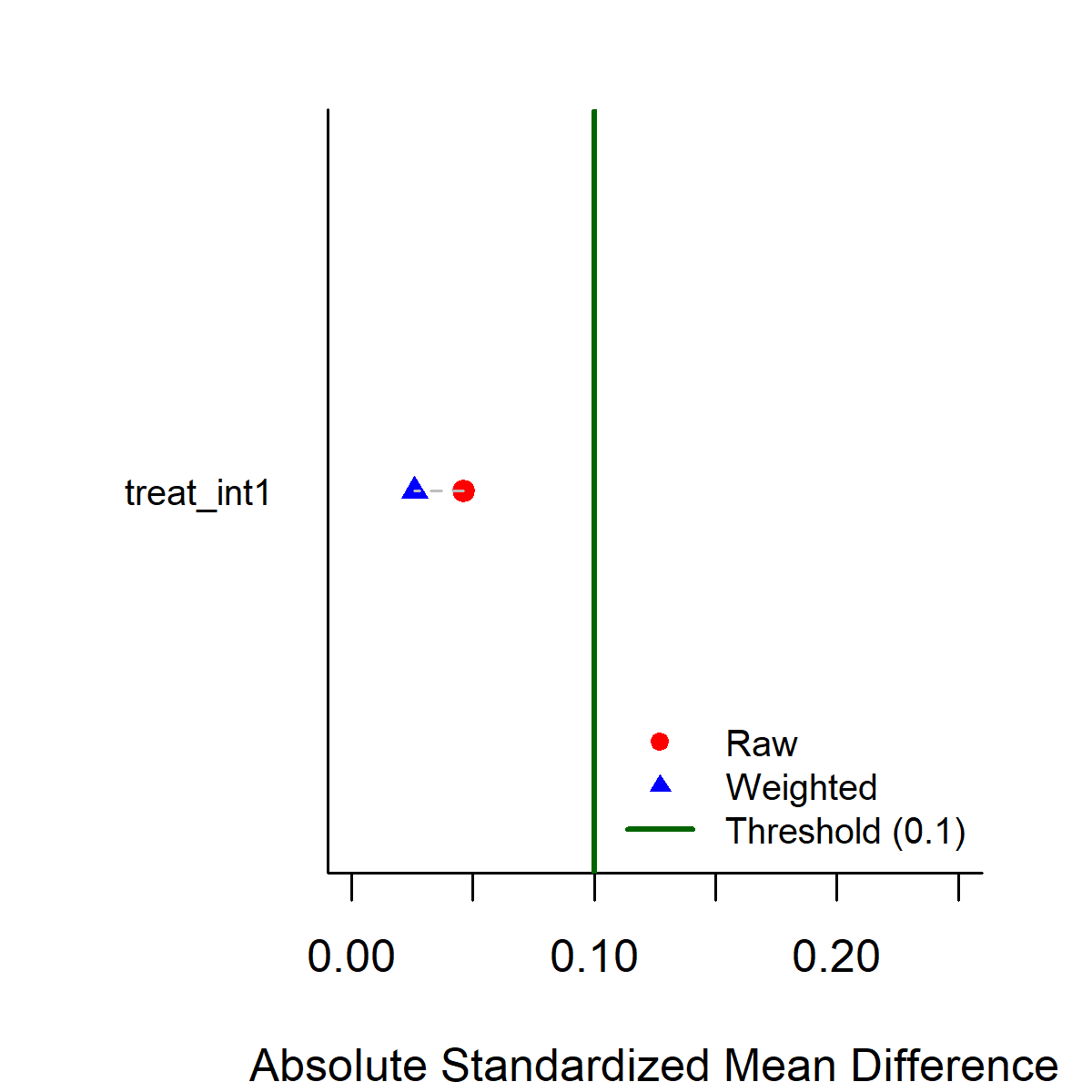}
    \caption{Target Covariate Balance for $\omega_{11^*}^{(Z,N)}$ (\ref{eq:Weight_Z_Bridge_PW_2})}
    \label{fig:Weight_ZN_1I_ZB}
\end{figure}

\begin{figure}[h]
    \centering
    \includegraphics[width=0.5\linewidth]{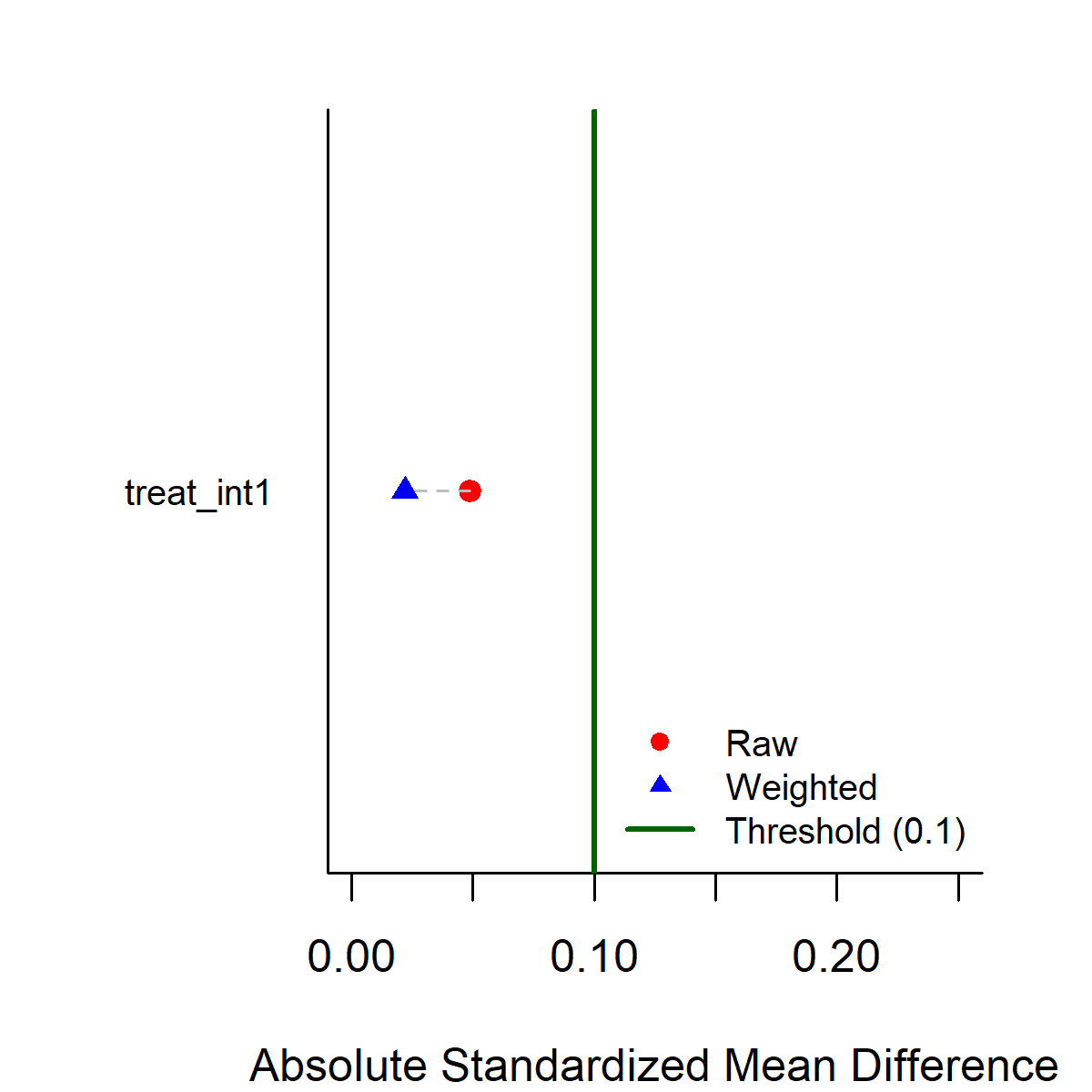}
    \caption{Target Covariate Balance for $\omega_{11^*}^{(Z,N)}$ (\ref{eq:Weight_N_Bridge_PW})}
    \label{fig:Weight_ZN_1I_NB}
\end{figure}

\begin{figure}[h]
    \centering
    \includegraphics[width=0.5\linewidth]{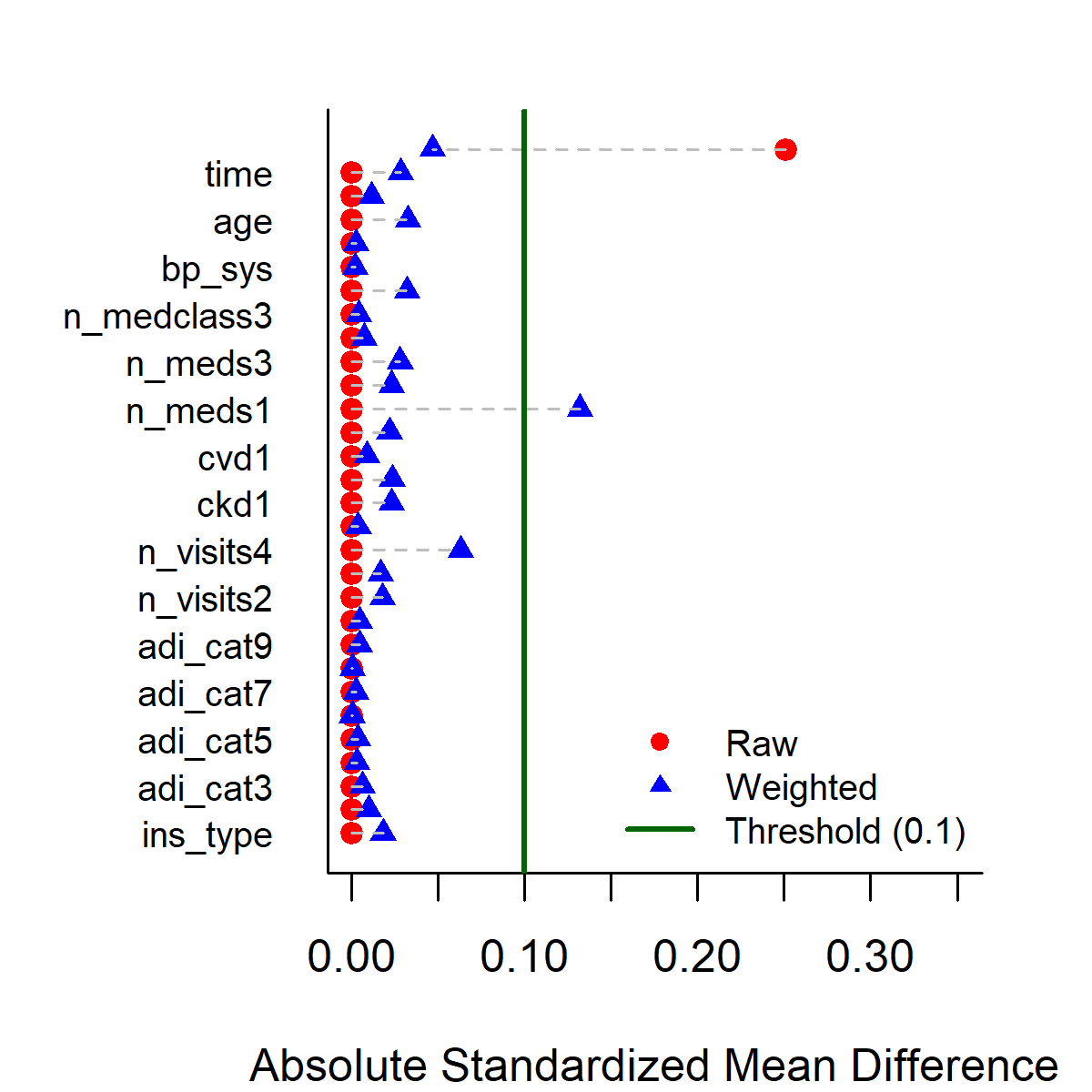}
    \caption{Target Covariate Balance for the first component of $\omega_{11^*}^{(Z,N)}$ 
    (\ref{eq:Weight_Z_Bridge_PW_2}) in odds terms}
    \label{fig:Weight_ZN_Zbridge_part_i}
\end{figure}

\begin{figure}[h]
    \centering
    \includegraphics[width=0.5\linewidth]{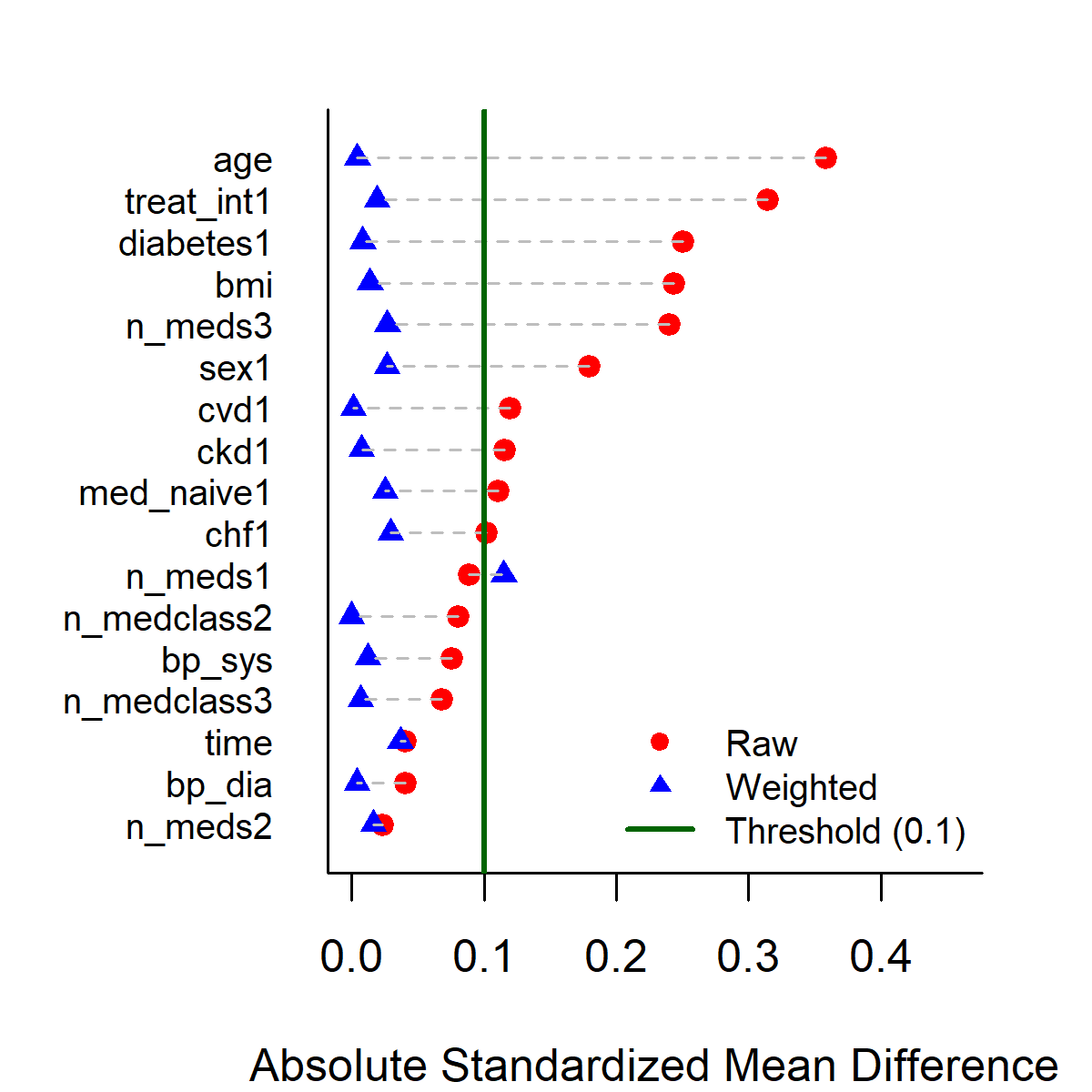}
    \caption{Target Covariate Balance for the second component (numerator) $\omega_{11^*}^{(Z,N)}$ (\ref{eq:Weight_Z_Bridge_PW_2}) in odds terms} 
    \label{fig:Weight_ZN_Zbridge_part_ii}
\end{figure}

\begin{figure}[h]
    \centering
    \includegraphics[width=0.5\linewidth]{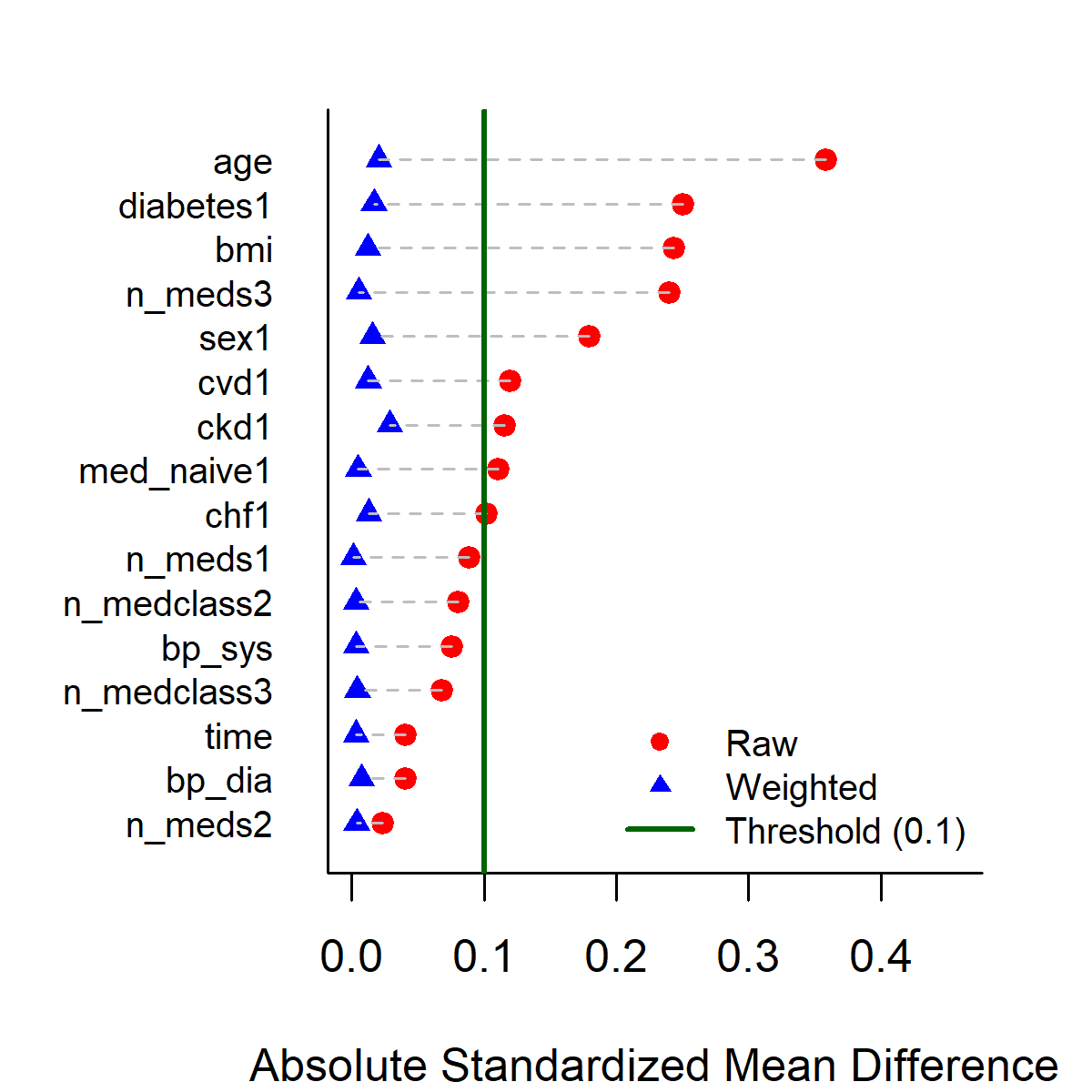}
    \caption{Target Covariate Balance for the second component (denominator) $\omega_{11^*}^{(Z,N)}$ (\ref{eq:Weight_Z_Bridge_PW_2}) in odds terms}
    \label{fig:Weight_ZN_Zbridge_part_iii}
\end{figure}

\begin{figure}[h]
    \centering
    \includegraphics[width=0.5\linewidth]{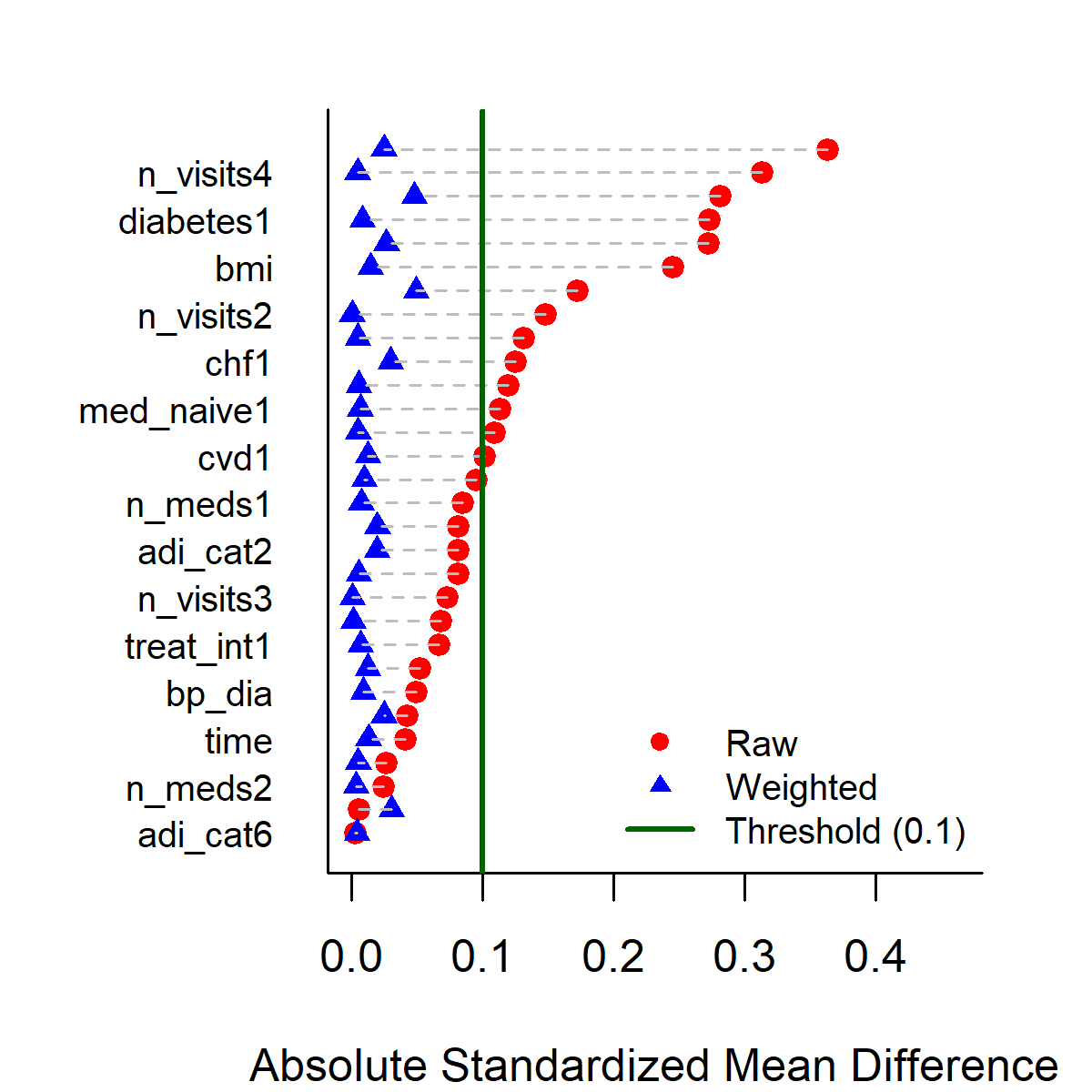}
    \caption{Target Covariate Balance for $\omega_{11^*}^{(Z,N)}$ (\ref{eq:Weight_N_Bridge_PW}) as expressed in odds terms (numerator)}
    \label{fig:Weight_ZN_Nbridge_part_i}
\end{figure}

\begin{figure}[h]
    \centering
    \includegraphics[width=0.5\linewidth]{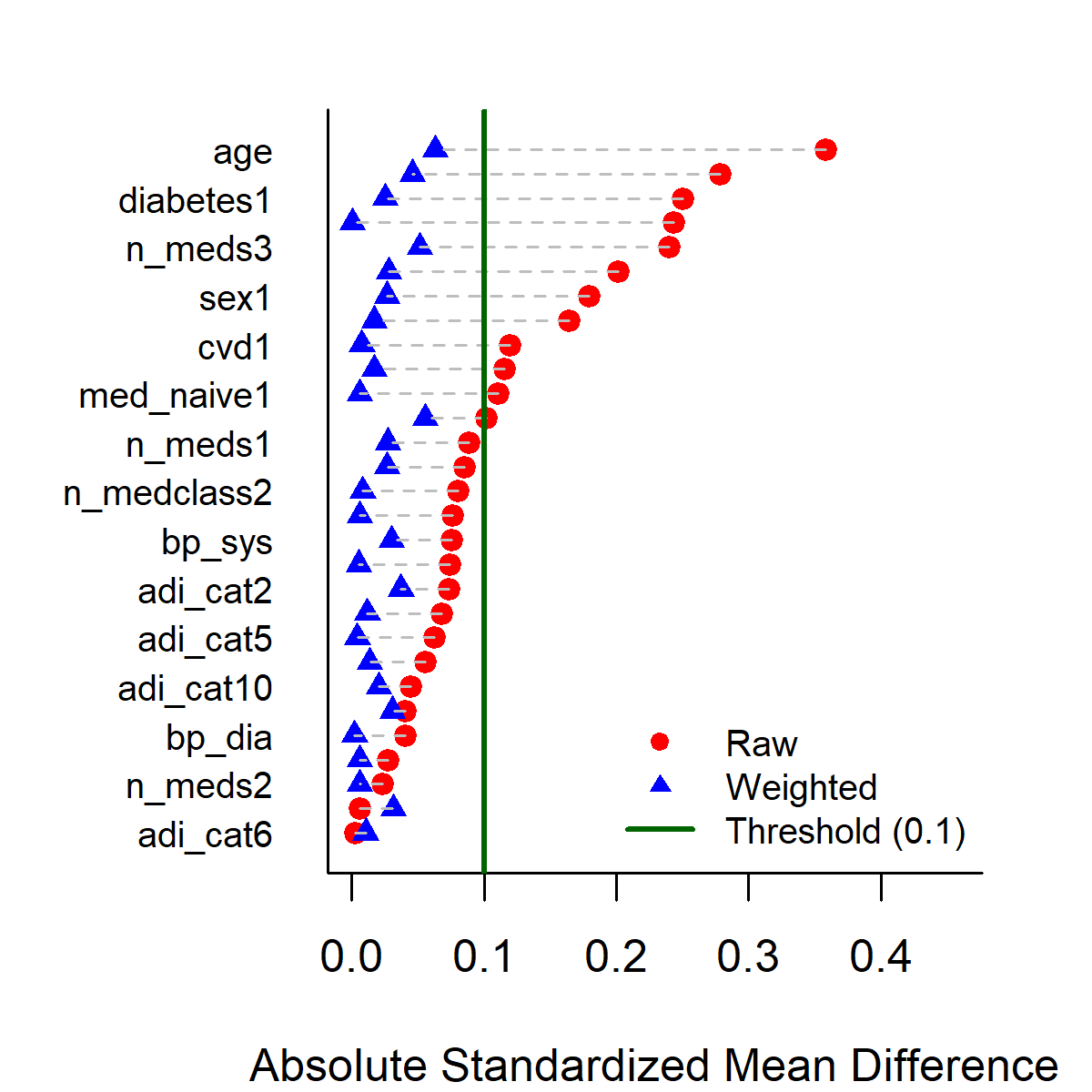}
    \caption{Target Covariate Balance for $\omega_{11^*}^{(Z,N)}$ (\ref{eq:Weight_N_Bridge_PW}) as expressed in odds terms (denominator)}
    \label{fig:Weight_ZN_Nbridge_part_ii}
\end{figure}

\begin{figure}[h]
    \centering
    \includegraphics[width=0.5\linewidth]{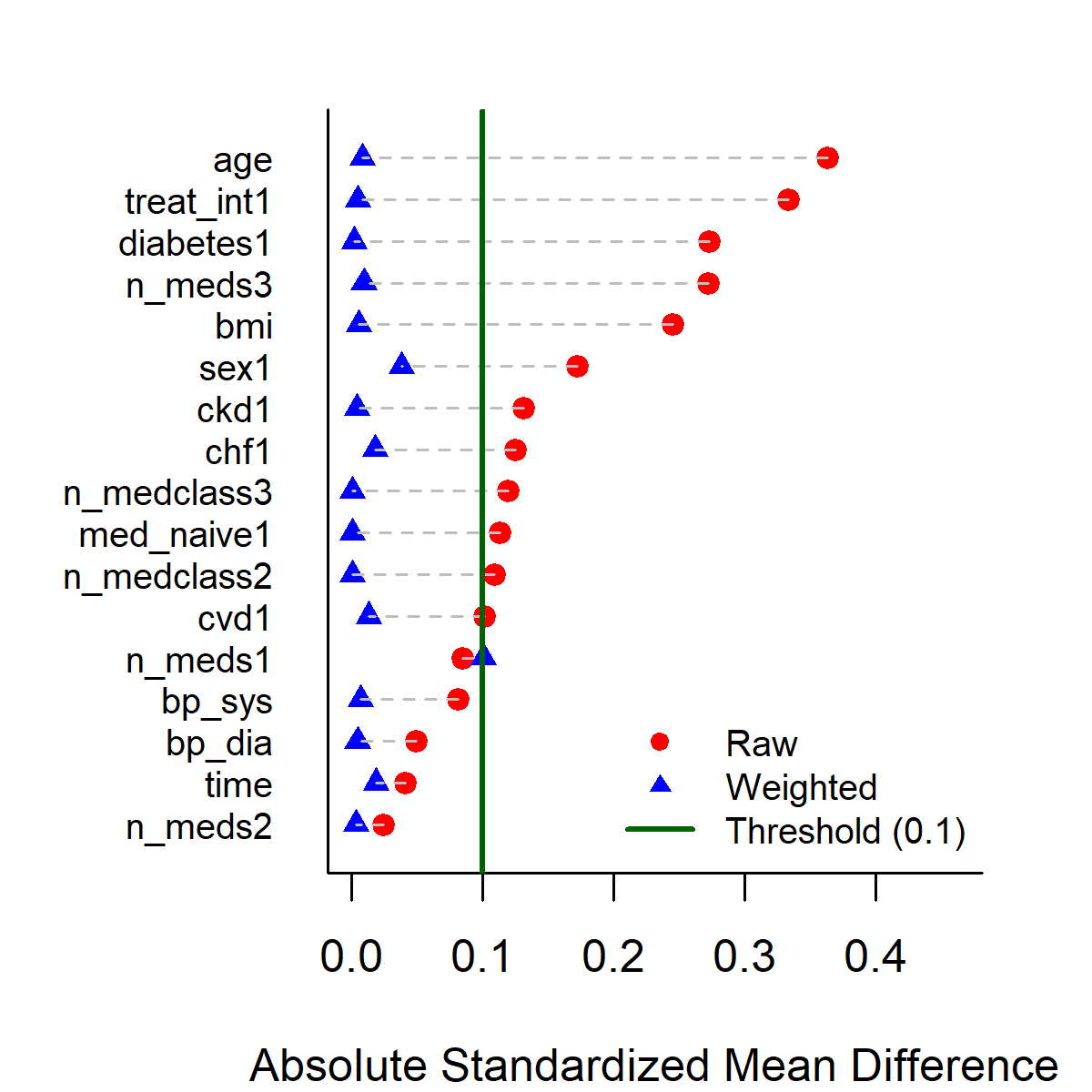}
    \caption{Target Covariate Balance for $\omega_{\blackdiamond 1^*}^{Z}$ (\ref{eq:Weight_Z_Bridge_RW}) as expressed in odds terms (numerator)}
    \label{fig:Weight_ZN_dI_part_i}
\end{figure}

\begin{figure}[h]
    \centering
    \includegraphics[width=0.5\linewidth]{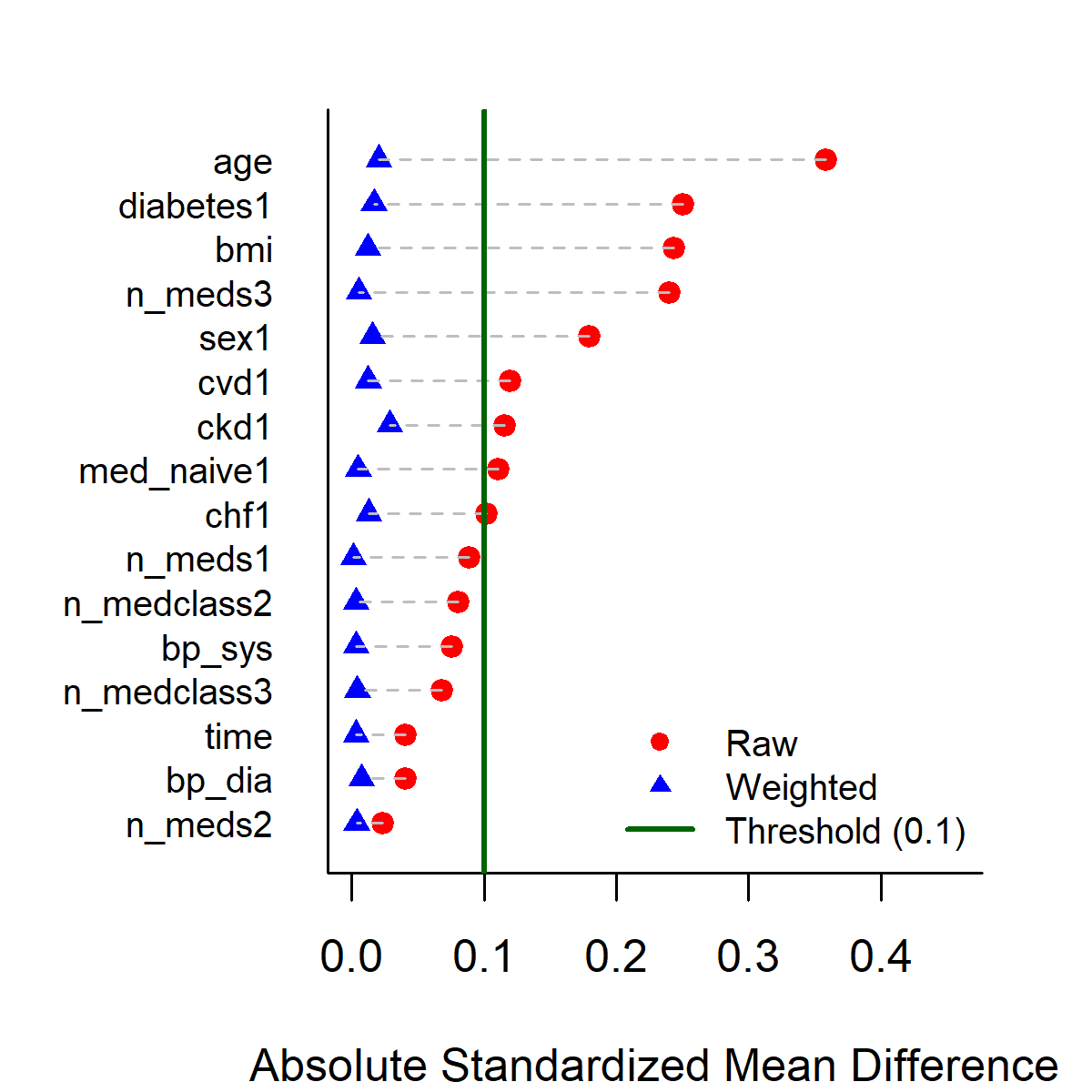}
\caption{Target Covariate Balance for $\omega_{\blackdiamond 1^*}^{Z}$ (\ref{eq:Weight_Z_Bridge_RW}) as expressed in odds terms (denominator)}
    \label{fig:Weight_ZN_dI_part_ii}
\end{figure}

\begin{figure}[h]
    \centering
    \includegraphics[width=0.5\linewidth]{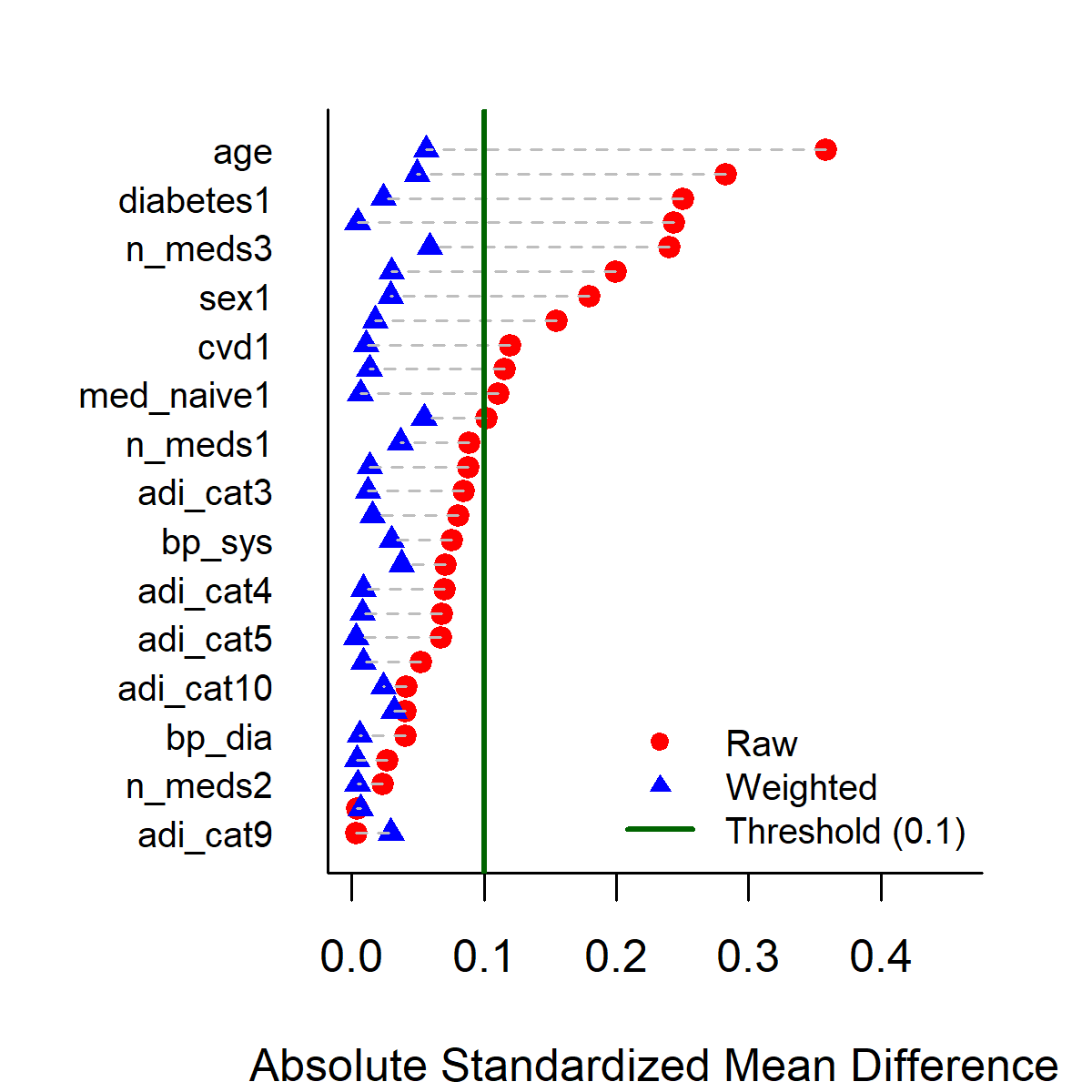}
    \caption{Target Covariate Balance for $\omega_{\diamond 1^*}^{(N,A_z)}$ (\ref{eq:Weight_N_Bridge_RW}) as expressed in odds terms (numerator)}
    \label{fig:Weight_ZNAz_dI_part_i}
\end{figure}

\begin{figure}[h]
    \centering
    \includegraphics[width=0.5\linewidth]{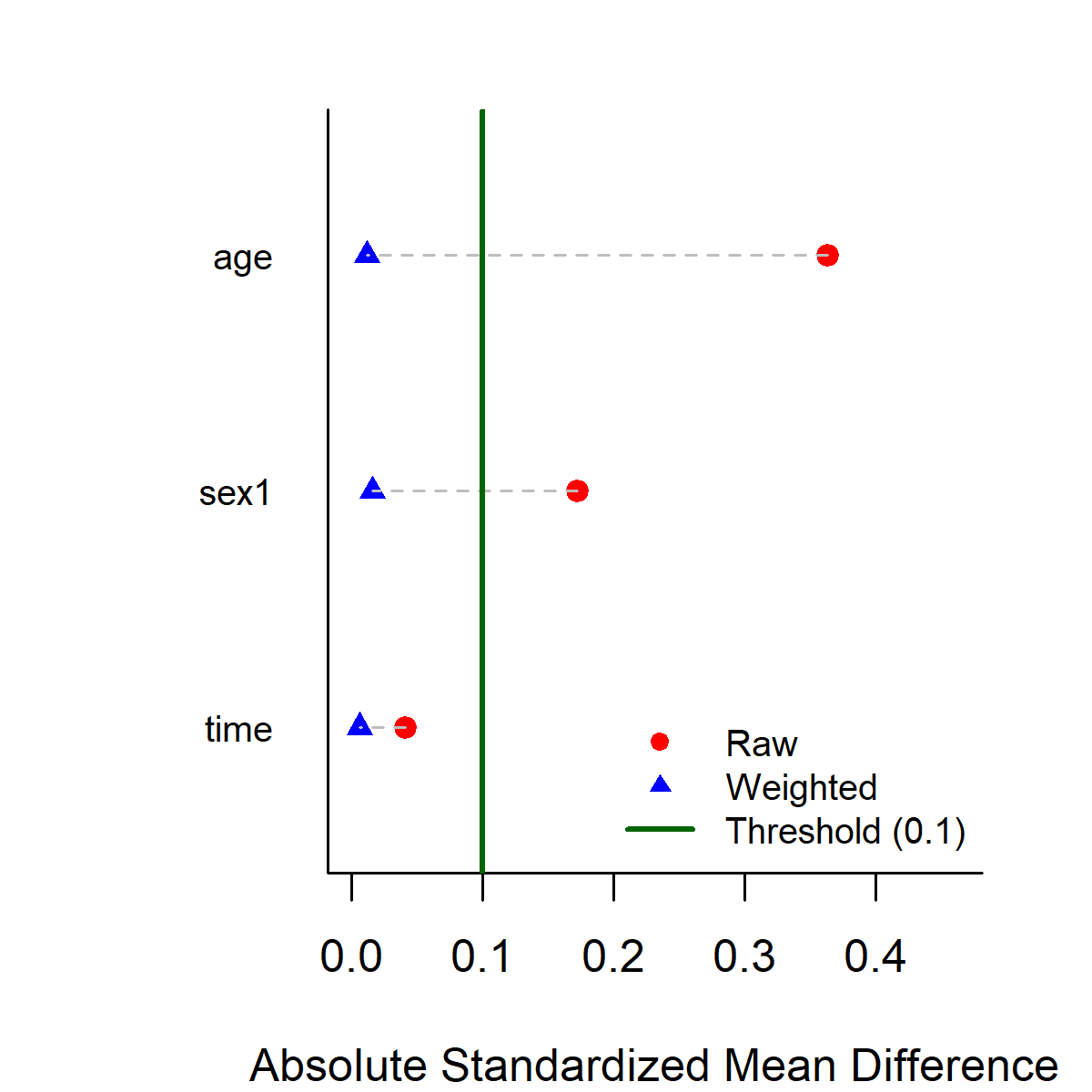}
    \caption{Target Covariate Balance for $\omega_{\diamond 1^*}^{(N,A_z)}$ (\ref{eq:Weight_N_Bridge_RW}) as expressed in odds terms (denominator)}
    \label{fig:Weight_ZNAz_dI_part_ii}
\end{figure}

\begin{figure}[h]
    \centering
    \includegraphics[width=0.5\linewidth]{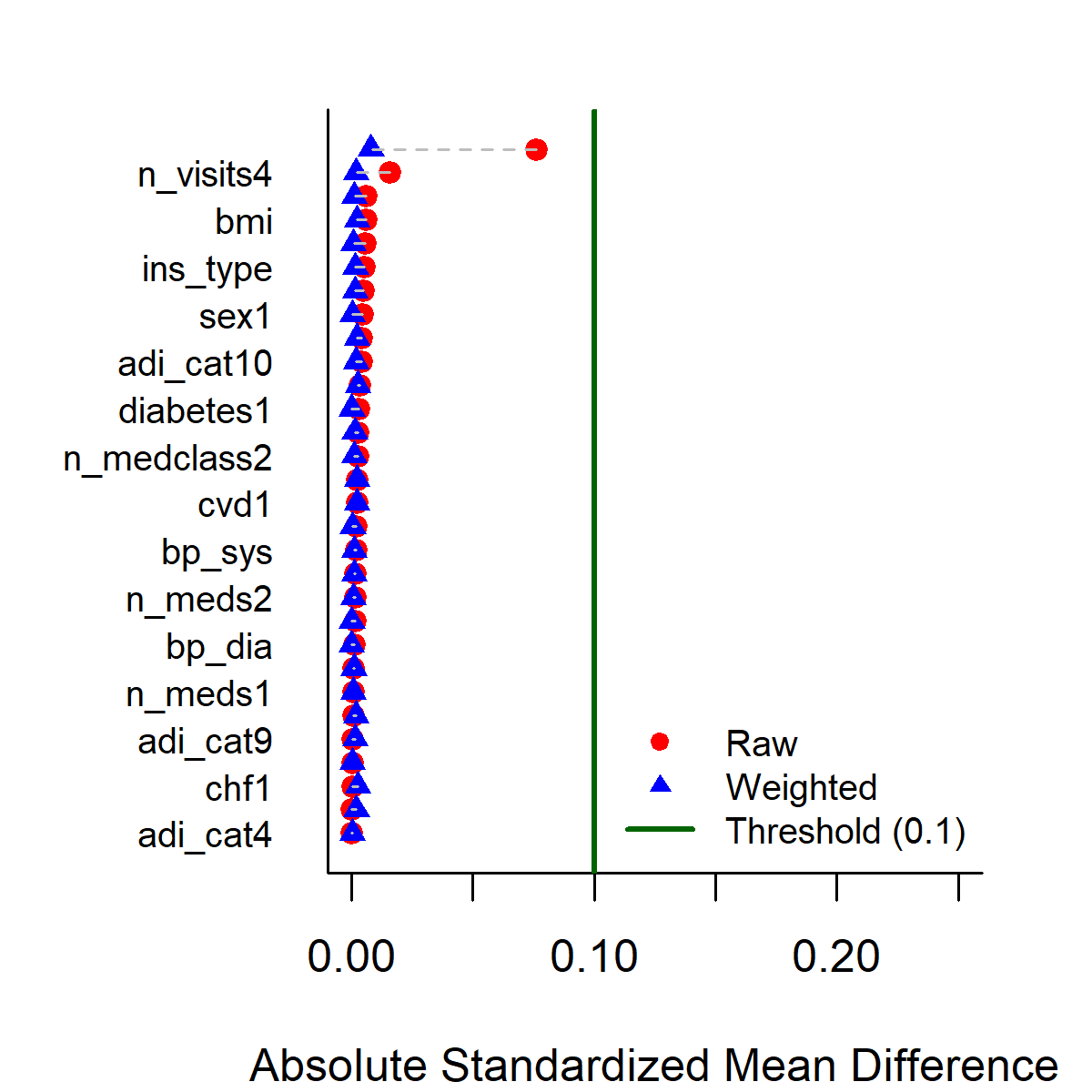}
    \caption{Target Covariate Balance for $\omega_{1}^C$ denominator}
    \label{fig:Weight_C_0_G1}
\end{figure}

\begin{figure}[h]
    \centering
    \includegraphics[width=0.5\linewidth]{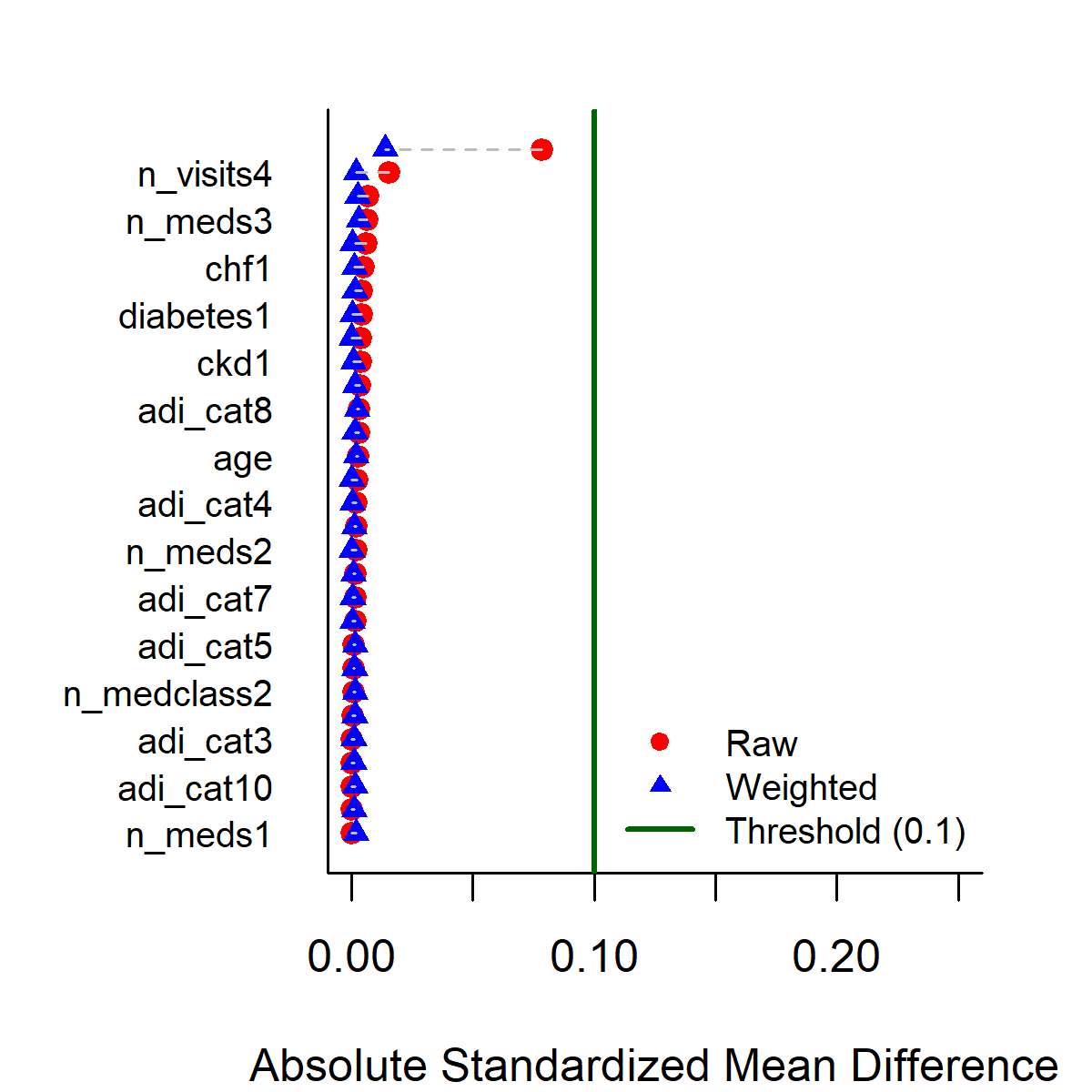}
    \caption{Target Covariate Balance for $\omega_{0}^C$ denominator}
    \label{fig:Weight_C_0_G0}
\end{figure}

\clearpage

\subsection{Proofs}

\subsubsection{Commonly Used Definitions} \label{subsubsec:Common_Definitions}

The proofs rely on the following shorthand notation, which express weighting functions:
\begin{align}
\omega_{11^*}^{(Z,N)}(Z,N,A_z,A_y)  &\coloneq \frac{\P_1^*(Z,N|A_z,A_y)}{\P_1(Z,N|A_z,A_y)} =\frac{\P_0(Z|A_z,A_y)\P_1(N_1|A_z,A_y)}{\P_1(Z,N|A_z,A_y)},\nonumber\\
\omega_{\blackdiamond 1^*}^{Z}(Z,A_z,A_y)&\coloneq\frac{\P_1^*(Z|A_z,A_y)}{\P_\blackdiamond(Z|A_z,A_y)}=\frac{\P_0(Z|A_z,A_y)}{\P_\blackdiamond(Z|A_z,A_y)}, \nonumber \\
\omega_{\diamond 1^*}^{(N,A_z)}(N,A_z,A_y)&\coloneq\frac{\P_1^*(N,A_z|A_y)}{\P_\diamond(N,A_z|A_y)}=\frac{\P_1(N,A_z|A_y)}{\P_\diamond(N,A_z|A_y)}, \nonumber \\
\omega_{01}^{A_z}(A_z,A_y) &\coloneq\frac{\P_1^*(A_z|A_y)}{\P_0(A_z|A_y)}=\frac{\P_1(A_z|A_y)}{\P_0(A_z|A_y)}, \nonumber \\
\omega_{gT}^{A_y}(A_y) &\coloneq\frac{\P_\std(A_y)}{\P_g(A_y)}, \\\nonumber 
\omega_{01^*}^{(Z,N,A_z)}(Z,N,A_z,A_y)&\coloneq\frac{\P_1^*(Z,N,A_z|A_y)}{\P_0(Z,N,A_z|A_y)}=\frac{\P_0(Z|A_z,A_y)\P_1(N,A_z|A_y)}{\P_0(Z,N,A_z|A_y)}. \nonumber 
\end{align}


The proofs also rely on the following shorthand notation, which express conditional sequential expectations:
\begin{align}
 \mu_1(Z,N,A_z,A_y) &\coloneq \E_1[Y|Z,N,A_z,A_y] \nonumber
\end{align}
\begin{align}
 \zeta_1^*(Z,A_z,A_y) 
 &\coloneq \E_1^*[\mu_1(Z,N,A_z,A_y)\mid Z,A_z,A_y]=\int\mu_1(Z,n,A_z,A_y)\P_1^*(n|Z,A_z,A_y)dn \nonumber \\
 &=\int\mu_1(Z,n,A_z,A_y)\P_1(n|Z,A_z,A_y)dn, \nonumber
\end{align}
\begin{align}
 \nu_1^*(N,A_z,A_y) 
 &\coloneq \E_1^*[\mu_1(Z,N,A_z,A_y)\mid N,A_z,A_y]=\int\mu_1(z,N,A_z,A_y)\P_1^*(z|N,A_z,A_y)dz \nonumber\\
 &=\int\mu_1(z,N,A_z,A_y)\P_0(z|A_z,A_y)dz, \nonumber
\end{align}
\begin{align}
 \kappa_1^*(A_z,A_y) 
 &\coloneq \E_1^*[\mu_1(Z,N,A_z,A_y)\mid A_z,A_y]=\int\mu_1(z,n,A_z,A_y)\P_1^*(z,n|A_z,A_y)dzdn \nonumber \\
 &=\int\mu_1(z,n,A_z,A_y)\P_0(z|A_z,A_y)\P_1(n|A_z,A_y)dzdn, \nonumber 
\end{align}
\begin{align}
 \eta_1^*(A_y) 
 &\coloneq \E_1^*[\mu_1(Z,N,A_z,A_y)\mid A_y]=\int\mu_1(z,n,a_z,A_y)\P_1^*(z,n,a_z|A_y)dzdnda_z \nonumber \\ 
 &=\int\mu_1(z,n,a_z,A_y)\P_0(z|a_z,A_y)\P_1(n,a_z|A_y)dzdnda_z.\nonumber 
\end{align}

At times, we will use $\P(b|c)$ to abbreviate $\P(B=b|C=c)$ and $\E[y|x]$ to abbreviate $\E[Y=y|X=x]$.

\newpage

\subsubsection{Proof of Linear Estimator (E-OBD)}

Given the linear causal model
\begin{align}
	\E_1[Y|Z,N,A_z,A_y]=\beta_0+\beta_z Z+\beta_n N+\beta_{a_z} A_z + \beta_{a_y}A_y \tag{\ref{eq:Y_Model_Linear}},
\end{align}
we have that 
\begin{align}
\theta_1&=\textstyle\sum_{y,z,n,a_z,a_y}y\P_1(y|z,n,z,a_z,a_y)\P_1(z|n,a_z,a_y)\P_1(n|a_z,a_y)\P_1(a_z|a_y)\P_\std(a_y) \nonumber\\
&=\textstyle\sum_{z,n,a_z,a_y}(\beta_0+\beta_z z+\beta_n n+\beta_{a_z}a_z+\beta_{a_y}a_y)\P_1(z|n,a_z,a_y)\P_1(n|a_z,a_y)\P_1(a_z|a_y)\P_\std(a_y)\frac{P_1(A_y=a_y)}{\P_1(A_y=a_y)} \nonumber \\
&=\beta_0+\beta_z\E_1[\omega_{1T}^{A_y}Z]+\beta_n\E_1[\omega_{1T}^{A_y}N]+\beta_{a_z}\E_1[\omega_{1T}^{A_y}A_z]+\beta_{a_y}\E_1[\omega_{1T}^{A_y}A_y] \nonumber
\end{align}
Note that we can also re-express $\theta_1$ as
\begin{align}
\theta_1&=\textstyle \sum_{y,a_y}y\P_1(y|a_y)\P_\std(a_y) = \textstyle \sum_{a_y}\E_1[Y|A_y=a_y]\P_1(a_y)\frac{\P_\std(a_y)}{\P_1(a_y)} = E_1[\omega_{1T}^{A_y}Y] \nonumber
\end{align}
We also have that
\begin{align}
\theta_1^*&=\textstyle\sum_{y,z,n,a_z,a_y}y\P_1(y|z,n,a_z,a_y)\P_0(z|a_z,a_y)\P_1(n|a_z,a_y)\P_1(a_z|a_y)\P_\std(a_y) \nonumber\\
&=\textstyle\sum_{z,n,a_z,a_y}(\beta_0+\beta_z z+\beta_n n+\beta_{a_z}a_z+\beta_{a_y}a_y)\P_0(z|a_z,a_y)\P_1(n|a_z,a_y)\P_1(a_z|a_y)\frac{\P_0(a_z|a_y)}{\P_0(a_z|a_y)}\P_\std(a_y)\frac{\P_1(a_y)}{\P_1(a_y)}\frac{\P_0(a_y)}{\P_0(a_y)} \nonumber \\ 
&=\beta_0+\beta_z\E_0[\omega_{01}^{A_z}\omega_{0T}^{A_y}Z]+\beta_n\E_1[\omega_{1T}^{A_y}N]+\beta_{a_z}\E_1[\omega_{1T}^{A_y}A_z]+\beta_{a_y}\E_1[\omega_{1T}^{A_y}A_y] \nonumber
\end{align}
Thus,
\begin{align}
\theta_1^* - \theta_1 = \beta_z(\E_0[\omega_{01}^{A_z}\omega_{0T}^{A_y}Z]-\E_1[\omega_{1T}^{A_y}Z]). \nonumber
\end{align}
Because $\theta_1=\E_1[\omega_{1T}^{A_y}Y]$, it follows that
\begin{align}
  \theta_1^*&=\E_1[\omega_{1T}^{A_y}Y]+\beta_z(\E_0[\omega_{01}^{A_z}\omega_{0T}^{A_y}Z]-\E_1[\omega_{1T}^{A_y}Z]). \nonumber
\end{align}

There is an alternate linear estimator, different from the main text, that can be derived upon nothing that
\begin{align}
\theta_1^*&=\textstyle\sum_{y,z,n,a_z,a_y}y\P_1(y|z,n,a_z,a_y)\P_0(z|a_z,a_y)\P_1(n|a_z,a_y)\P_1(a_z|a_y)\P_\std(a_y) \nonumber\\
&=\textstyle\sum_{z,n,a_z,a_y}(\beta_0+\beta_z z+\beta_n n+\beta_{a_z}a_z+\beta_{a_y}a_y)\P_0(z|a_z,a_y)\frac{\P_1(z|a_z,a_y)}{\P_1(z|a_z,a_y)}\P_1(n|a_z,a_y)\P_1(a_z|a_y)\P_\std(a_y)\frac{\P_1(a_y)}{\P_1(a_y)} \nonumber \\ 
&=\beta_0+\beta_z\E_1[\omega_{11^*}^{Z}\omega_{1T}^{A_y}Z]+\beta_n\E_1[\omega_{1T}^{A_y}N]+\beta_{a_z}\E_1[\omega_{1T}^{A_y}A_z]+\beta_{a_y}\E_1[\omega_{1T}^{A_y}A_y] \nonumber
\end{align}
Then, the alternate linear estimator would be based on the following expression of $\theta_1^*$:
\begin{align}
  \theta_1^*&=\E_1[\omega_{1T}^{A_y}Y]+\beta_z(\E_1[\omega_{11^*}^{Z}\omega_{1T}^{A_y}Z]-\E_1[\omega_{1T}^{A_y}Z]). \nonumber
\end{align}

Because $\omega_{11^*}^{Z}(Z,A_z,A_y)=\frac{\P_0(z|a_z,a_y)}{\P_1(z|a_z,a_y)}$ can be expressed in terms of odds, i.e., $\omega_{11^*}^{Z}(Z,A_z,A_y)=\frac{\odds(G=1|A_z,A_y)}{\odds(G=1|Z,A_z,A_y)}$, this suggests an alternate linear estimator that, like the linear estimator of the main text, does not require modeling the density of $Z$. Unlike the linear estimator of the main text, this alternative does not reduce to a causal version of the detailed Oaxaca-Blinder decomposition when there are no allowable covariates $A_z$ or $A_y$.

\newpage

\subsubsection{Proof of Pure Weighting Estimator Expressions}

\begin{align}
    \theta_1^*&=\textstyle \sum_{y,z,n,a_z,a_y}y\P_1(y | z,n,a_z,a_y)\P_0(z|a_z,a_y)\P_1(n|a_z,a_y)\P_1(a_z|a_y) \P_\std(a_y) \nonumber\\
    &=\textstyle \sum_{y,z,n,a_z,a_y}y\P_1(y | z,n,a_z,a_y)\P_1(z|n,a_z,a_y)\P_1(n|a_z,a_y)\P_1(a_z|a_y) \P_1(a_y)\nonumber \\ &~~~~~~~~~~~~~~~~~~~~~~\times \frac{\P_0(z|a_z,a_y)\P_1(n|a_z,a_y)}{\P_1(z|n,a_z,a_y)\P_1(n|a_z,a_y)}\times\frac{\P_\std(a_y)}{\P_1(a_y)}\nonumber\\
    &=\textstyle \sum_{y,z,n,a_z,a_y}y\P_1(y | z,n,a_z,a_y)\P_1(z,n|a_z,a_y)\P_1(a_z|a_y) \P_1(a_y) \nonumber \\ &~~~~~~~~~~~~~~~~~~~~~~\times \frac{\P_0(z|a_z,a_y)\P_1(n|a_z,a_y)}{\P_1(z,n|a_z,a_y)}\times\frac{\P_\std(a_y)}{\P_1(a_y)}\nonumber\\
    &=\E_1\Biggl[\omega_{11^*}^{(Z,N)}\omega_{1T}^{A_y} Y \Bigg] \nonumber
\end{align}
where 
\begin{align}
\omega_{11^*}^{(Z,N)}(Z,N,A_z,A_y) &\coloneq \frac{\P_1^*(Z,N|A_z,A_y)}{\P_1(Z,N|A_z,A_y)}\nonumber \\&=\frac{\P_0(Z|A_z,A_y)\P_1(N|A_z,A_y)}{\P_1(Z,N|A_z,A_y)} \nonumber
\end{align}
and $\omega_{1T}^{A_y}(A_y)$ is defined in Section \ref{subsubsec:Common_Definitions} above.

\textbf{Z-Model-PW}

The Z-Model-PW expression follows by substituting in the derivations of the weights $\omega_{11^*}^{(Z,N)}$ found in (\ref{eq:Weight_Z_Model_PW}). 

\textbf{N-Model-PW}

The N-Model-PW expression follows by substituting in the derivations of the weights $\omega_{11^*}^{(Z,N)}$ found in (\ref{eq:Weight_N_Model_PW}). 

\textbf{Z-Bridge-PW}

The Z-Bridge-PW expressions follow by substituting in the derivations of the weights $\omega_{11^*}^{(Z,N)}$ found in (\ref{eq:Weight_Z_Bridge_PW_2}). 

\textbf{N-Bridge-PW}

The N-Bridge-PW expression follows by substituting in the derivations of the weights $\omega_{11^*}^{(Z,N)}$ found in (\ref{eq:Weight_N_Bridge_PW}).

\newpage

\subsubsection{Proof of Regress-then-Weight Estimator Expressions}

\textbf{Z-Model-RW} (version 1)
\begin{align}
    \theta_1^*&=\textstyle \sum_{y,z,n,a_z,a_y}y\P_1(y | z,n,a_z,a_y)\P_0(z|a_z,a_y)\P_1(n|a_z,a_y)\P_1(a_z|a_y) \P_\std(a_y) \nonumber\\
    &=\textstyle \sum_{z,n,a_z,a_y}\mu_1(z,n,a_z,a_y)\P_1(z|n,a_z,a_y)\P_1(n|a_z,a_y)\P_1(a_z|a_y) \P_1(a_y)\nonumber \\ &~~~~~~~~~~~~~~~~~~~~~~\times \frac{\P_0(z|a_z,a_y)\P_1(n|a_z,a_y)}{\P_1(z|n,a_z,a_y)\P_1(n|a_z,a_y)}\times\frac{\P_\std(a_y)}{\P_1(a_y)}\nonumber\\
    &=\textstyle \sum_{z,n,a_z,a_y}\mu_1(z,n,a_z,a_y)\P_1(z,n|a_z,a_y)\P_1(a_z|a_y) \P_1(a_y) \nonumber \\ &~~~~~~~~~~~~~~~~~~~~~~\times \frac{\P_0(z|a_z,a_y)\P_1(n|a_z,a_y)}{\P_1(z,n|a_z,a_y)}\times\frac{\P_\std(a_y)}{\P_1(a_y)}\nonumber\\
    &=\E_1\Biggl[\omega_{11^*}^{(Z,N)}\omega_{1T}^{A_y} \mu_1(Z,N,A_z,A_y) \Bigg] \nonumber
\end{align}
where $\mu_1(Z,N,A_z,A_y)$ and the weighting functions $\omega_{11^*}^{(Z,N)}(Z,N,A_z,A_y)$ and $\omega_{1T}^{A_y}$ are defined in Section \ref{subsubsec:Common_Definitions} above. See (\ref{eq:Weight_Z_Model_PW}), (\ref{eq:Weight_N_Model_PW}), (\ref{eq:Weight_Z_Bridge_PW_2}), and (\ref{eq:Weight_N_Bridge_PW}) for alternative expressions of $\omega_{11^*}^{(Z,N)}(Z,N,A_z,A_y)$.

\textbf{Z-Model-RW} (version 2) and \textbf{N-Model-RW}
\begin{align}
    \theta_1^*&=\textstyle \sum_{y,z,n,a_z,a_y}y\P_1(y | z,n,a_z,a_y)\P_0(z|a_z,a_y)\P_1(n|a_z,a_y)\P_1(a_z|a_y) \P_\std(a_y) \nonumber\\
    &=\textstyle \sum_{z,n,a_z,a_y}\mu_1(z,n,a_z,a_y)\P_0(z|n,a_z,a_y)\P_0(n|a_z,a_y)\P_0(a_z|a_y) \P_0(a_y)\nonumber \\ &~~~~~~~~~~~~~~~~~~~~~~\times \frac{\P_0(z|a_z,a_y)\P_1(n|a_z,a_y)\P_1(a_z|a_y)}{\P_0(z|n,a_z,a_y)\P_0(n|a_z,a_y)\P_0(a_z|a_y)}\times\frac{\P_\std(a_y)}{\P_0(a_y)}\nonumber\\
    &=\textstyle \sum_{z,n,a_z,a_y}\mu_1(z,n,a_z,a_y)\P_0(z,n|a_z,a_y)\P_0(a_z|a_y) \P_0(a_y) \nonumber \\ &~~~~~~~~~~~~~~~~~~~~~~\times \frac{\P_0(z|a_z,a_y)\P_1(n|a_z,a_y)\P_1(a_z|a_y)}{\P_0(z,n,a_z|a_y)}\times\frac{\P_\std(a_y)}{\P_0(a_y)}\nonumber\\
    &=\E_1\Biggl[\omega_{01^*}^{(Z,N,A_z)}\omega_{0T}^{A_y} \mu_1(Z,N,A_z,A_y) \Bigg] \nonumber
\end{align}
where 
\begin{align}
\omega_{01^*}^{(Z,N,A_z)}(Z,N,A_z,A_y) &\coloneq \frac{\P_1^*(Z,N,A_z|A_y)}{\P_0(Z,N,A_z|A_y)}\nonumber \\&=\frac{\P_0(Z|A_z,A_y)\P_1(N|A_z,A_y)\P_1(A_z|A_y)}{\P_0(Z,N,A_z|A_y)} \nonumber
\end{align}
and $\mu_1(Z,N,A_z,A_y)$ and $\omega_{0T}^{A_y}(A_y)$ are defined in Section \ref{subsubsec:Common_Definitions} above.

The expression for Z-Model-RW (version 2) follows by substituting in the derivations of the weights $\omega_{01^*}^{(Z,N,A_z)}(Z,N,A_z,A_y)$ found in (\ref{eq:Weight_Z_Model_RW}).

The N-Model-RW estimator follows by substituting in the derivations of the weights $\omega_{01^*}^{(Z,N,A_z)}(Z,N,A_z,A_y)$ found in (\ref{eq:Weight_N_Model_RW}).

\newpage

\textbf{Z-Bridge-RW}
\begin{align}
    \theta_1^*&=\textstyle \sum_{y,z,n,a_z,a_y}y\P_1(y | z,n,a_z,a_y)\P_0(z|a_z,a_y)\P_1(n|a_z,a_y)\P_1(a_z|a_y) \P_\std(a_y) \nonumber\\
    &=\textstyle \sum_{z,n,a_z,a_y}\mu_1(z,n,a_z,a_y)\P_\blackdiamond(z|n,a_z,a_y)\P_\blackdiamond(n|a_z,a_y)\P_\blackdiamond(a_z|a_y) \P_\blackdiamond(a_y)\nonumber \\ &~~~~~~~~~~~~~~~~~~~~~~\times \frac{\P_0(z|a_z,a_y)\P_1(n|a_z,a_y)}{\P_\blackdiamond(z|n,a_z,a_y)\P_\blackdiamond(n|a_z,a_y)}\times\frac{\P_1(A_z|A_y)}{\P_\blackdiamond(A_z|A_y)}\times\frac{\P_\std(a_y)}{\P_{\blackdiamond}(a_y)}\nonumber\\
    &=\textstyle \sum_{z,n,a_z,a_y}\mu_1(z,n,a_z,a_y)\P_\blackdiamond(z,n|a_z,a_y)\P_\blackdiamond(a_z|a_y) \P_\blackdiamond(a_y) \nonumber \\ &~~~~~~~~~~~~~~~~~~~~~~\times \frac{\P_0(z|a_z,a_y)\P_1(n|a_z,a_y)}{\P_\blackdiamond(z,n|a_z,a_y)}\times\frac{\P_\std(a_y)}{\P_\blackdiamond(a_y)}\nonumber\\
    &=\E_\blackdiamond\Biggl[\omega_{\blackdiamond 1^*}^{Z}\omega_{\blackdiamond T}^{A_y} \mu_1(Z,N,A_z,A_y) \Bigg] \nonumber
\end{align}
where 
\begin{align}
\omega_{\blackdiamond 1^*}^{Z}(Z,A_z,A_y) &\coloneq \frac{\P_1^*(Z,N|A_z,A_y)}{\P_\blackdiamond(Z,N|A_z,A_y)}\nonumber \\&=\frac{\P_0(Z|A_z,A_y)\P_1(N|A_z,A_y)}{\P_\blackdiamond(Z,N|A_z,A_y)} \nonumber \\ &=\frac{\P_0(Z|A_z,A_y)}{\P_\blackdiamond(Z|A_z,A_y)} \nonumber ,
\end{align}
and $\omega_{\blackdiamond T}^{A_y}(A_y):=\frac{\P_\std(A_y)}{\P_\blackdiamond(A_y)}$ is proportional to $\omega_{1T}^{A_y}(A_y)$ which is defined above in Section \ref{subsubsec:Common_Definitions}. The Z-Bridge-RW estimator follows by substituting in the derivation of the weights $\omega_{\blackdiamond 1^*}^{Z}(Z,A_z,A_y)$ as found in (\ref{eq:Weight_Z_Bridge_RW}).

\textbf{N-Bridge-RW}
\begin{align}
    \theta_1^*&=\textstyle \sum_{y,z,n,a_z,a_y}y\P_1(y | z,n,a_z,a_y)\P_0(z|a_z,a_y)\P_1(n|a_z,a_y)\P_1(a_z|a_y) \P_\std(a_y) \nonumber\\
    &=\textstyle \sum_{z,n,a_z,a_y}\mu_1(z,n,a_z,a_y)\P_\diamond(z|n,a_z,a_y)\P_\diamond(n|a_z,a_y)\P_\diamond(a_z|a_y) \P_\diamond(a_y)\nonumber \\ &~~~~~~~~~~~~~~~~~~~~~~\times \frac{\P_0(z|a_z,a_y)\P_1(n|a_z,a_y)\P_1(a_z|a_y)}{\P_\diamond(z|n,a_z,a_y)\P_\diamond(n|a_z,a_y)\P_\diamond(a_z|a_y)}\times\frac{\P_\std(a_y)}{\P_\diamond(a_y)}\nonumber\\
    &=\textstyle \sum_{z,n,a_z,a_y}\mu_1(z,n,a_z,a_y)\P_\diamond(z,n|a_z,a_y)\P_\diamond(a_z|a_y) \P_\diamond(a_y) \nonumber \\ &~~~~~~~~~~~~~~~~~~~~~~\times \frac{\P_0(z|a_z,a_y)\P_1(n|a_z,a_y)\P_1(a_z|a_y)}{\P_\diamond(z,n,a_z|a_y)}\times\frac{\P_\std(a_y)}{\P_\diamond(a_y)}\nonumber\\
    &=\E_1\Biggl[\omega_{\diamond1^*}^{(N,A_z)}\omega_{\diamond T}^{A_y} \mu_1(Z,N,A_z,A_y) \Bigg] \nonumber
\end{align}
where 
\begin{align}
\omega_{\diamond1^*}^{(N,A_z)}(N,A_z,A_y) &\coloneq \frac{\P_1^*(Z,N,A_z|A_y)}{\P_\diamond(Z,N,A_z|A_y)}\nonumber \\&=\frac{\P_0(Z|A_z,A_y)\P_1(N|A_z,A_y)\P_1(A_z|A_y)}{\P_\diamond(Z,N,A_z|A_y)} \nonumber \\&=\frac{\P_1(N|A_z,A_y)\P_1(A_z|A_y)}{\P_\diamond(N,A_z|A_y)},
\end{align}
and $\omega_{\diamond T}^{A_y}(A_y):=\frac{\P_\std(A_y)}{\P_\diamond(A_y)}$ is proportional to $\omega_{0T}^{A_y}(A_y)$ which is defined above in Section \ref{subsubsec:Common_Definitions}. The N-Bridge-RW estimator follows by substituting in the derivation of the weights $\omega_{\diamond 1^*}^{(N,A_z)}(N,A_z,A_y)$ as found in (\ref{eq:Weight_N_Bridge_RW}).

\newpage

\subsubsection{Proof of Sequential Regression Expressions}

\textbf{Proof of Z-Model-SR}
\begin{align}
    \theta_1^*&=\textstyle \sum_{y,z,n,a_z,a_y}y\P_1(y | z,n,a_z,a_y)\P_0(z|a_z,a_y)\P_1(n|a_z,a_y)\P_1(a_z|a_y) \P_\std(a_y) \nonumber\\
    &=\textstyle \sum_{a_y}\underbrace{\textstyle\sum_{z,n,a_z}\mu_1(z,n,a_z,a_y)P_0(z|a_z,a_y)\P_1(n|a_z,a_y)\P_1(a_z|a_y)}_{\eta_1^*(a_y)}\P_\std(a_y)\nonumber \\
    &=\textstyle \sum_{a_y}\underbrace{\E_1^*\bigl[\mu_1(Z,N,A_z,a_y)|a_y\bigr]}_{\eta_1^*(a_y)}\P_\std(a_y)\nonumber \\
    &=\E_\std \bigl[\underbrace{\E_1^*[\mu_1(Z,N,A_z,A_y)|A_y]}_{\eta_1^*(A_y)}\bigr] \tag{\ref{eq:SeqExp_Z_Model_SR}}
\end{align}

\textbf{Proof of N-Model-SR}
\begin{align}
    \theta_1^*&=\textstyle \sum_{y,z,n,a_z,a_y}y\P_1(y | z,n,a_z,a_y)\P_0(z|a_z,a_y)\P_1(n|a_z,a_y)\P_1(a_z|a_y) \P_\std(a_y)\nonumber \\
    &=\textstyle \sum_{a_y} \underbrace{\textstyle \sum_{a_z}\underbrace{\textstyle\sum_{z,n}\mu_1(z,n,a_z,a_y)\P_0(z|a_z,a_y)\P_1(n|a_z,a_y)}_{\kappa_1^*(a_z,a_y)}\P_1(a_z|a_y)}_{\eta_1^*(a_y)}\P_\std(a_y)\nonumber \\
    &=\textstyle \sum_{a_y}\underbrace{\textstyle \sum_{a_z}\underbrace{\E_1^*\bigl[\mu_1( Z,N,a_z,a_y)|a_z,a_y\bigr]}_{\kappa_1^*(a_z,a_y)}\P_1(a_z|a_y)}_{\eta_1^*(a_y)}\P_\std(a_y)\nonumber \\
    &=\textstyle \sum_{a_y}\underbrace{\E_1\Bigl[\underbrace{\E_1^*\bigl[\mu_1( Z,N,A_z,a_y)|A_z,a_y\bigr]}_{\kappa_1^*(A_z,a_y)}\Big|a_y\Bigr]}_{\eta_1^*(a_y)}\P_\std(a_y)\nonumber \\
    &=\E_\std\biggl[\underbrace{\E_1\Bigl[\underbrace{\E_1^*\bigl[\mu_1(Z,N,A_z,A_y)|A_z,A_y\bigr]}_{\kappa_1^*(A_z,A_y)}\Big|A_y\Bigr]}_{\eta_1^*(A_y)}\biggr] \tag{\ref{eq:SeqExp_N_Model_SR}}
\end{align}

\textbf{Proof of Z-Bridge-SR}
\begin{align}
    \theta_1^*&=\textstyle \sum_{y,z,n,a_z,a_y}y\P_1(y | z,n,a_z,a_y)\P_0(z|a_z,a_y)\P_1(n|a_z,a_y)\P_1(a_z|a_y) \P_\std(a_y)  \nonumber\\
    &=\textstyle \sum_{a_y}\underbrace{\textstyle \sum_{a_z}\underbrace{\textstyle \sum_{z}\underbrace{ \textstyle \sum_{n}\mu_1( z,n,a_z,a_y)\P_1(n|a_z,a_y)}_{\zeta_1^*(z,a_z,a_y)}\P_0(z|a_z,a_y)}_{\kappa_1^*(a_z,a_y)}\P_1(a_z|a_y)}_{\eta_1^*(a_y)}\P_\std(a_y) \nonumber\\
    &=\textstyle \sum_{a_y}\underbrace{\textstyle \sum_{a_z}\underbrace{\textstyle \sum_{z}\underbrace{\E_\blackdiamond\bigl[\mu_1( z,N,a_z,a_y)|z,a_z,a_y\bigr]}_{\zeta_1^*(z,a_z,a_y)}\P_0(z|a_z,a_y)}_{\kappa_1^*(a_z,a_y)}\P_1(a_z|a_y)}_{\eta_1^*(a_y)}\P_\std(a_y) \nonumber\\
    &=\textstyle \sum_{a_y}\underbrace{\textstyle \sum_{a_z}\underbrace{\E_0\Bigl[\underbrace{\E_\blackdiamond\bigl[\mu_1( Z,N,A_z,A_y)|Z,a_z,a_y\bigr]}_{\zeta_1^*(Z,a_z,a_y)}\Big|a_z,a_y\Bigr]}_{\kappa_1^*(a_z,a_y)}\P_1(a_z|a_y)}_{\eta_1^*(a_y)}\P_\std(a_y) \nonumber\\
    &=\textstyle \sum_{a_y}\underbrace{\E_1\biggl[\underbrace{\E_0\Bigl[\underbrace{\E_\blackdiamond\bigl[\mu_1( Z,N,A_z,a_y)|Z,A_z,a_y\bigr]}_{\zeta_1^*(Z,A_z,a_y)}\Big|A_z,a_y\Bigr]}_{\kappa_1^*(A_z,a_y)}\biggr|a_y\biggr]}_{\eta_1^*(a_y)}\P_\std(a_y) \nonumber \\
    &=\E_\std\Biggl[\underbrace{\E_1\biggl[\underbrace{\E_0\Bigl[\underbrace{\E_\blackdiamond\bigl[\mu_1( Z,N,A_z,A_y)|Z,A_z,A_y\bigr]}_{\zeta_1^*(Z,A_z,A_y)}\Big|A_z,A_y\Bigr]}_{\kappa_1^*(A_z,A_y)}\biggr|A_y\biggr]}_{\eta_1^*(A_y)}\Biggr] \tag{\ref{eq:Seq_Exp_Z_Bridge_SR}}
\end{align}

The third equality follows by definition of the artificial group $G=\blackdiamond$ in Definition \ref{def:Group_Z_Bridge}.

\textbf{Proof of N-Bridge-SR}
\begin{align}
    \theta_1^*&=\textstyle \sum_{y,z,n,a_z,a_y}y\P_1(y | z,n,a_z,a_y)\P_0(z|a_z,a_y)\P_1(n|a_z,a_y)\P_1(a_z|a_y) \P_\std(a_y) \nonumber\\
    &=\textstyle \sum_{a_y}\underbrace{\textstyle \sum_{n,a_z} \underbrace{\textstyle \sum_{z}\mu_1(z,n,a_z,a_y)\P_0(z|a_z,a_y)}_{\nu_1^*(n,a_z,a_y)}\P_1(n|a_z,a_y)\P_1(a_z|a_y)}_{\eta_1^*(a_y)}\P_\std(a_y) \nonumber\\
    &=\textstyle \sum_{a_y}\underbrace{\textstyle \sum_{n,a_z}\textstyle \underbrace{\E_\diamond\bigl[\mu_1(Z,n,a_z,a_y)|n,a_z,a_y\bigr]}_{\nu_1^*(n,a_z,a_y)}\P_1(n|a_z,a_y)\P_1(a_z|a_y)}_{\eta_1^*(a_y)}\P_\std(a_y) \nonumber\\
    &=\textstyle \sum_{a_y}\underbrace{\E_1\Bigl[\underbrace{\E_\diamond\bigl[\mu_1(Z,N,A_z,a_y)|N,A_z,a_y\bigr]}_{\nu_1^*(N,A_z,a_y)}\Big|a_y\Bigr]}_{\eta_1^*(a_y)}\P_\std(a_y) \nonumber\\
    &=\E_\std\biggr[\underbrace{\E_1\Bigl[\underbrace{\E_\diamond\bigl[\mu_1(Y | Z,N,A_z,A_y)|N,A_z,A_y\bigr]}_{\nu_1^*(N,A_z,A_y)}\Big|A_y\Bigr]}_{\eta_1^*(A_y)}\biggl] \tag{\ref{eq:Seq_Exp_N_Bridge_SR}}
\end{align}

The third equality follows by definition of the artificial group $G=\diamond$ in Definition \ref{def:Group_N_Bridge}.

\newpage

\subsubsection{Influence functions for $\theta_g$ and $\theta_1^*$}  \label{sec:if-1*}

The form of the influence functions for these parameters depends on how the standard $A_y$ distribution is defined. Relevant to the context of this paper is the case where it is the $A_y$ distribution of a subpopulation marked with $T=1$ where $T$ is coded using a known binary deterministic function $t(G,A_y)$. This allows options such as $T=G$ (standardizing to the $A_y$ distribution of the disadvantaged group), or if age is part of $A_y$, $T=G\cdot(\text{age}>26)$ (standardizing to the $A_y$ distribution of those in the disadvantaged group above age 26). 

We first state the theorem below which gives the influence functions for this $T=t(G,A_y)$ case of interest, assuming we observe $n$ iid copies of $O=\{G,A_y,A_z,N,Z,Y\}$. Then we comment on two other cases where the influence functions are slightly different.

\begin{theorem}[$T=t(G,A_y)$ case influence functions]\label{thm:if}
In this case, the influence functions of $\theta_g$ and $\theta_1^*$ are given as:
\begin{align*}
    \varphi_{\theta_g}(O)
    &=\frac{\I(G=g)}{\P(G=g)}\frac{\P_\std(A_y)}{\P_g(A_y)}[Y-\eta_g(A_y)]+
    \\
    &~~~~\frac{T}{\P(T=1)}[\eta_g(A_y)-\theta_g],\tag{\ref{eq:if-theta.g}}
    \\
    \varphi_{\theta_1^*}(O)
    &=\frac{G}{\P(G=1)}\frac{\P_\std(A_y)}{\P_1(A_y)}\frac{\P_0(Z\mid A_z,A_y)\P_1(N\mid A_z,A_y)}{\P_1(Z,N\mid A_z,A_y)}[Y-\mu_1(Z,N,A_z,A_y)]+
    \\
    &~~~~\frac{1-G}{\P(G=0)}\frac{\P_\std(A_y)}{\P_0(A_y)}\frac{\P_1(A_z\mid A_y)}{\P_0(A_z\mid A_y)}[\zeta_1^*(Z,A_z,A_y)-\kappa_1^*(A_z,A_y)]+
    \\
    &~~~~\frac{G}{\P(G=1)}\frac{\P_\std(A_y)}{\P_1(A_y)}[\nu_1^*(N,A_z,A_y)-\eta_1^*(A_y)]+
    \\
    &~~~~\frac{T}{\P(T=1)}[\eta_1^*(A_y)-\theta_1^*],\tag{\ref{eq:if-theta.1*}}
\end{align*}
where
\begin{align*}
    \eta_g(A_y)
    &\coloneq \E_g[Y\mid A_y],
    \\
    \mu_1(Z,N,A_z,A_y) 
    &\coloneq \E_1[Y\mid Z,N,A_z,A_y],
    \\
    \zeta_1^*(Z,A_z,A_y) 
    &\coloneq \E_1^*[Y\mid Z,A_z,A_y]=\int\mu_1(Z,n,A_z,A_y)\P_1(n\mid A_z,A_y)dn,
    \\
    \nu_1^*(N,A_z,A_y) 
    &\coloneq \E_1^*[Y\mid N,A_z,A_y]=\int\mu_1(z,N,A_z,A_y)\P_0(z\mid A_z,A_y)dz,
    \\
    \kappa_1^*(A_z,A_y) 
    &\coloneq \E_1^*[Y\mid A_z,A_y]=\int\mu_1(z,n,A_z,A_y)\P_0(z\mid A_z,A_y)\P_1(n\mid A_z,A_y)dzdn,
    \\
    \eta_1^*(A_y) 
    &\coloneq \E_1^*[Y\mid A_y]=\int\mu_1(z,n,a_z,A_y)\P_0(z\mid a_z,A_y)\P_1(n,a_z\mid A_y)dzdnda_z.
\end{align*}

\end{theorem}

\bigskip

Before proving the theorem, we note two other cases. In the first case, $\P_\std(A_y)$ is a known distribution. Here we do not need a $T$ indicator to mark a population whose $A_y$ distribution is used as the standard distribution, and there is no uncertainty due to learning this distribution from data. In this case, the two influence functions remain mostly the same, except the last term involving $T$ is dropped.
In the second case, $\P_\std(A_y)$ is equated to the $A_y$ distribution of a separate population, represented by an external sample. Now let $T=1$ denote membership in the external sample and $T=0$ membership in the main sample, and expand $O$ to cover both, specifically, $O=\Big\{T,G(T=0),A_y,A_z(T=0),N(T=0),Z(T=0),Y(T=0)\Big\}$. In this case, the two influence functions look mostly unchanged, except the factors $\displaystyle\frac{\I(G=g)}{\P(G=g)}$ (for $g=0,1$) turn into $\displaystyle\frac{\I(G=g,T=0)}{\P(G=g,T=0)}$, and the $\P_g$ and $\E_g$ notation now indicates conditioning on $(G=g,T=0)$. The proofs for these two cases are simple adaptations of the proof of Theorem \ref{thm:if}, thus left out.

\bigskip

The proof of Theorem \ref{thm:if} makes repeated use of the following identity, which we state as a lemma.

\begin{lemma}\label{lm}
$$\textup{E}[A\mid B,D=d]=\textup{E}\left[\frac{\textup{I}(D=d)}{\textup{P}(D=d\mid B)}A\mid B\right]=\textup{E}\left[\frac{\textup{I}(D=d)}{\textup{P}(D=d)}\frac{\textup{P}(B)}{\textup{P}(B\mid D=d)}A\mid B\right].\label{eq:lemma}$$
\end{lemma}

The proof of this lemma is simple so we re-create it here. It has appeared previously in Appendix A of \cite{nguyen2024sensitivity}.

\begin{proof}[Proof of Lemma \ref{lm}]
    The RHS is a simple re-expression of the expectation in the middle by multiplying and dividing the integrand by $\P(B)$ and re-factoring the denominator.

    To obtain the first equal sign, we start with the expression in the middle:
    \begin{align*}
        \E\left[\frac{\I(D=d)}{\P(D=d\mid B)}A\mid B\right]
        &=\E\left\{\E\left[\frac{\I(D=d)}{\P(D=d\mid B)}A\mid B,D\right]\mid B\right\}
        \\
        &=\E\left\{\frac{\I(D=d)}{\P(D=d\mid B)}\E[A\mid B,D=d]\mid B\right\}
        \\
        &=\E\left[\frac{\I(D=d)}{\P(D=d\mid B)}\mid B\right]\E[A\mid B,D=d]
        \\
        &=\E[A\mid B,D=d]=\mathrm{LHS}.
    \end{align*}
\end{proof}

Now we are ready to derive the two influence functions in Theorem \ref{thm:if}. 

\begin{proof}[\textbf{Proof of Theorem \ref{thm:if}, part 1: the IF of }$\theta_g$]\hfill

    For this part of the proof only, let $O=\{G,A_y,Y\}$ and ignore $A_z,N,Z$. Assume we have $n$ iid copies of $O$. 

    Due to symmetry, we can derive the IF for $\theta_1$ and extrapolate to $\theta_0$.
    
    Our estimand is:
    \begin{align*}
        \theta_1
        &=\E_\text{std}\{\underbrace{\E[Y\mid A_y,G=1]}_{\eta_1(A_y)}\}
        \\
        &=\E[\eta_1(A_y)\mid T=1]
        =\E\left[\frac{T}{\P(T=1)}\eta_1(A_y)\right]
        =\frac{\E\{t(G,A_y)\E[Y\mid A_y,G=1]\}}{\E[t(G,A_y)]}.
    \end{align*}

    Consider the following factorization of the joint density of $O$:
    \begin{align*}
        \P(O)={\color{magenta}\P(G,A_y)}{\color{blue}\P(Y\mid G=1,A_y)^G}\P(Y\mid G=0,A_y)^{1-G}.
    \end{align*}
    The observed-data Hilbert space (ie the space of mean-zero finite variance 1-dimensional functions of $O$ equipped with the covariance inner product) is the direct sum of three subspaces:
    $$\mathcal{H}={\color{magenta}\mathcal{H}_1}\oplus{\color{blue}\mathcal{H}_2}\oplus\mathcal{H}_3,$$
    where 
    \begin{align*}
        {\color{magenta}\mathcal{H}_1}&={\color{magenta}\{m(G,A):\E[m(G,A)]=0\}},
        \\
        {\color{blue}\mathcal{H}_2}&={\color{blue}\{Gm(Y,A_y):\E[m(Y,A_y)\mid G=1,A_y]=0\}},
        \\
        \mathcal{H}_3&=\{(1-G)m(Y,A_y):\E[m(Y,A_y)\mid G=0,A_y]=0\}.
    \end{align*}
    Consider a parametric submodel (of the nonparametric model) of $O$ based on the factorization above:
    \begin{align*}
        f(O,\beta)={\color{magenta}f_1(G,A_y,\beta_1)}{\color{blue}f_2(A_y,Y,\beta_2)^G}f_3(A_y,Y,\beta_3)^{1-G}.
    \end{align*}
    We want to write the implied model $\theta_1(\beta)$ for $\theta_1$, and find the function $\varphi_{\theta_1}(O)\in\mathcal{H}$ whose covariance with the score of the parametric submodel is equal to the derivative of $\theta_1(\beta)$ wrt. $\beta$, evaluated at the truth (denoted $\beta^0$). This function is the IF for $\theta_1$.

    Based on the above expression of $\theta_1$, we have
    $$\theta_1(\beta)=\frac{\int\!\!\int t(g,a_y)\int y{\color{blue}f_2(y,a_y,\beta_2)}dy\,{\color{magenta}f_1(g,a_y,\beta_1)}da_y\,dg}{\int\!\!\int t(g,a_y){\color{magenta}f_1(g,a_y,\beta_1)}da_y\,dg},$$
    which involves $\beta_1$ and $\beta_2$ (but not $\beta_3$), so the IF for $\theta_1$ is the sum of a term in $\mathcal{H}_1$ and a term in $\mathcal{H}_2$. We call these two terms $\varphi_{\theta_1,1}(O)$ and $\varphi_{\theta_1,2}(O)$, respectively.
    To derive these terms, we will make use of
    \begin{align*}
        \theta_1(\beta_1)
        &=\frac{\int\!\!\int t(g,a_y)\eta_1(a_y){\color{magenta}f_1(g,a_y,\beta_1)}da_y\,d_g}{\int\!\!\int t(g,a_y){\color{magenta}f_1(g,a_y,\beta_1)}da_y\,d_g},
        \\
        \theta_1(\beta_2)
        &=\E_\std\left[\int y{\color{blue}f_2(y,a_y,\beta_2)}dy\right],
    \end{align*}
    which are obtained from $\theta_1(\beta)$ by setting $\beta_2$ and by setting $\beta_1$, respectively, to the truth. 
    
    Taking derivatives and manipulating the derivatives, we have the following.
    (Throughout, we assume that regularity conditions hold for interchanging derivatives and integrals.)
    \begin{align*}
        \frac{\partial\theta_1(\beta)}{\partial\beta_1}\Big|_{\beta=\beta^0}
        &=\frac{\partial\theta_1(\beta_1)}{\partial\beta_1}\Big|_{\beta_1=\beta_1^0}
        \\
        &=\frac{\partial}{\partial\beta_1}\frac{\int\!\!\!\int t(g,a_y)\eta_1(a_y)f_1(g,a_y,\beta_1)da_y\,d_g}{\int\!\!\!\int t(g,a_y)f_1(g,a_y,\beta_1)da_y\,d_g}\Big|_{\beta_1=\beta_1^0}
        \\
        &=\frac{\frac{\partial}{\partial\beta_1}\int\!\!\!\int t(g,a_y)\eta_1(a_y)f_1(g,a_y,\beta_1)da_y\,d_g\Big|_{\beta_1=\beta_1^0}}{\P(T=1)}-\theta_1\frac{\frac{\partial}{\partial\beta_1}\int\!\!\!\int t(g,a_y)f_1(g,a_y,\beta_1)da_y\,d_g\Big|_{\beta_1=\beta_1^0}}{\P(T=1)}
        \\
        &=\frac{\int\!\!\!\int t(g,a_y)\eta_1(a_y)\frac{\partial}{\partial\beta_1}f_1(g,a_y,\beta_1)\Big|_{\beta_1=\beta_1^0}da_y\,d_g}{\P(T=1)}-\theta_1\frac{\int\!\!\!\int t(g,a_y)\frac{\partial}{\partial\beta_1}f_1(g,a_y,\beta_1)\Big|_{\beta_1=\beta_1^0}da_y\,d_g}{\P(T=1)}
        \\
        &=\frac{\int\!\!\!\int t(g,a_y)[\eta_1(a_y)-\theta_1]\frac{\partial}{\partial\beta_1}f_1(g,a_y,\beta_1)\Big|_{\beta_1=\beta_1^0}da_y\,d_g}{\P(T=1)}
        \\
        &=\frac{\int\!\!\!\int t(g,a_y)[\eta_1(a_y)-\theta_1]S_1(g,a_y,\beta_1^0)\P(g,a_y)da_y\,d_g}{\P(T=1)}~~~~(\text{$S_1()$ is the score function w.r.t. $\beta_1$})
        \\
        &=\frac{\E\{S_1(G,A_y,\beta_1^0)T[\eta_1(A_y)-\theta_1]\}}{\P(T=1)}
        \\
        &=\E\Big\{S_1(G,A_y,\beta_1^0)\underbrace{\frac{T}{\P(T=1)}[\eta_1(A_y)-\theta_1]}_{\textstyle\varphi_{\theta_1,1}(O)\in\mathcal{H}_1}\Big\}.
    \end{align*}
    (Note that $\frac{T}{\P(T=1)}[\eta_1(A_y)-\theta_1]$ belongs in $\mathcal{H}_1$ because it is a function of $G,A_y$, and its expectation is zero because $\theta_1=\E[\eta_1(A_y)\mid T=1]$.)
    
    \begin{align*}
        \frac{\partial\theta_1(\beta)}{\partial\beta_2}\Big|_{\beta=\beta^0}
        &=\frac{\partial\theta_1(\beta_2)}{\partial\beta_2}\Big|_{\beta_2=\beta_2^0}
        \\
        &=\E_\text{std}\left[\int y\frac{\partial f_2(y,A_y,\beta_2)}{\partial\beta_2}\Big|_{\beta_2=\beta_2^0}dy\right]
        \\
        &=\E_\text{std}\left[\int y S_2(y,A_y,\beta_2^0)\P(y\mid A_y,G=1)dy\right]~~~~(\text{$S_2()$ is the score function w.r.t. $\beta_2$})
        \\
        &=\E_\text{std}\left\{\E[S_2(Y,A_y,\beta_2^0)Y\mid A_y,G=1]\right\}
        \\
        &=\E_\text{std}\Big(\E\Big\{S_2(Y,A_y,\beta_2^0)[Y-\eta_1(A_y)]\mid A_y,G=1\Big\}\Big)~~~~(\text{because}~\E[S_2()\mid A_y,G=1]=0)
        \\
        &=\E_\text{std}\left(\E\left\{S_2(Y,A_y,\beta_2^0)\frac{G}{\P(G=1)}\frac{\P(A_y)}{\P(A_y\mid G=1)}[Y-\eta_1(A_y)]\mid A_y\right\}\right)~~~~(\text{by Lemma \ref{lm}})
        \\
        &=\E\left(\frac{\P_\text{std}(A_y)}{\P(A_y)}\E\left\{S_2(Y,A_y,\beta_2^0)\frac{G}{\P(G=1)}\frac{\P(A_y)}{\P(A_y\mid G=1)}[Y-\eta_1(A_y)]\mid A_y\right\}\right)
        \\
        &=\E\bigg\{S_2(Y,A_y,\beta_2^0)\underbrace{\frac{G}{\P(G=1)}\frac{\P_\text{std}(A_y)}{\P(A_y\mid G=1)}[Y-\eta_1(A_y)]}_{\textstyle\varphi_{\theta_1,2}(O)\in\mathcal{H}_2}\bigg\}.
    \end{align*}

    Putting the two terms together,
    \begin{align*}
        \varphi_{\theta_1}(O)=
        \underbrace{\frac{G}{\P(G=1)}\frac{\P_\std(A_y)}{\P(A_y\mid G=1)}[Y-\eta_1(A_y)]}_{\textstyle\varphi_{\theta_1,2}(O)}
        +
        \underbrace{\frac{T}{\P(T=1)}[\eta_1(A_y)-\theta_1]}_{\textstyle\varphi_{\theta_1,1}(O)},
    \end{align*}
    and leveraging symmetry, we have
    \begin{align*}
        \varphi_{\theta_g}(O)=\frac{\I(G=g)}{\P(G=g)}\frac{\P_\std(A_y)}{\P(A_y\mid G=g)}[Y-\eta_g(A_y)]+\frac{T}{\P(T=1)}[\eta_g(A_y)-\theta_g].\tag{\ref{eq:if-theta.g}}
    \end{align*}
    
\end{proof}

\begin{proof}[\textbf{Proof of Theorem \ref{thm:if}, part 2: the IF of }$\theta_1^*$]\hfill
    
    We are back to considering the full observed data vector $O=\{G,A_y,A_z,N,Z,Y\}$, with $n$ iid copies of $O$.

    Our estimand is 
    \begin{align*}
        \theta_1^*
        &=\E_\text{std}\bigg[\int\!\!\!\int\!\!\!\int\E[Y\mid z,n,a_z,A_y,G=1]\P(z\mid a_z,A_y,G=0)\P(n,a_z\mid A_y,G=1)dz\,dn\,da_z\bigg]
        \\
        &=\E\bigg[\int\!\!\!\int\!\!\!\int\E[Y\mid z,n,a_z,A_y,G=1]\P(z\mid a_z,A_y,G=0)\P(n,a_z\mid A_y,G=1)dz\,dn\,da_z \mid T=1\bigg]
        \\
        &=\frac{\E\left\{t(G,A_y)\int\!\!\!\int\!\!\!\int\E[Y\mid z,n,a_z,A_y,G=1]\P(z\mid a_z,A_y,G=0)\P(n,a_z\mid A_y,G=1)dz\,dn\,da_z\right\}}{\E[t(G,A_y)]}.
    \end{align*}
    
    Consider the following factorization of the joint density of $O$:
    \begin{align*}
        \P(O)=~
        &{\color{magenta}\P(G,A_y)}\times
        \\
        &\big[{\color{teal}\P(A_z,N\mid G=1,A_y)}\P(Z\mid G=1,A_y,A_z,N){\color{blue}\P(Y\mid G=1,A_y,A_z,N,Z)}\big]^G\times
        \\
        &\big[\P(A_z\mid G=0,A_y){\color{red}\P(Z\mid G=0,A_y,A_z)}\P(N,Y\mid G=0,A_y,A_z,Z)\big]^{1-G}.
    \end{align*}
    Note the reverse order of $Z,N$ and $N,Z$ used in the factorization for the two groups, and also the lumping of $(A_z,N)$ for the $G=1$ group. This is deliberate and it is tailored to the estimand.
    
    The observed-data Hilbert space (i.e., the space of mean-zero finite-variance 1-dimensional functions of $O$ equipped with the covariance inner product) is thus the direct sum of the following subspaces:
    \begin{align*}
        \mathcal{H}={\color{magenta}\mathcal{H}_1}\oplus{\color{teal}\mathcal{H}_2}\oplus\mathcal{H}_3\oplus{\color{blue}\mathcal{H}_4}\oplus\mathcal{H}_5\oplus{\color{red}\mathcal{H}_6}\oplus\mathcal{H}_7,
    \end{align*}
    where
    \begin{align*}
        {\color{magenta}\mathcal{T}_1}&={\color{magenta}\{m(G,A_y):\E[(G,A_y)]=0\}},
        \\
        {\color{teal}\mathcal{T}_2}&={\color{teal}\{Gm(N,A_z,A_y):\E[m(N,A_z,A_y)\mid G=1,A_y]=0\}},
        \\
        \mathcal{T}_3&=\{Gm(Z,N,A_z,A_y):\E[m(Z,N,A_z,A_y)\mid G=1,A_y,A_z,N]=0\},
        \\
        {\color{blue}\mathcal{T}_4}&={\color{blue}\{Gm(Y,Z,N,A_z,A_y):\E[m(Y,Z,N,A_z,A_y)\mid G=1,A_y,A_z,N,Z]=0\}},
        \\
        \mathcal{T}_5&=\{(1-G)m(A_z,A_y):\E[m(A_z,A_y)\mid G=0,A_y]=0\},
        \\
        {\color{red}\mathcal{T}_6}&={\color{red}\{(1-G)m(Z,A_z,A_y):\E[m(N,A_z,A_y)\mid G=0,A_y,A_z]=0\}},
        \\
        \mathcal{T}_7&=\{(1-G)m(Y,Z,N,A_z,A_y):\E[m(Y,Z,N,A_z,A_y)\mid G=0,A_y,A_z,Z]=0\}.
    \end{align*}
    Consider a parametric submodel (of the nonparametric model) of $O$ based on the factorization above:
    \begin{align*}
        f(O,\beta)=~
        &{\color{magenta}f_1(G,A_y,\beta_1)}\times
        \\
        &\big[
        {\color{teal}f_2(N,A_z,A_y,\beta_2)}
        f_3(Z,N,A_z,A_y,\beta_3)
        {\color{blue}f_4(Y,Z,N,A_z,A_y,\beta_4)}
        \big]^G\times
        \\
        &\big[
        f_5(A_z,A_y,\beta_5)
        {\color{red}f_6(Z,A_z,A_y,\beta_6)}
        f_7(Y,N,Z,A_z,A_y,\beta_7)
        \big]^{1-G}.
    \end{align*}

    Based on the above expression of $\theta_1^*$, we have
    \begin{align*}
        \theta_1^*(\beta)=
        \frac{\int\!\!\!\int t(g,a_y)\left[
        \int\!\!\!\int\!\!\!\int\!\!\!\int
        y
        {\color{blue}f_4(y,z,n,a_z,a_y,\beta_4)}
        {\color{red}f_6(z,a_z,a_y,\beta_6)}
        {\color{teal}f_2(n,a_z,a_y,\beta_2)}
        dy\,dz\,dn\,da_z
        \right]{\color{magenta}f_1(g,a_y,\beta_1)}da_y,d_g}
        {\int\!\!\!\int t(g,a_y){\color{magenta}f_1(g,a_y,\beta_1)}da_y,d_g},
    \end{align*}
    which involves four elements of $\beta$. The IF for $\theta_1^*$ thus is the sum of four terms that reside in $\mathcal{H}_1$, $\mathcal{H}_2$, $\mathcal{H}_4$ and $\mathcal{H}_6$. To derive these terms, we will make use of
    \begin{align*}
        \theta_1^*(\beta_1)
        &=\frac{\int\!\!\!\int t(g,a_y)\eta_1^*(a_y){\color{magenta}f_1(g,a_y,\beta_1)}da_y\,dg}{\int\!\!\!\int t(g,a_y){\color{magenta}f_1(g,a_y,\beta_1)}da_y\,dg}
        \\
        \theta_1^*(\beta_2)
        &=\E_\std\left[\int\!\!\!\int\nu_1^*(n,a_z,A_y){\color{teal}f_2(n,a_z,A_y,\beta_2)}dn\,da_z\right]
        \\
        \theta_1^*(\beta_4)
        &=\E_\std\left\{\E^*\left[\int y{\color{blue}f_4(y,Z,N,A_z,A_y,\beta_4)}dy\mid A_y,G=1\right]\right\}
        \\
        \theta_1^*(\beta_6)
        &=\E_\std\left\{\E\left[\int\zeta_1^*(z,A_z,A_y){\color{red}f_6(z,A_z,A_y,\beta_6)}dz\mid A_y,G=1\right]\right\},
    \end{align*}
    which are obtained from $\theta_1(\beta)$ by keeping one element of $\beta$ at a time and setting the others to the truth.

    Now we manipulate derivatives to derive the four terms. (Again, we assume that regularity conditions hold for interchanging derivatives and integrals.)

    For the term in $\mathcal{H}_1$, we use similar reasoning as in the first part of the proof to obtain
    \begin{align*}
        \frac{\partial\theta_1^*(\beta)}{\partial\beta_1}\Big|_{\beta=\beta^0}=\frac{\partial\theta_1^*(\beta_1)}{\partial\beta_1}\Big|_{\beta_1=\beta_1^0}=\E\bigg\{S_1(G,A_y,\beta_1^0)\underbrace{\frac{T}{\P(T=1)}[\eta_1^*(A_y)-\theta_1^*]}_{\textstyle\varphi_{\theta_1^*,1}(O)\in\mathcal{H}_1}\bigg\}.
    \end{align*}
    The other terms are derived as follows.
    \begin{align*}
    \frac{\partial\theta_1^*(\beta)}{\partial\beta_2}\Big|_{\beta=\beta^0}
    &=\frac{\partial\theta_1^*(\beta_2)}{\partial\beta_2}\Big|_{\beta_2=\beta_2^0}
    \\
    &=\E_\text{std}\left[\int\!\!\!\int\nu_1^*(n,a_z,A_y)\frac{\partial f_2(n,a_z,A_y,\beta_2)}{\partial\beta_2}\Big|_{\beta_2=\beta_2^0}dn\,da_z\right]
    \\
    &=\E_\text{std}\left[\int\!\!\!\int\nu_1^*(n,a_z,A_y)S_2(n,a_z,A_y,\beta_2^0)\P(n,a_z\mid A_y,G=1)dn\,da_z\right]~~~~(\text{$S_2()$ is the score function w.r.t. $\beta_2$})
    \\
    &=\E_\text{std}\{\E[S_2(N,A_z,A_y,\beta_2^0)\nu_1^*(N,A_z,A_y)\mid A_y,G=1]\}
    \\
    &=\E_\text{std}\bigg(\E\bigg\{S_2(N,A_z,A_y,\beta_2^0)[\nu_1^*(N,A_z,A_y)-\eta_1^*(A_y)]\mid A_y,G=1\bigg\}\bigg)~~~~(\text{because }\E[S_2()\mid A_y,G=1]=0)
    \\
    &=\E_\text{std}\bigg(\E\bigg\{S_2(N,A_z,A_y,\beta_2^0)\frac{G}{\P(G=1)}\frac{\P(A_y)}{\P(A_y\mid G=1)}[\nu_1^*(N,A_z,A_y)-\eta_1^*(A_y)]\mid A_y\bigg\}\bigg)~~~~(\text{by Lemma \ref{lm}})
    \\
    &=\E\bigg(\frac{\P_\text{std}(A_y)}{\P(A_y)}\E\bigg\{S_2(N,A_z,A_y,\beta_2^0)\frac{G}{\P(G=1)}\frac{\P(A_y)}{\P(A_y\mid G=1)}[\nu_1^*(N,A_z,A_y)-\eta_1^*(A_y)]\mid A_y\bigg\}\bigg)
    \\
    &=\E\bigg\{S_2(N,A_z,A_y,\beta_2^0)\underbrace{\frac{G}{\P(G=1)}\frac{\P_\text{std}(A_y)}{\P(A_y\mid G=1)}[\nu_1^*(N,A_z,A_y)-\eta_1^*(A_y)]}_{\textstyle\varphi_{\theta_1^*,2}(O)\in\mathcal{H}_2}\bigg\}.
\end{align*}
\begin{align*}
    \frac{\partial\theta_1^*(\beta)}{\partial\beta_4}\Big|_{\beta=\beta^0}
    &=\frac{\partial\theta_1^*(\beta_4)}{\partial\beta_4}\Big|_{\beta_4=\beta_4^0}
    \\
    &=\E_\text{std}\left\{\E^*\left[\int y\frac{\partial f_4(y,Z,N,A_z,A_y,\beta_4)}{\partial\beta_4}\Big|_{\beta_4=\beta_4^0}dy\mid A_y,G=1\right]\right\}
    \\
    &=\E_\text{std}\left\{\E^*\left[\int y\, S_4(y,Z,N,A_z,A_y,\beta_4^0)\P(y\mid Z,N,A_z,A_y,G=1)dy\mid A_y,G=1\right]\right\}
    \\
    &~~~~~~~~~~~~~~~~~~~~~~~~~~~~~~~~~~~~~~~~~~~~~~~~~~~~~~~~~~~~~~~~~~~~~~~~~~~~~~~~~(\text{$S_4()$ is the score function w.r.t. $\beta_4$})
    \\
    &=\E_\text{std}\big(\E^*\big\{\E[S_4(Y,Z,N,A_z,A_y,\beta_4^0)Y\mid Z,N,A_z,A_y,G=1]\mid A_y,G=1\big\}\big)
    \\
    &=\E_\text{std}\bigg[\E^*\bigg(\E\bigg\{S_4(Y,Z,N,A_z,A_y,\beta_4^0)[Y-\mu_1(Z,N,A_z,A_y)]\mid Z,N,A_z,A_y,G=1\bigg\}\mid A_y,G=1\bigg)\bigg]
    \\
    &~~~~~~~~~~~~~~~~~~~~~~~~~~~~~~~~~~~~~~~~~~~~~~~~~~~~~~~~~~~~~~~~~~~~~~~~~~~~~~~~~(\text{because }\E[S_4()\mid Z,N,A_z,A_y,G=1]=0)
    \\
    &=\E_\text{std}\bigg[\E\bigg(\frac{\P(Z\mid A_z,A_y,G=0)}{\P(Z\mid N,A_z,A_y,G=1)}\E\bigg\{S_4(Y,Z,N,A_z,A_y,\beta_4^0)[Y-\mu_1(Z,N,A_z,A_y)]\mid Z,N,A_z,A_y,G=1\bigg\}
    \\
    &~~~~~~~~~~~~~~~~~~\mid A_y,G=1\bigg)\bigg]~~~(\text{density ratio weighting to swap distribution over which expectation is taken})
    \\
    &=\E_\text{std}\bigg(\E\bigg\{S_4(Y,Z,N,A_z,A_y,\beta_4^0)\frac{\P(Z\mid A_z,A_y,G=0)}{\P(Z\mid N,A_z,A_y,G=1)}[Y-\mu_1(Z,N,A_z,A_y)]\mid A_y,G=1\bigg\}\bigg)
    \\
    &=\E_\std\bigg(\E\bigg\{S_4(Y,Z,N,A_z,A_y,\beta_4^0)\frac{G}{\P(G=1)}\frac{\P(A_y)}{\P(A_y\mid G=1)}\frac{\P(Z\mid A_z,A_y,G=0)}{\P(Z\mid N,A_z,A_y,G=1)}[Y-\mu_1(Z,N,A_z,A_y)]\mid A_y\bigg\}\bigg)
    \\
    &~~~~~~~~~~~~~~~~~~~~~~~~~~~~~~~~~~~~~~~~~~~~~~~~~~~~~~~~~~~~~~~~~~~~~~~~~~~~~~~~~~~~~~~~~~~~~~~~~~(\text{by Lemma \ref{lm}})
    \\
    &=\E\bigg(\frac{\P_\std(A_y)}{\P(A_y)}\E\bigg\{S_4(Y,Z,N,A_z,A_y,\beta_4^0)\frac{G}{\P(G=1)}\frac{\P(A_y)}{\P(A_y\mid G=1)}\frac{\P(Z\mid A_z,A_y,G=0)}{\P(Z\mid N,A_z,A_y,G=1)}[Y-\mu_1(Z,N,A_z,A_y)]
    \\
    &~~~~~~~~~~~~~~~~~~\mid A_y\bigg\}\bigg)
    \\
    &=\E\bigg\{S_4(Y,Z,N,A_z,A_y,\beta_4^0)\underbrace{\frac{G}{\P(G=1)}\frac{\P_\text{std}(A_y)}{\P(A_y\mid G=1)}\frac{\P(Z\mid A_z,A_y,G=0)}{\P(Z\mid N,A_z,A_y,G=1)}[Y-\mu_1(Z,N,A_z,A_y)]}_{\textstyle\varphi_{\theta_1^*,4}(O)\in\mathcal{H}_4}\bigg\}.
\end{align*}
\begin{align*}
    \frac{\partial\theta_1^*(\beta)}{\partial\beta_6}\Big|_{\beta=\beta^0}
    &=\frac{\partial\theta_1^*(\beta_6)}{\partial\beta_6}\Big|_{\beta_6=\beta_6^0}
    \\
    &=\E_\text{std}\left\{\E\left[\int\zeta_1^*(z,A_z,A_y)\frac{\partial}{\partial\beta_6}f_6(z,A_z,A_y,\beta_6)\Big|_{\beta_6=\beta_6^0}dz\mid A_y,G=1\right]\right\}
    \\
    &=\E_\text{std}\left\{\E\left[\int\zeta_1^*(z,A_z,A_y)S_6(z,A_z,A_y,\beta_6^0)\P(z\mid A_z,A_y,G=0)dz\mid A_y,G=1\right]\right\}
    \\
    &~~~~~~~~~~~~~~~~~~~~~~~~~~~~~~~~~~~~~~~~~~~~~~~~~~~~~~~~~~~~~~~~~~~~~~~~~~~~~~~~~(\text{$S_6()$ is the score function w.r.t. $\beta_6$})
    \\
    &=\E_\text{std}\big(\E\big\{\E[S_6(Z,A_z,A_y,\beta_6^0)\zeta_1^*(Z,A_z,A_y)\mid A_z,A_y,G=0]\mid A_y,G=1\big\}\big)
    \\
    &=\E_\text{std}\bigg[\E\bigg(\E\bigg\{S_6(Z,A_z,A_y,\beta_6^0)[\zeta_1^*(Z,A_z,A_y)-\kappa_1^*(A_z,A_y)]\mid A_z,A_y,G=0\bigg\}\mid A_y,G=1\bigg)\bigg]
    \\
    &~~~~~~~~~~~~~~~~~~~~~~~~~~~~~~~~~~~~~~~~~~~~~~~~~~~~~~~~~~~~~~~~~~~~~~~~~~~~~~~~~(\text{because }\E[S_6()\mid A_z,A_y,G=0]=0)
    \\
    &=\E_\text{std}\left[\E\left(\E\left\{S_6(Z,A_z,A_y,\beta_6^0)\frac{1-G}{\P(G=0)}\frac{\P(A_z,A_y)}{\P(A_z,A_y\mid G=0)}[\zeta_1^*(Z,A_z,A_y)-\kappa_1^*(A_z,A_y)]\mid A_z,A_y\right\}\mid A_y,G=1\right)\right]
    \\
    &~~~~~~~~~~~~~~~~~~~~~~~~~~~~~~~~~~~~~~~~~~~~~~~~~~~~~~~~~~~~~~~~~~~~~~~~~~~~~~~~~~~~~~~~~~~~~~~~~~~~~(\text{by Lemma \ref{lm}})
    \\
    &=\E_\std\bigg[\E\bigg(\frac{\P(A_z\mid A_y,G=1)}{\P(A_z\mid A_y)}\E\bigg\{S_6(Z,A_z,A_y,\beta_6^0)\frac{1-G}{\P(G=0)}\frac{\P(A_z,A_y)}{\P(A_z,A_y\mid G=0)}[\zeta_1^*(Z,A_z,A_y)-\kappa_1^*(A_z,A_y)]
    \\
    &~~~~~~~~~~~~~~~~~~\mid A_z,A_y\bigg\}\mid A_y\bigg)\bigg]~~~(\text{density ratio weighting to swap distribution over which expectation is taken})
    \\
    &=\E_\std\left(\E\left\{S_6(Z,A_z,A_y,\beta_6^0)\frac{1-G}{\P(G=0)}\frac{\P(A_y)\P(A_z\mid A_y,G=1)}{\P(A_z,A_y\mid G=0)}[\zeta_1^*(Z,A_z,A_y)-\kappa_1^*(A_z,A_y)]\mid A_y\right\}\right)
    \\
    &=\E\left(\frac{\P_\std(A_y)}{\P(A_y)}\E\left\{S_6(Z,A_z,A_y,\beta_6^0)\frac{1-G}{\P(G=0)}\frac{\P(A_y)\P(A_z\mid A_y,G=1)}{\P(A_z,A_y\mid G=0)}[\zeta_1^*(Z,A_z,A_y)-\kappa_1^*(A_z,A_y)]\mid A_y\right\}\right)
    \\
    &=\E\bigg\{S_6(Z,A_z,A_y,\beta_6^0)\underbrace{\frac{1-G}{\P(G=0)}\frac{\P_\text{std}(A_y)\P(A_z\mid A_y,G=1)}{\P(A_z,A_y\mid G=0)}[\zeta_1^*(Z,A_z,A_y)-\kappa_1^*(A_z,A_y)]}_{\textstyle\varphi_{\theta_1^*,6}(O)\in\mathcal{H}_6}\bigg\}.
\end{align*}

Putting the terms together, we have the IF for $\theta_1^*$:
\begin{align*}
    \varphi_{\theta_1^*}(O)
    &=\frac{G}{\P(G=1)}\frac{\P_\std(A_y)}{\P(A_y\mid G=1)}\frac{\P(Z\mid A_z,A_y,G=0)}{\P(Z\mid N,A_z,A_y,G=1)}[Y-\mu_1(Z,N,A_z,A_y)]+ && (\varphi_{\theta_1^*,4}(O))
    \\
    &~~~~\frac{1-G}{\P(G=0)}\frac{\P_\text{std}(A_y)\P(A_z\mid A_y,G=1)}{\P(A_z,A_y\mid G=0)}[\zeta_1^*(Z,A_z,A_y)-\kappa_1^*(A_z,A_y)]+ && (\varphi_{\theta_1^*,6}(O))
    \\
    &~~~~\frac{G}{\P(G=1)}\frac{\P_\std(A_y)}{\P(A_y\mid G=1)}[\nu_1^*(N,A_z,A_y)-\eta_1^*(A_y)]+ && (\varphi_{\theta_1^*,2}(O))
    \\
    &~~~~\frac{T}{\P(T=1)}[\eta_1^*(A_y)-\theta_1^*]. && (\varphi_{\theta_1^*,1}(O))\tag{\ref{eq:if-theta.1*}}
\end{align*}

\end{proof}

\subsubsection{Proof for Robustness of Sequential Weighted Regression (SWR) estimators}

\textbf{General reasoning}

As this section proves robustness properties of all the sequential weighted regression (SWR) estimators, we provide a general reasoning used for all estimators. It involves three steps:
\begin{enumerate}[topsep=0em]
    \item Start with the set of estimating equations that the nuisance estimators and the estimator of $\theta_1^*$ based on them (referred to generically as $\hat\theta_1^*$) solve.
    
    For regression models (models of outcome means), the estimating equation is typically vector-valued. We require that all these models are mean-recovering, and will use only the element of the estimating equation that reflects this mean-recovering feature.
    
    \item Assume regularlity conditions hold such that the nuisance estimators and $\hat\theta_1^*$ converge to certain probability limits (indicated with a $^\dagger$ superscript), and the estimating equations imply a set of equalities involving the probability limits.

    For details on the regularity conditions, see \cite{boos2013EssentialStatisticalInference}, theorem 7.1. We assume these regularity conditions hold throughout, and will keep this implicit below for conciseness.
    
    \item Consider different cases where certain regression models are correctly specified and/or certain weighting functions are consistently estimated (replacing their probability limits with the corresponding true functionals) and show that in that case the probability limit $\theta_1^{*\dagger}$ coincide with the true value ($\theta_1^*$).

    When we say a regression model is correctly specified here, we mean it in a local sense, that the model is correctly specified for the conditional expectation of the probability limit of the dependent variable. For example, regarding a step regressing $\hat\kappa_1^*(A_z,A_y)$ on $A_y,G=1$ to estimate $\eta_1^*(A_y)$, by correct specification we mean that the model is correct for $\E[\kappa_1^{*\dagger}(A_z,A_y)\mid A_y,G=1]$, and do not mean that the model is correct for $\eta_1^*(A_y)$. Of course the latter is true if $\kappa_1^{*\dagger}()=\kappa_1^*()$, ie $\kappa_1^*()$ is consistently estimated.
\end{enumerate}

\textbf{Two common pieces}

Before addressing each of the SWR estimators, we note that two pieces of reasoning are used with all of them to obtain the same equalities. To avoid repetition, we bring them upfront.
\begin{itemize}[topsep=0em]
    \item One piece concerns the common last step of all the SWR estimators: averaging $\hat\eta_1^*(A_y)$ over the standard distribution of $A_y$ to estimate $\theta_1^*$. 
        Here $\hat\theta_1^*$ solves the estimating equation
        $$\P_n\{T[\hat\eta_1^*(A_y)-{\color{red}\theta_1^*}]=0,$$
        so we have the following equality involving the probability limits of $\hat\eta_1^*()$ and $\hat\theta_1^*$:
    \begin{align}
        \theta_1^{*\dagger}=\E[\eta_1^{*\dagger}\mid T=1]=\E_\text{std}[\eta_1^{*\dagger}(A_y)].\label{ee:0}
    \end{align}

    \item The second piece concerns the first step of all the SWR estimators: estimating $\mu_1()$ by regressing $Y$ on $Z,N,A_z,A_y$ in the $G=1$ sample weighted by $\hat\omega_{11^*}^{(Z,N)}()$. Here $\hat\mu_1()$ solves a vector-valued estimating equation, the mean-recovering element of which is
    $$\P_n\Big\{G\,\hat\omega_{11^*}^{(Z,N)}(Z,N,A_z,A_y)\hat\omega_{1T}^{A_y}(A_y)[Y-{\color{red}\mu_1(Z,N,A_z,A_y)}]\Big\}=0,$$
    which implies the following equality involving the probability limits:
    $$\E\Big\{\omega_{11^*}^{(Z,N)\dagger}(Z,N,A_z,A_y)\omega_{1T}^{A_y\dagger}(A_y)[Y-\mu_1^\dagger(Z,N,A_z,A_y)]\mid G=1\Big\}=0,$$
    which can be re-expressed as
    \begin{align}
        \E\Big\{\omega_{11^*}^{(Z,N)\dagger}(Z,N,A_z,A_y)\omega_{1T}^{A_y\dagger}(A_y)[\mu_1(Z,N,A_z,A_y)-\mu_1^\dagger(Z,N,A_z,A_y)]\mid G=1\Big\}=0.\label{ee:1}
    \end{align}
    Consider two cases: if the outcome regression model is correctly specified then (\ref{ee:1}) implies
    \begin{align}
        \mu_1^\dagger(Z,N,A_z,A_y)=\mu_1(Z,N,A_z,A_y);\tag{\ref{ee:1}a}
    \end{align}
    and if the weighting functions are consistently estimated then (\ref{ee:1}) implies
    \begin{align}
        \theta_1^*=\E_\text{std}\Big\{\E^*[\mu_1^\dagger(Z,N,A_z,A_y)\mid A_y,G=1]\Big\}.\tag{\ref{ee:1}b}
    \end{align}
\end{itemize}

Now we turn to the individual estimators.

\textbf{Z-Model-SWR}
\label{ssec:proof_robust_Z_Model_SWR}

This estimator involves estimating $\mu_1()$, then $\eta_1^*()$, then $\theta_1^*$.

The $\mu_1()$ and $\theta_1^*$ steps give us the equalities (\ref{ee:1}) and (\ref{ee:0}).

In the $\eta_1^*()$ step, $\hat\eta_1^*()$ solves
\begin{align*}
    \P_n\Big\{G\,\hat\omega_{1T}^{A_y}(A_y)\Big[\hat\mu_1(\overset{\blacktriangle}{Z},N,A_z,A_y)-{\color{red}\eta_1^*(A_y)}\Big]\Big\}&=0,
\end{align*}
where $(N,A_z,A_y)$ are data from the $G=1$ sample (hence the appearance of $G$ in the equation), and $\overset{\blacktriangle}{Z}$ is the $Z$ value simulated to create the artificial sample $G=\blacktriangle$, and this simulation is from $\hat\lambda_0^Z()$.

Here we \textbf{assume consistent estimation of the density $\lambda_0^Z()$}, ie $\lambda_0^{Z\dagger}()=\lambda_0^Z()$. This implies the equality
\begin{align*}
    \E\Big\{\omega_{1T}^{A_y\dagger}(A_y)\Big[\mu_1^\dagger(\overset{\blacktriangle}{Z},N,A_z,A_y)-\eta_1^{*\dagger}(A_y)\Big]\mid G=1\Big\}&=0
\end{align*}
where the conditional density of $\overset{\blacktriangle}{Z}$ given $(N,A_z,A_y)$ is $\lambda_0^Z()$. This implies%
\footnote{This is obtained by iterated expectation first conditioning on $N,A_z,A_y,G=1$ then on $A_y,G=1$, then on $G=1$, and noting that
\begin{align*}
    \E\{\E[\mu_1^\dagger(\overset{\blacktriangle}{Z},N,A_z,A_y)\mid N,A_z,A_y,G=1]\mid A_y,G=1\}
    &=\E\left\{\int\mu_1^\dagger(z,N,A_z,A_y)\P(z\mid A_z,A_y,G=0)dz\mid A_y,G=1\right\}
    \\
    &=\E\{\E^*[\mu_1^\dagger(Z,N,A_z,A_y)\mid N,A_z,A_y,G=1]\mid A_y,G=1\}
    \\
    &=\E^*[\mu_1^\dagger(Z,N,A_z,A_y)\mid A_y,G=1].
\end{align*}%
}
\begin{align}
    \E\Big(\omega_{1T}^{A_y\dagger}(A_y)\Big\{\E^*[\mu_1^\dagger(Z,N,A_z,A_y)\mid A_y,G=1]-\eta_1^{*\dagger}(A_y)\Big\}\mid G=1\Big)&=0.\label{ee:2-Zsim}
\end{align}

Now we rely on the combination of (\ref{ee:1}), (\ref{ee:2-Zsim}) and (\ref{ee:0}). For each of the two weighted regression steps, we consider two cases: \textbf{either the regression is correctly specified
or the weighting function is consistently estimated}. These lead to different implications of (\ref{ee:1}) and (\ref{ee:2-Zsim}), marked a and b below. We can mix and match these, and combine them with (\ref{ee:0}) to obtain $\theta_1^{*\dagger}=\theta_1^*$.

\resizebox{\textwidth}{!}{%
\begin{tabular}{cc|c}
    \textbf{a) outcome regression correctly specified}
    &
    \textbf{b) weighting function consistent}
    & weight. fun.
    \\\hline
    \parbox{8.5cm}{\begin{align}
        \mu_1^\dagger(Z,N,A_z,A_y)=\mu_1(Z,N,A_z,A_y)\tag{\ref{ee:1}a}
    \end{align}}
    & \parbox{10.5cm}{\begin{align}
        \theta_1^*=\E_\text{std}\Big\{\E^*[\mu_1^\dagger(Z,N,A_z,A_y)\mid A_y,G=1]\Big\}\tag{\ref{ee:1}b}
    \end{align}}
    & $\omega_{11^*}^{ZN}()\omega_{1T}^{A_y}()$
    \\\hline
    \parbox{8.5cm}{\begin{align}
        \eta_1^{*\dagger}(A_y)=\E^*[\mu_1^\dagger(Z,N,A_z,A_y)\mid A_y,G=1]\tag{\ref{ee:2-Zsim}a}
    \end{align}} 
    & \parbox{10.5cm}{\begin{align}
        \E_\text{std}\Big\{\E^*[\mu_1^\dagger(Z,N,A_z,A_y)\mid A_y,G=1]-\eta_1^{*\dagger}(A_y)\Big\}=0\tag{\ref{ee:2-Zsim}b}
    \end{align}} 
    & $\omega_{1T}^{A_y}()$
    \\\hline
    \multicolumn{2}{c}{%
    \parbox{5cm}{\begin{align}
        \theta_1^{*\dagger}=\E_\text{std}[\eta_1^{*\dagger}(A_y)]\tag{\ref{ee:0}}
    \end{align}}}
    \\\hline
\end{tabular}%
}

\textbf{N-Model-SWR}

This estimator involves estimating $\mu_1()$, then $\kappa_1^*()$, then $\eta_1^*()$, then $\theta_1^*$.  

The $\mu_1()$ and $\theta_1^*$ give us equalities (\ref{ee:1}) and (\ref{ee:0}). 

In the $\eta_1^*()$ step, $\hat\eta_1^*()$ solves
\begin{align*}
    \P_n\Big\{G\,\hat\omega_{1T}^{A_y}(A_y)\Big[\hat\kappa_1^*(A_z,A_y)-{\color{red}\eta_1^*(A_y)}\Big]\Big\}&=0,
\end{align*}
which gives us the equality
\begin{align}
    \E\Big(\omega_{1T}^{A_y\dagger}(A_y)\Big\{\E[\kappa_1^{*\dagger}(A_z,A_y)\mid A_y,G=1]-\eta_1^{*\dagger}(A_y)\Big\}\mid G=1\Big)=0.\label{ee:3-Zcsim.hwt}
\end{align}

In the $\kappa_1^*()$ step, $\hat\kappa_1^*()$ solves
$$\P_n\Big\{(1-G)\hat\omega_{1T}^{A_y}(A_y)\hat\omega_{01}^{A_z,A_y}(A_z,A_y)\Big[\hat\mu_1(Z,\overset{\triangle}{N},A_z,A_y)-{\color{red}\kappa_1^*(A_z,A_y)}\Big]\Big\}=0,$$
where $(Z,A_z,A_y)$ are data from the $G=0$ sample (hence the $(1-G)$ in the equation), and $\overset{\triangle}{N}$ is the $N$ value simulated to create the artificial sample $G=\triangle$, and this simulation is from $\hat\lambda_1^N()$.

Here we \textbf{assume consistent estimation of the density $\lambda_1^N()$}, ie $\lambda_1^{N\dagger}()=\lambda_1^N()$. This implies the equality
$$\E\Big\{\omega_{0T}^{A_y\dagger}(A_y)\omega_{01}^{A_z\dagger}(A_z,A_y)\Big[\mu_1^\dagger(Z,\overset{\triangle}{N},A_z,A_y)-\kappa_1^{*\dagger}(A_z,A_y)\Big]\mid G=0\Big\}=0,$$
where the conditional density of $\overset{\triangle}{N}$ given $(Z,A_z,A_y)$ is $\lambda_1^N()$. 
This can be re-expressed%
\footnote{This is obtained by iterated expectation first conditioning on $A_z,A_y,G=0$ and then on $G=0$, and noting that
\begin{align*}
    \E_\triangle[\mu_1^\dagger(Z,\overset{\triangle}{N},A_z,A_y)\mid A_z,A_y,G=0]
    &=\int\!\!\!\int\mu_1^\dagger(z,n,A_z,A_y)\P_\triangle(n\mid z,A_z,A_y)\P(z\mid A_z,A_y,G=0)dn\,dz
    \\
    &=\E^*[\mu_1^\dagger(Z,N,A_z,A_y)\mid A_z,A_y,G=0].
\end{align*}}
as 
\begin{align}
    \E\Big(\omega_{0T}^{A_y\dagger}(A_y)\omega_{01}^{A_z\dagger}(A_z,A_y)\Big\{\E^*[\mu_1^\dagger(Z,N,A_z,A_y)\mid A_z,A_y,G=1]-\kappa_1^{*\dagger}(A_z,A_y)\Big\}\mid G=0\Big)=0,\label{ee:2-Nwt}
\end{align}

Thus we have the combination of (\ref{ee:1}), (\ref{ee:2-Nwt}), (\ref{ee:3-Zcsim.hwt}) and (\ref{ee:0}).

For each of the three weighted regression steps, we consider two cases: \textbf{either the regression is correctly specified or the weighting function is consistently estimated}. These lead to different implications of (\ref{ee:1}), (\ref{ee:2-Nwt}), (\ref{ee:3-Zcsim.hwt}), marked a and b below. We can mix and match these, and combine them with (\ref{ee:0}) to obtain $\theta_1^{*\dagger}=\theta_1^*$. 

\resizebox{\textwidth}{!}{%
\begin{tabular}{@{}cc|c}
    \textbf{a) outcome regression correctly specified}
    & \textbf{b) weighting function consistent}
    & weight. fun.
    \\\hline
    \parbox{10cm}{\begin{align}
        \mu_1^{*\dagger}()=\mu_1^*()\tag{\ref{ee:1}a}
    \end{align}}
    & \parbox{11.5cm}{\begin{align}
        \theta_1^*=\E_\text{std}\Big\{\E^*[\mu_1^\dagger(Z,N,A_z,A_y)\mid A_y,G=1\Big\}\tag{\ref{ee:1}b}
    \end{align}}
    & $\omega_{11^*}^{(Z,N)}()\omega_{1T}^{A_y}()$
    \\\hline
    \parbox{10cm}{\begin{align}
        \kappa_1^{*\dagger}(A_z,A_y)=\E^*[\mu_1^{*\dagger}(Z,N,A_z,A_y)\mid A_z,A_y,G=1]\tag{\ref{ee:2-Nwt}a}
    \end{align}}
    & \parbox{11.5cm}{\begin{align}
        \E_\text{std}\Big\{\E^*\Big[\mu_1^\dagger(Z,N,A_z,A_y)-\kappa_1^{*\dagger}(A_z,A_y)\mid A_y,G=1\Big]\Big\}=0\tag{\ref{ee:2-Nwt}b}
    \end{align}}
    & $\omega_{01}^{A_z}()\omega_{0T}^{A_y}()$
    \\\hline
    \parbox{10cm}{\begin{align}
        \eta_1^{*\dagger}(A_y)=\E[\kappa_1^{*\dagger}(A_z,A_y)\mid A_y,G=1]\tag{\ref{ee:3-Zcsim.hwt}a}
    \end{align}}
    & \parbox{11.5cm}{\begin{align}
        \E_\text{std}\Big\{\E[\kappa_1^{*\dagger}(A_z,A_y)\mid A_y,G=1]-\eta_1^{*\dagger}(A_y)\Big\}=0\tag{\ref{ee:3-Zcsim.hwt}b}
    \end{align}}
    & $\omega_{1T}^{A_y}()$
    \\\hline
    \multicolumn{2}{c}{%
    \parbox{5cm}{\begin{align}
        \theta_1^{*\dagger}=\E_\text{std}[\eta_1^{*\dagger}(A_y)]\tag{\ref{ee:0}}
    \end{align}}}
    \\\hline
\end{tabular}%
}

\textbf{Z-Bridge-SWR}

This estimator involves estimating $\mu_1()$, then $\zeta_1^*()$, then $\kappa_1^*()$, then $\eta_1^*()$, then $\theta_1^*$.

The $\mu_1()$, $\eta_1^*()$ and $\theta_1^*$ steps give us the equalities (\ref{ee:1}), (\ref{ee:3-Zcsim.hwt}),  and (\ref{ee:0}).

In the $\kappa_1^*()$ step, $\hat\kappa_1^*()$ solves
$$\P_n\Big\{(1-G)\hat\omega_{0T}^{A_y}(A_y)\hat\omega_{01}^{A_z}(A_z,A_y)[\hat\zeta_1^*(Z,A_z,A_y)-{\color{red}\kappa_1^*(A_z,A_y)}]\Big\}=0,$$
which gives us the equality
$$\E\Big\{\omega_{0T}^{A_y\dagger}(A_y)\omega_{01}^{A_z\dagger}(A_z,A_y)[\zeta_1^{*\dagger}(Z,A_z,A_y)-\kappa_1^{*\dagger}(A_z,A_y)]\mid G=0\Big\}=0,$$
which can be re-expressed%
\footnote{This is obtained by iterated expectation, first conditioning on $A_z,A_y,G=0$, then on $G=0$ and noting that
\begin{align*}
    \E[\zeta_1^{*\dagger}(Z,A_z,A_y)\mid A_z,A_y,G=0]
    &=\int\zeta_1^{*\dagger}(z,A_z,A_y)\P(z\mid A_z,A_y,G=0)dz
    \\
    &=\int\zeta_1^{*\dagger}(z,A_z,A_y)\P^*(z\mid A_z,A_y,G=1)dz
    \\
    &=\E^*[\zeta_1^{*\dagger}(Z,A_z,A_y)\mid A_z,A_y,G=1].
\end{align*}
}
as
\begin{align}
    \E\Big(\omega_{0T}^{A_y\dagger}(A_y)\omega_{01}^{A_z\dagger}(A_z,A_y)\Big\{\E^*[\zeta_1^{*\dagger}(Z,A_z,A_y)\mid A_z,A_y,G=1]-\kappa_1^{*\dagger}(A_z,A_y)\Big\}\mid G=0\Big)=0.\label{ee:3-Zcsim}
\end{align}

In the $\zeta_1^*()$ step, $\hat\zeta_1^*()$ solves
$$\P_n\Big\{G\hat\omega_{1T}^{A_y}(A_y)\hat\omega_{\blackdiamond1^*}^Z(Z,A_z,A_y)(\overset{\blackdiamond}{Z},A_z,A_y)[\hat\mu_1(\overset{\blackdiamond}{Z},N,A_z,A_y)-{\color{red}\zeta_1^*(\overset{\blackdiamond}{Z},A_z,A_y)}]\Big\}=0,$$
where $(N,A_z,A_y)$ are data from the $G=1$ sample (hence $G$ in the equation) and $\overset{\blackdiamond}{Z}$ is the $Z$ value simulated for the artificial $G=\blackdiamond$ group.
This gives us the equality
$$\E\Big\{\omega_{1T}^{A_y\dagger}(A_y)\omega_{\blackdiamond1^*}^{Z\dagger}(\overset{\blackdiamond}{Z},A_z,A_y)[\mu_1^\dagger(\overset{\blackdiamond}{Z},N,A_z,A_y)-\zeta_1^{*\dagger}(\overset{\blackdiamond}{Z},A_z,A_y)]\mid G=1\Big\}=0,$$
which can be re-expressed%
\footnote{This is obtained by iterated expectation first conditioning on $\overset{\blackdiamond}{Z},A_z,A_y,G=1$ and then on $G=1$.}
as
\begin{align}
    \E\Big(\omega_{1T}^{A_y\dagger}(A_y)\omega_{\blackdiamond1^*}^{Z\dagger}(\overset{\blackdiamond}{Z},A_z,A_y)\Big\{\E_\blackdiamond[\mu_1^\dagger(\overset{\blackdiamond}{Z},N,A_z,A_y)\mid \overset{\blackdiamond}{Z},A_z,A_y,G=1]-\zeta_1^{*\dagger}(\overset{\blackdiamond}{Z},A_z,A_y)\Big\}\mid G=1\Big)=0,\label{ee:2-Zcsim}
\end{align}
where note that 
$$\E_\blackdiamond[m(N,z,A_z,A_y)\mid z,A_z,A_y,G=1]=\E^*[m(N,z,A_z,A_y)\mid z,A_z,A_y,G=1].$$

Now we work with the combination of (\ref{ee:1}), (\ref{ee:2-Zcsim}), (\ref{ee:3-Zcsim}), (\ref{ee:3-Zcsim.hwt}) and (\ref{ee:0}). For each of the four weighted regression steps, we consider two cases: \textbf{either the regression is correctly specified or the weighting function is consistantly estimated}. This leads to different implications of (\ref{ee:1}), (\ref{ee:2-Zcsim}), (\ref{ee:3-Zcsim}) and (\ref{ee:3-Zcsim.hwt}), marked as a and b below. We can mix and match these, and combine them with (\ref{ee:0}) to obtain $\theta_1^{*\dagger}=\theta_1^*$.

\resizebox{\textwidth}{!}{%
\begin{tabular}{@{}cc|c}
    \textbf{a) outcome regression correctly specified}
    & \textbf{b) weighting function consistent}
    & weight. fun.
    \\\hline
    \parbox{10cm}{\begin{align}
        \mu_1^{*\dagger}(Z,N,A_z,A_y)=\mu_1^*(Z,N,A_z,A_y)\tag{\ref{ee:1}a}
    \end{align}}
    & \parbox{11.5cm}{\begin{align}
        \theta_1^*=\E_\text{std}\Big\{\E^*[\mu_1^\dagger(Z,N,A_z,A_y)\mid A_y,G=1\Big\}\tag{\ref{ee:1}b}
    \end{align}}
    & $\omega_{11^*}^{(Z,N)}()\omega_{1T}^{A_y}()$
    \\\hline
    \parbox{10cm}{\begin{align}
        \zeta_1^{*\dagger}(Z,A_z,A_y)=\E^*[\mu_1^{*\dagger}(Z,N,A_z,A_y)\mid Z,A_z,A_y,G=1]\tag{\ref{ee:2-Zcsim}a}
    \end{align}}
    & \parbox{11.5cm}{\begin{align}
        \E_\text{std}\Big\{\E^*\Big[\mu_1^\dagger(Z,N,A_z,A_y)-\zeta_1^{*\dagger}(Z,A_z,A_y)\mid A_y,G=1\Big]\Big\}=0\tag{\ref{ee:2-Zcsim}b}
    \end{align}}
    & $\omega_{\blackdiamond1^*}^Z()\omega_{1T}^{A_y}()$
    \\\hline
    \parbox{10cm}{\begin{align}
        \kappa_1^{*\dagger}(A_z,A_y)=\E^*[\zeta_1^{*\dagger}(Z,A_z,A_y)\mid A_z,A_y,G=1]\tag{\ref{ee:3-Zcsim}a}
    \end{align}}
    & \parbox{11.5cm}{\begin{align}
        \E_\text{std}\Big\{\E^*\Big[\zeta_1^{*\dagger}(Z,A_z,A_y)-\kappa_1^{*\dagger}(A_z,A_y)\mid A_y,G=1\Big]\Big\}=0\tag{\ref{ee:3-Zcsim}b}
    \end{align}}
    & $\omega_{01}^{A_z}()\omega_{0T}^{A_y}()$
    \\\hline
    \parbox{10cm}{\begin{align}
        \eta_1^{*\dagger}(A_y)=\E[\kappa_1^{*\dagger}(A_z,A_y)\mid A_y,G=1]\tag{\ref{ee:3-Zcsim.hwt}a}
    \end{align}}
    & \parbox{11.5cm}{\begin{align}
        \E_\text{std}\Big\{\E[\kappa_1^{*\dagger}(A_z,A_y)\mid A_y,G=1]-\eta_1^{*\dagger}(A_y)\Big\}=0\tag{\ref{ee:3-Zcsim.hwt}b}
    \end{align}}
    & $\omega_{1T}^{A_y}()$
    \\\hline
    \multicolumn{2}{c}{%
    \parbox{5cm}{\begin{align}
        \theta_1^{*\dagger}=\E_\text{std}[\eta_1^{*\dagger}(A_y)]\tag{\ref{ee:0}}
    \end{align}}}
    \\\hline
\end{tabular}%
}

\textbf{N-Bridge-SWR}

This estimator involves estimating $\mu_1()$, then $\nu_1^*()$, then $\eta_1^*()$, then $\theta_1^*$.

The $\mu_1()$ and $\theta_1^*$ steps give use the equalities (\ref{ee:1}) and (\ref{ee:0}).

In the $\eta_1^*()$ step, $\hat\eta_1^*()$ solves
$$\P_n\Big\{G\,\hat\omega_{1T}^{A_y}(A_y)\left[\hat\nu_1^*(N,A_z,A_y)-{\color{red}\eta_1^*(A_y)}\right]\Big\}=0,$$
which implies the equality
\begin{align}
    \E\Big\{\omega_{1T}^{A_y\dagger}(A_y)[\nu_1^{*\dagger}(N,A_z,A_y)-\eta_1^{*\dagger}(A_y)]\mid G=1\Big\}=0.\label{ee:3-Ncsim}
\end{align}

In the $\nu_1^*()$ step, $\hat\nu_1^*()$ solves
$$\P_n\Big\{(1-G)\hat\omega_{1T}^{A_y}(A_y)\hat\omega_{\diamond1}^{(Z,N,A_z)}(\overset{\diamond}{N},A_z,A_y)\big[\hat\mu_1(Z,\overset{\diamond}{N},A_z,A_y)-{\color{red}\nu_1^*(\overset{\diamond}{N},A_z,A_y)}\big]\Big\}=0,$$
where $(Z,A_z,A_y)$ are data from the $G=0$ sample (hence $(1-G)$ in the equation), and $\overset{\diamond}{N}$ is the $N$ value simulated for the artificial $G=\diamond$ group.
This implies the equality:
$$\E\Big\{\omega_{1T}^{A_y\dagger}(A_y)\omega_{\diamond1}^{(Z,N,A_z)\dagger}(\overset{\diamond}{N},A_z,A_y)[\mu_1^\dagger(Z,\overset{\diamond}{N},A_z,A_y)-\nu_1^{*\dagger}(\overset{\diamond}{N},A_z,A_y)]\mid G=0\Big\}=0,$$
which can be re-expressed%
\footnote{This is obtained by iterated expectation first conditioning on $\overset{\diamond}{N},A_z,A_y,G=0$, then on $G=0$.}
as
\begin{align}
    \E\Big(\omega_{1T}^{A_y\dagger}(A_y)\omega_{\diamond1}^{(Z,N,A_z)\dagger}(\overset{\diamond}{N},A_z,A_y)\Big\{\int\mu_1^\dagger(z,\overset{\diamond}{N},A_z,A_y)\P(z\mid A_z,A_y,G=0)dz-\nu_1^{*\dagger}(\overset{\diamond}{N},A_z,A_y)\Big\}\mid G=0\Big)=0.\label{ee:2-Ncsim}
\end{align}
Now we have the combination of (\ref{ee:1}), (\ref{ee:2-Ncsim}), (\ref{ee:3-Ncsim}) and (\ref{ee:0}). For each of the three weighted regression steps, we consider two cases: \textbf{either the regression is correctly specified or the weighting function is consistently estimated}. These lead to different implications of (\ref{ee:1}), (\ref{ee:2-Ncsim}) and (\ref{ee:3-Ncsim}), marked as a and b below. We can mix and match these, and combine them with (\ref{ee:0}) to obtain $\theta_1^{*\dagger}=\theta_1^*$.

\resizebox{\textwidth}{!}{%
\begin{tabular}{@{}cc|c}
    \textbf{a) outcome regression correctly specified}
    & \textbf{b) weighting function consistent}
    & weight. fun.
    \\\hline
    \parbox{10cm}{\begin{align}
        \mu_1^{*\dagger}(Z,N,A_z,A_y)=\mu_1^*(Z,N,A_z,A_y)\tag{\ref{ee:1}a}
    \end{align}}
    & \parbox{11cm}{\begin{align}
        \theta_1^*=\E_\text{std}\Big\{\E^*[\mu_1^\dagger(Z,N,A_z,A_y)\mid A_y,G=1\Big\}\tag{\ref{ee:1}b}
    \end{align}}
    & $\omega_{11^*}^{(Z,N)}()\omega_{1T}^{A_y}()$
    \\\hline
    \parbox{10cm}{\begin{align}
        \nu_1^{*\dagger}(N,A_z,A_y)=\E^*[\mu_1^{*\dagger}(Z,N,A_z,A_y)\mid N,A_z,A_y,G=1]\tag{\ref{ee:2-Ncsim}a}
    \end{align}}
    & \parbox{11cm}{\begin{align}
        \E_\text{std}\Big\{\E^*\Big[\mu_1^\dagger(Z,N,A_z,A_y)-\nu_1^{*\dagger}(N,A_z,A_y)\mid A_y,G=1\Big]\Big\}=0\tag{\ref{ee:2-Ncsim}b}
    \end{align}}
    & $\omega_{\diamond1^*}^{(N,A_z)}()\omega_{1T}^{A_y}()$
    \\\hline
    \parbox{10cm}{\begin{align}
        \eta_1^{*\dagger}(A_y)=\E[\nu_1^{*\dagger}(N,A_z,A_y)\mid A_y,G=1]\tag{\ref{ee:3-Ncsim}a}
    \end{align}}
    & \parbox{11cm}{\begin{align}
        \E_\text{std}\Big\{\E[\nu_1^{*\dagger}(N,A_z,A_y)\mid A_y,G=1]-\eta_1^{*\dagger}(A_y)\Big\}=0\tag{\ref{ee:3-Ncsim}b}
    \end{align}}
    & $\omega_{1T}^{A_y}()$
    \\\hline
    \multicolumn{2}{c}{%
    \parbox{5cm}{\begin{align}
        \theta_1^{*\dagger}=\E_\text{std}[\eta_1^{*\dagger}(A_y)]\tag{\ref{ee:0}}
    \end{align}}}
    \\\hline
\end{tabular}%
}

\end{document}